\documentclass{article}
\usepackage[utf8]{inputenc}
\usepackage{amsthm}
\usepackage{amsmath}
\usepackage{amsfonts}
\usepackage{amssymb}
\usepackage{mathtools}
\usepackage{commath}

\usepackage{hyperref}
\usepackage{graphicx}
\graphicspath{{./Figs/}}
\usepackage{placeins}
\usepackage{subcaption}
\usepackage{algorithmic}
\usepackage{algorithm}
\usepackage{float}
\usepackage{enumerate}
\usepackage{xcolor}
\usepackage{amsopn}
\usepackage{epstopdf}

\usepackage[left=2cm,right=2cm,top=2cm,bottom=2cm]{geometry}

\newcommand{\so}[1]{{\text{SO}\negmedspace\left(#1\right)}}

\newcommand{\sfun}{\mu_{\mathrm{s}}}
\newcommand{\pset}{\mathcal{P}}
\newcommand{\emax}{\mathrm{E}_{\max}}

\usepackage{xcolor}

\newcommand{\rev}[1]{\textcolor{black}{#1}}

\newcommand{\sout}[1]{\ignorespaces}

\newtheorem{theorem}{Theorem}

\newtheorem{lemma}[theorem]{Lemma}
\newtheorem{corollary}[theorem]{Corollary}

\newtheorem{definition}[theorem]{Definition}

\begin{document}

\title{Centering noisy images with application to cryo-EM}
\author{Ayelet Heimowitz, Nir Sharon, and Amit Singer}
\date{}

\maketitle

\begin{abstract}
  We target the problem of estimating the center of mass of noisy 2-D images.  We assume that the noise dominates the image, and thus many standard approaches are vulnerable to estimation errors. Our approach uses a surrogate function to the geometric median, which is a robust estimator of the center of mass. We mathematically analyze cases in which the geometric median fails to provide a reasonable estimate of the center of mass, and prove that our surrogate function leads to a successful estimate.
  
  One particular application for our method is to improve 3-D reconstruction in single-particle cryo-electron microscopy (cryo-EM). We show how to apply our approach for a better translational alignment of macromolecules picked from experimental data. In this way, we facilitate the succeeding steps of reconstruction and streamline the entire cryo-EM pipeline, saving valuable computational time and supporting resolution enhancement.
\end{abstract}

\section{Introduction}

The center of mass, also known as the centroid for objects with uniform densities, is the point around which all the mass of a system (or object)  is concentrated. Formally, the center of mass is defined as the arithmetic mean of all the points in the system (object) weighted by their local densities. Alternatively, the center of mass can be defined as the point for which the sum of squared distances from all other points in the system (object) is minimized. The correct identification of this point is crucial to many applications for several reasons.

One reason for the importance of the center of mass is that we can consider the sum of external forces that act on a system (object) as working on an object of the same mass located at the center of mass. Therefore, this point is essential when describing the motion of a system (object), see, e.g.,~\cite{feynman2011feynman}. In astrophysics, for example, the correct identification of the center of mass of a cloud of stars (orbiting each other) allows to represent the motion caused by external forces (this is true for binary stars, galaxies, etc.), see, e.g.,~\cite{feynman2018feynman, thomas2006comparison}.  

Another reason for the importance of the center of mass is its role as a reference point. Specifically, the center of mass can be used to facilitate the reconstruction of a 3-D structure from multiple tomographic projections. This property follows from the Fourier slice theorem, which states that the Fourier transform of a 2-D tomographic projection of a 3-D object is a central slice in the Fourier transform of the 3-D object. A consequence of this theorem is that the 3-D object can be reconstructed from its tomographic projections by fitting together the Fourier transforms of many 2-D projection images~\cite{van1987commonlines}. However, since any translation in image space leads to a modulation in Fourier space, all the tomographic projections must be aligned. This alignment can be done via the center of mass as, under the assumption of linearity in the projection and image detector stages, the center of mass of an object is projected to the center of mass of each of its tomographic projections (this assumption holds in many cases, e.g., under the weak-phase object approximation in our case study of cryo-electron microscopy~\cite{frank2006threeB}). Such alignment is especially important as many areas of science and engineering involve tomographic projection (mostly through imaging techniques), including archeology, medical imaging, structural biology, geophysics, plasma physics, materials science, astrophysics, and more. 

When the available data is noisy, identification of the center of mass becomes challenging. This, however, is precisely the type of observations available in many applications in astrophysics~\cite{roopashree2010centroid, thomas2006comparison} and biology~\cite{frank2006threeB}. To address the difficulties in estimating the center of mass in the presence of high levels of noise, several approaches have been introduced. These include template matching~\cite{sigworth2004noise}, thresholding and weighted center of mass~\cite{roopashree2010centroid}. In this paper, we follow a different approach, wherein the median of mass (also known as the geometric median) serves as a robust estimator of the center of mass~\cite{fletcher2008gm}.  Unfortunately, while the geometric median is more robust to noise than the center of mass, its computation is affected by \rev{noise with non-zero mean and by} near-by systems (objects). \rev{We therefore present a new tool for estimating the center of mass of systems (objects) in noisy images. Our suggested tool is closely related to the geometric median, while retaining robustness to non-zero mean noise and to near-by systems (objects). }

\begin{figure}[t]
	\begin{center}
		\begin{subfigure}[t]{0.23\textwidth}
			\centering
			\includegraphics[width=\textwidth]{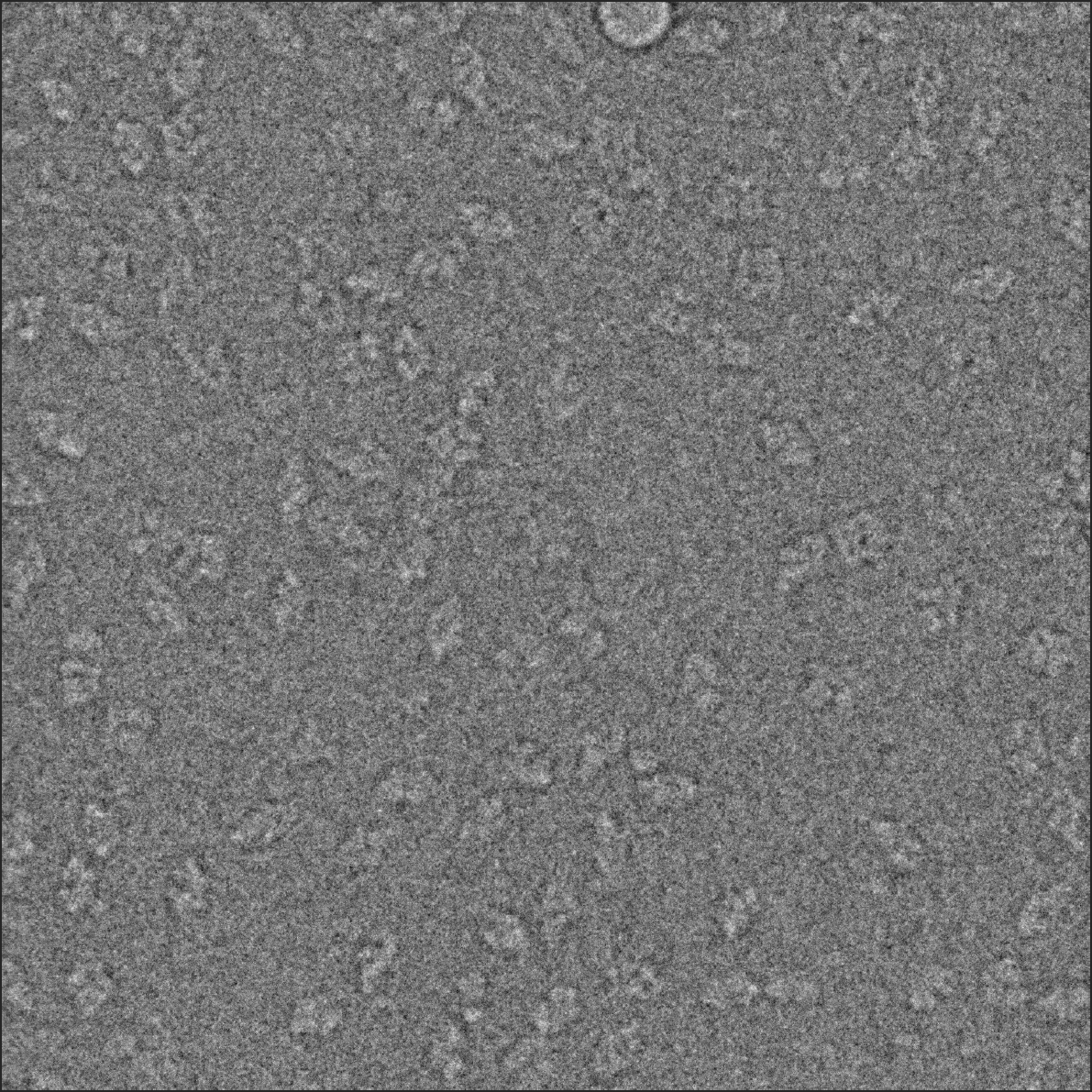}
			\caption{}
			\label{fig:micrograph1}
		\end{subfigure}
		\begin{subfigure}[t]{0.23\textwidth}
			\centering
			\includegraphics[width=\textwidth]{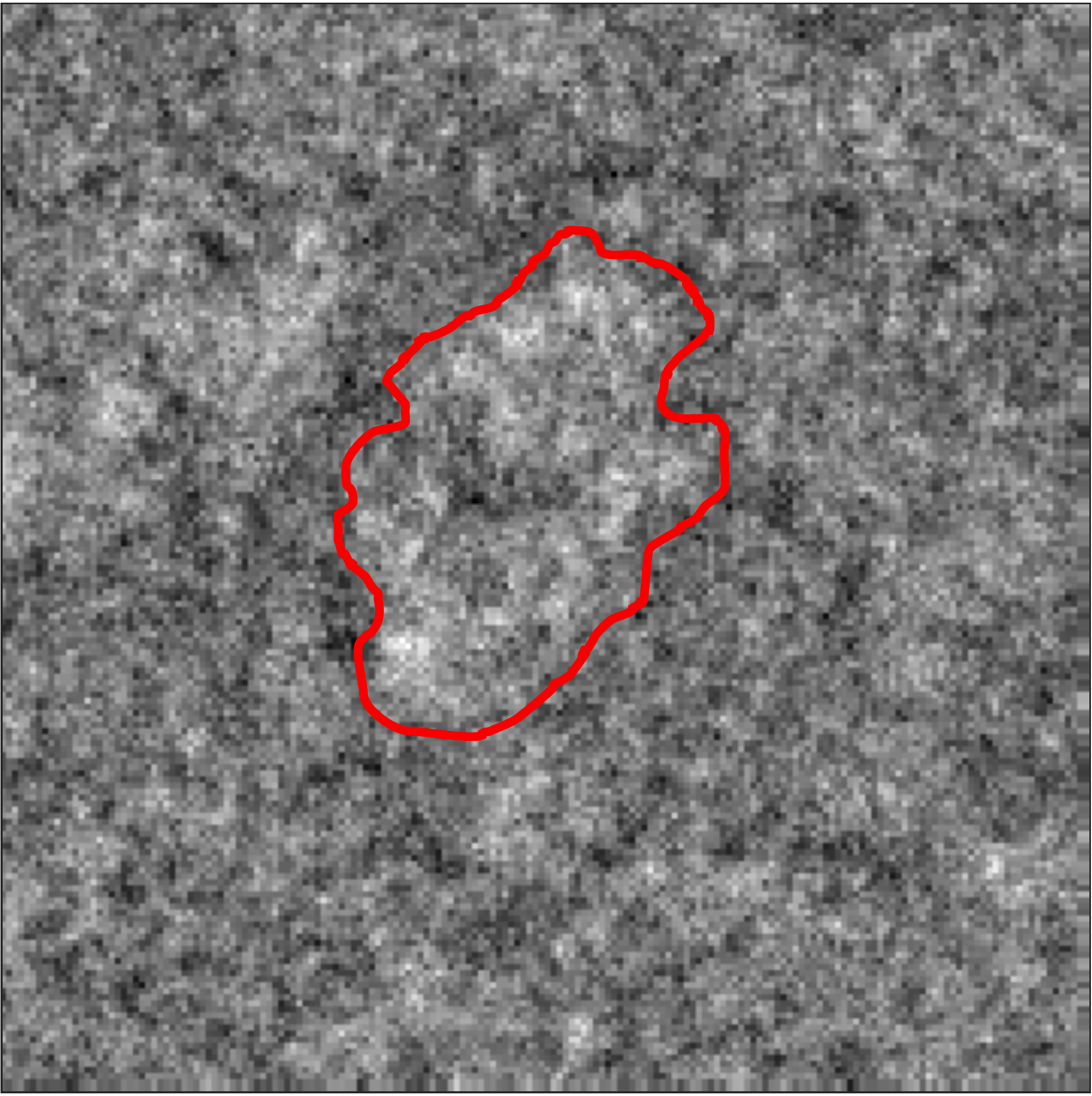}
			\caption{}
			\label{fig:zoom1}
		\end{subfigure}
		\begin{subfigure}[t]{0.23\textwidth}
			\centering
			\includegraphics[width=\textwidth]{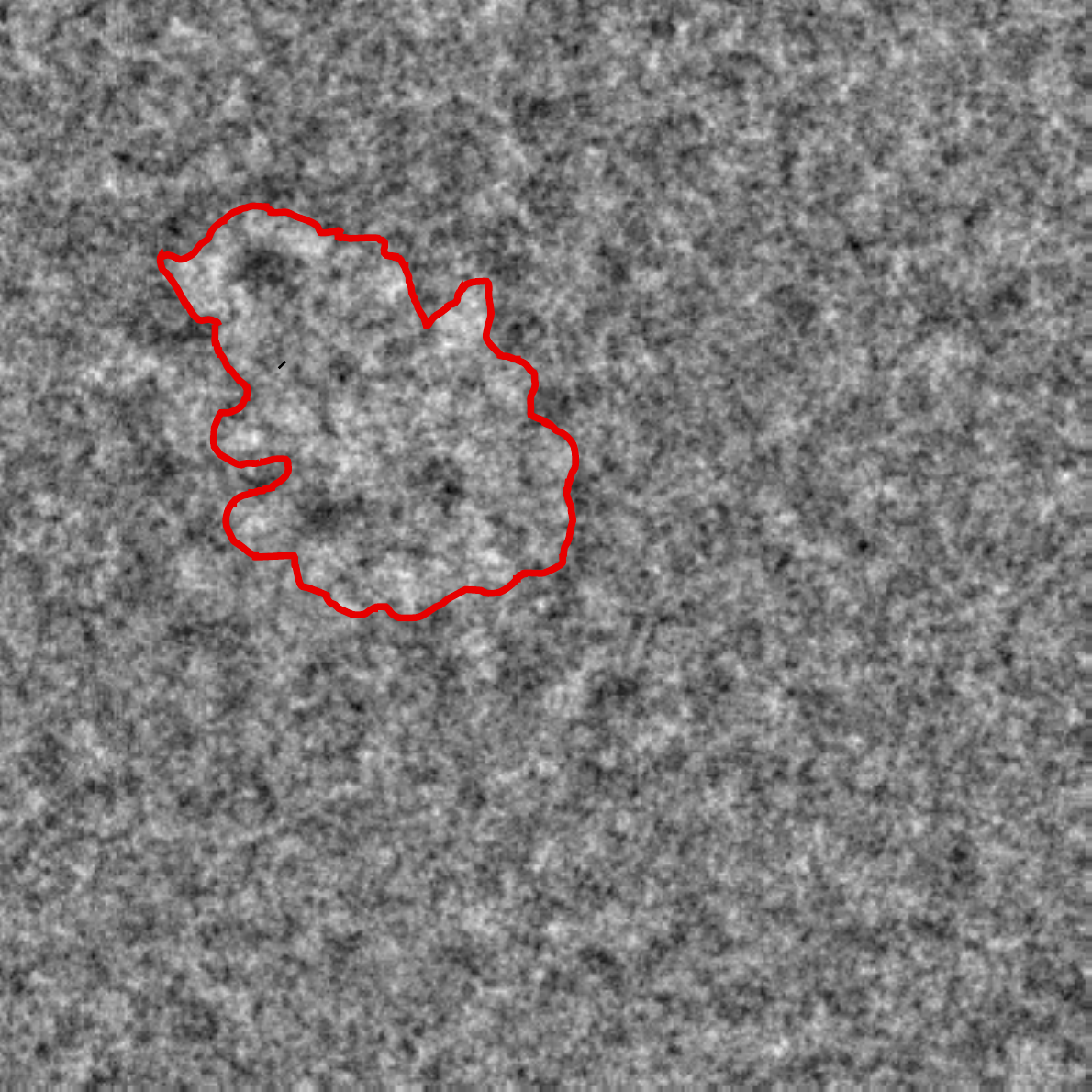}
			\caption{}
			\label{fig:zoom2}
		\end{subfigure}\\
		\begin{subfigure}[t]{0.23\textwidth}
			\centering
			\includegraphics[width=\textwidth]{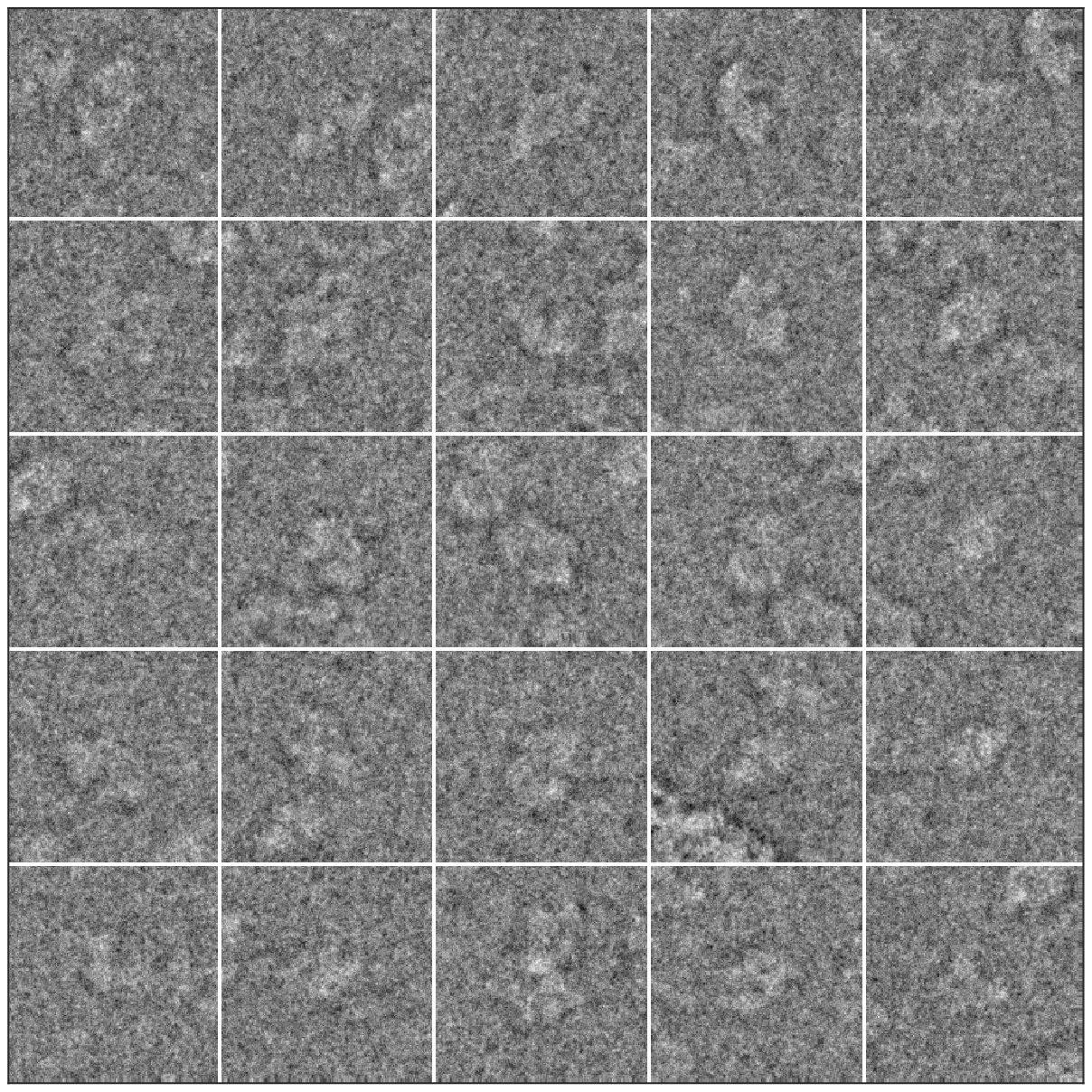}
			\caption{}
			\label{fig:picked1}
		\end{subfigure}
		\begin{subfigure}[t]{0.23\textwidth}
			\centering
			\includegraphics[width=\textwidth]{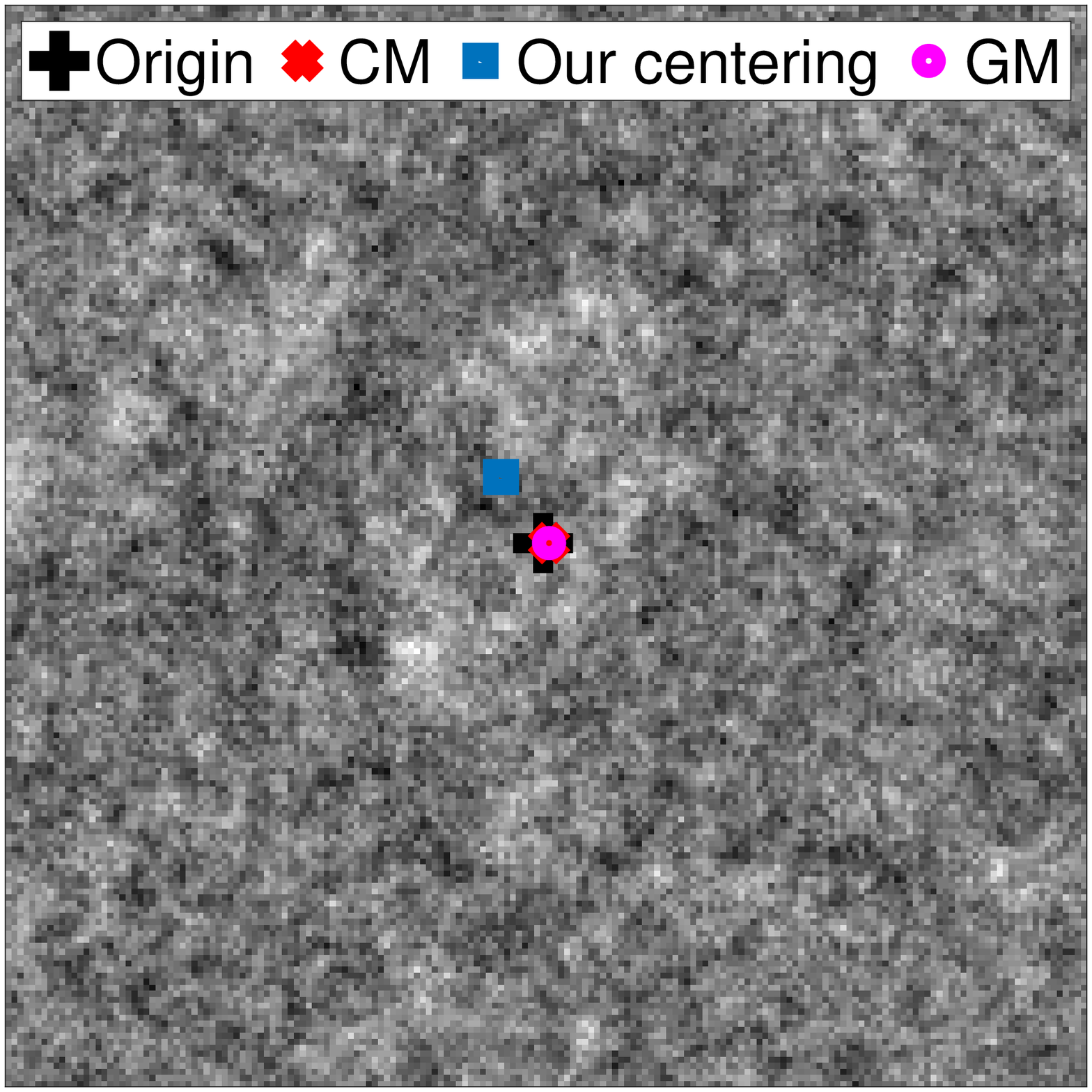}
			\caption{}
			\label{fig:center1}
		\end{subfigure}
		\begin{subfigure}[t]{0.23\textwidth}
			\centering
			\includegraphics[width=\textwidth]{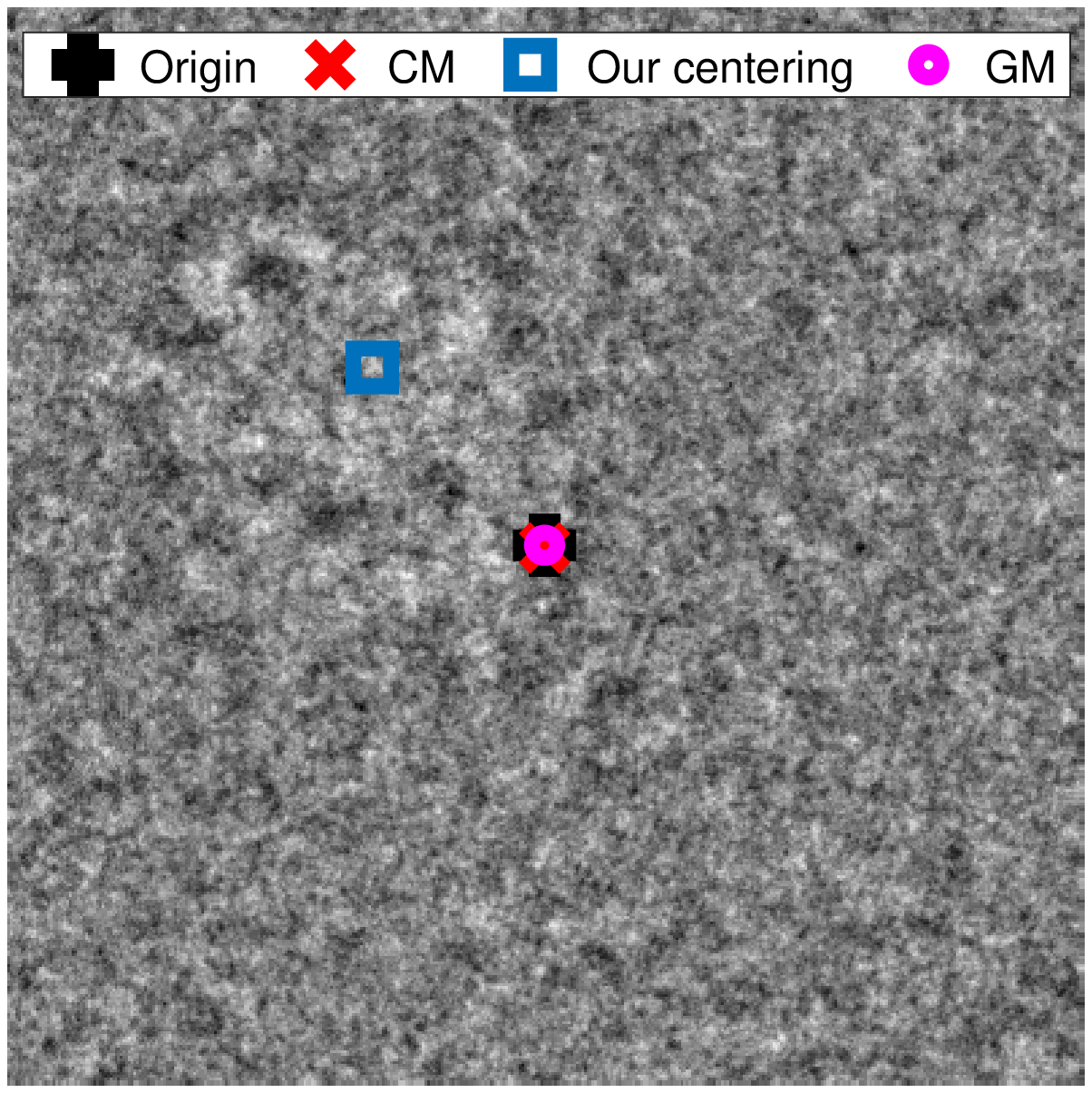}
			\caption{}
			\label{fig:center2}
		\end{subfigure}
		\caption{Cryo-EM experimental image of $\beta$-galactosidase taken from the EMPIAR-10017 dataset~\cite{iudin2016empiar, scheres2015relion} and 80S ribosome taken from the EMPIAR-10028~\cite{wong2014empiar} dataset.  (\subref{fig:micrograph1}) A section of an experimental image of $\beta$-galactosidase. There are many projections located close together.  
			(\subref{fig:zoom1}) - (\subref{fig:zoom2})  Zoom-in of projection images of $\beta$-galactosidase and ribosome, correspondingly. For the reader's convenience, we demarcated the projection of interest. 
			(\subref{fig:picked1}) Regions containing projections of $\beta$-galactosidase. The size of each region is $180 \times 180$ pixels. We note that, as the particles are dispersed with high density throughout the experimental image, we will often need to identify the center of a projection given a window that contains multiple projections. This will cause a direct computation of the geometric median to fail.
			(\subref{fig:center1})-(\subref{fig:center2}) Approximated centers of panel \subref{fig:zoom1} and \subref{fig:zoom2}.  The center of mass and geometric median coincide with the origin \rev{(which is the pixel located at the intersection of the central row and central column)} of the projection image (an explanation of this appears in Section~\ref{sec:cm_gm}), while our method provides a better approximation of the center of mass of the projection. }
		\label{fig:experimental_intro}
	\end{center}
\end{figure}

\rev{As such}, we present in this paper the following contributions. First, we introduce a surrogate function to the center of mass. We explore the mathematical properties of this function and, specifically, its connection with the geometric median, which is a robust approximation of the center of mass. Furthermore, we show that in cases where the geometric median is affected by near-by objects or extreme noise, our surrogate function retains its reliability. Another contribution is the idea of angular denoising, wherein uncorrelated noise, found along each ring of a constant radius in the image, is averaged out. This denoising is self-contained in each image, as we do not use any data external to the image. Due to the angular denoising, our surrogate function has higher robustness to noise and offers a powerful tool to analyze noisy images. 

\rev{As a case study, we use} single particle  cryo-electron microscopy (cryo-EM). The goal of cryo-EM is to resolve the 3-D structure of macromolecules at near atomic resolutions (typically under $0.3-0.4$ nm)~\cite{nogales2015development}. In this method, many instances of a specimen are embedded into vitreous ice and imaged in an electron microscope.  As a result, each experimental image contains many noisy projections located in close proximity, and the center of mass of each projection is of interest. An example of such experimental images is presented in Figures~\ref{fig:micrograph1} and~\ref{fig:picked1}. 
A demonstration of the advantages of our surrogate function over direct computation of the center of mass and  geometric median is presented in Figures~\ref{fig:zoom1},~\ref{fig:zoom2},~\ref{fig:center1} and~\ref{fig:center2}.

We use images of macromolecules collected by an electron microscope to experimentally verify our centering method. In particular, we show that for a dataset of the 80S ribosome our method decreases the average shift of the estimated center of mass from the true center of mass by $35\%$. We demonstrate similar improvements in assessing the image centers on a dataset of  $\beta$-galactosidase. Furthermore, we confirm that our method facilitates the use of projections that would otherwise be unusable. We demonstrate the importance of these additional projections on a TRPV1 dataset, where the employment of our method leads to a better reconstruction. \rev{Lastly, we show that reconstruction from centered projection images will reduce the overall runtime for reconstruction}. 

This paper is organized as follows. In Section~\ref{sec:P3DA}, we establish the theoretical aspects of our alignment method. Section~\ref{sec:algorithm} describes in detail the proposed algorithm. In Section~\ref{sec:numerics}, we present numerical examples to demonstrate the performance of our algorithm in various settings, as well as the effect of our method on 3-D reconstruction in cryo-EM. Finally, we conclude the paper in Section~\ref{sec:conclusion}. 

For research reproducibility purpose, MATLAB code is publicly available as a package at the following link: \href{https://github.com/nirsharon/RACER}{https://github.com/nirsharon/RACER}. Additionally, GPU-enabled code of our method is also available at the link:
\href{https://github.com/ayeletheimowitz/noisy_image_centering}{https://github.com/ayeletheimowitz/noisy\_image\_centering}.

\section{Translational Alignment} \label{sec:P3DA}

\subsection{Center of mass and geometric median}
\label{sec:cm_gm}

In this section we present the formal definitions of the center of mass and the geometric median and show they are unsuitable for center of mass estimation in noisy experimental images.

\begin{definition}[center of mass]
	The center of mass (CM) of a finite set of points $\{ p_i \}_{i=1}^n \subset \mathbb{R}^N$ and their associated non-negative weights $\{ \omega_i \}_{i=1}^n \subset \mathbb{R}_{+}$ is defined as their weighted average,
	\begin{equation} \label{eqn:CM_as_arithmetic_mean}
	\mu =  \frac{1}{\sum_{j=1}^n \omega_j}  \sum_{i=1}^n \omega_i  p_i  .
	\end{equation}
\end{definition}

An alternative way to define the CM is \rev{as follows,}
\begin{definition}[center of mass]
	The center of mass (CM) of a finite set of points $\{ p_i \}_{i=1}^n \subset \mathbb{R}^N$ and their associated non-negative weights $\{ \omega_i \}_{i=1}^n \subset \mathbb{R}_{+}$ is defined as  the minimizer of the Fr\'echet variance,
	\begin{equation} \label{eqn:CM}
	\mu =  \arg \min_x \sum_{i=1}^n \omega_i d^2(p_i,x) .
	\end{equation}
\end{definition}

The Fr\'echet mean is greatly influenced by outliers and extreme values. We therefore suggest the use of the Fr\'echet median as a robust estimation of the CM~\cite{fletcher2008gm}. The Fr\'echet  median is also known as the median of mass or the geometric median.
\begin{definition}[geometric median] \label{def:GM}
	The geometric median (GM) of a finite set of points $\{ p_i \}_{i=1}^n$ and their associated non-negative weights $\{ \omega_i \}_{i=1}^n$ is defined using the Fr\'echet median as  
	\begin{equation} \label{eqn:MOM}
	\mu_1 =  \arg \min_{x} \sum_{i=1}^n \omega_i d(p_i,x) .
	\end{equation}
\end{definition}

In our setting, the set of points $\pset = \{ p_i \}_{i=1}^n$ is the parametric description of the image (that is, the image grid) and the weights are the intensities of the image (pixel values).  We therefore reformulate the CM of an image as
\begin{equation} \label{eqn:image_CM}
\mu =  \arg \min_{x \in \pset} \sum_{p_i \in \pset} I \left( p_i \right) d^2(p_i,x) ,
\end{equation}
where $I$ is the  image we aim to center. We further reformulate the GM as
\begin{equation} \label{eqn:moving_MOM}
\mu_1 =  \arg \min_{x \in \pset} \sum_{p_i \in \pset} I \left( p_i \right) d(p_i,x) .
\end{equation}
We note that, \rev{as $I(p_i)$ represents mass, the pixel values of the noiseless image are assumed to be non-negative. Additionally, since }
the image is discretized, the true CM may be off-grid. As we do not have any interest in sub-pixel accuracy, we ignore this option, and search for a minimizer on the pixel grid. That is, we limit $x \in \pset$ in~\eqref{eqn:image_CM} and~\eqref{eqn:moving_MOM}. In addition, unless the nonzero pixels in the image are collinear, the GM is uniquely determined~\cite{weiszfeld1937gm}. Therefore, we will consider only cases of images having a unique GM.

While the GM has higher robustness to noise than the CM, it is still susceptible to errors. One such scenario is when the object\footnote{\rev{By object we mean a single, connected (physically and conceptually) region in the image with non-zero mass.}} we wish to center is in close proximity to other entities. In cryo-EM, for example, several projections of a particle will often appear next to each other. In this case, the value of $\sum_{p_i \in \pset} I \left( p_i \right) d(p_i,x)$ can contain contributions from  nearby objects, effectively masking the minimum of \eqref{eqn:moving_MOM}.  This is exemplified in Figure~\ref{fig:geomean}.

\begin{figure}[t]
	\begin{center}
		\begin{subfigure}[t]{0.25\textwidth}
			\centering
			\includegraphics[width=\textwidth]{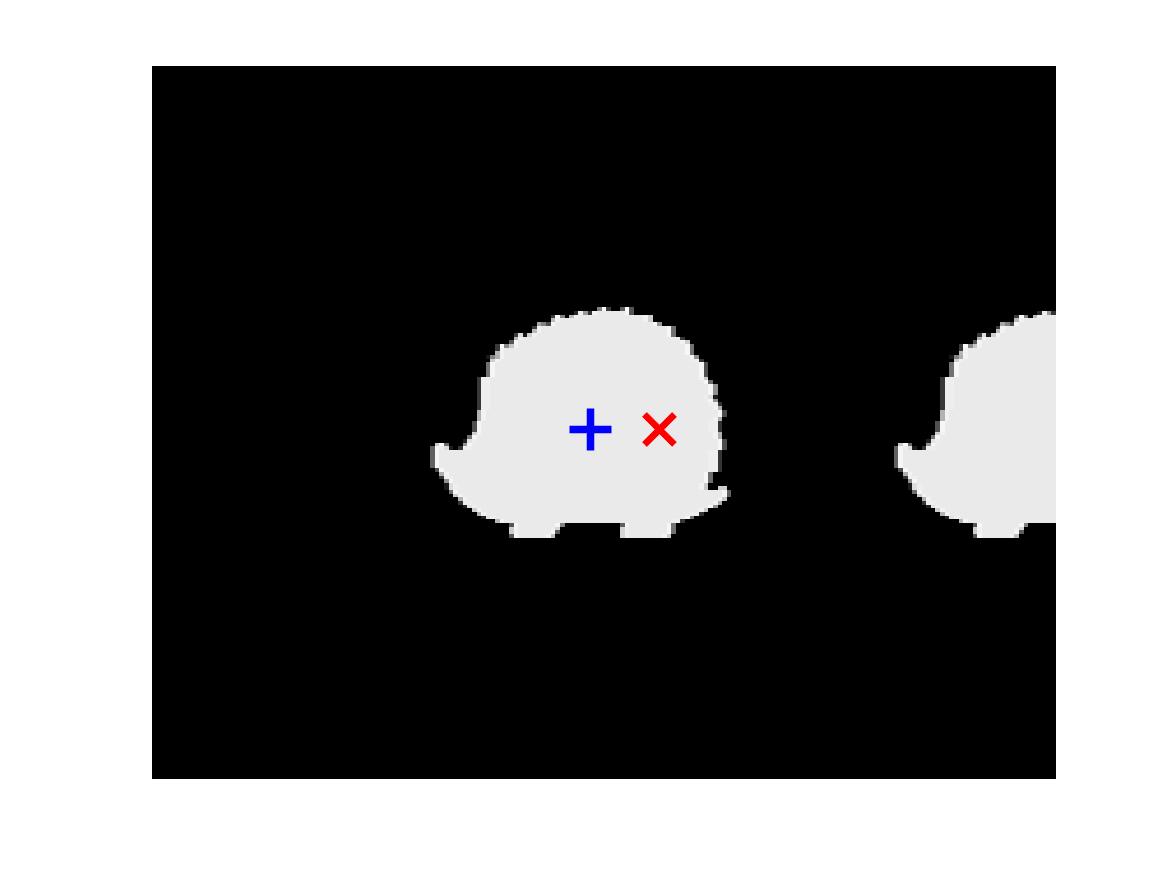}
			\caption{}
		\end{subfigure} \qquad
		\begin{subfigure}[t]{0.25\textwidth}
			\centering
			\includegraphics[width=\textwidth]{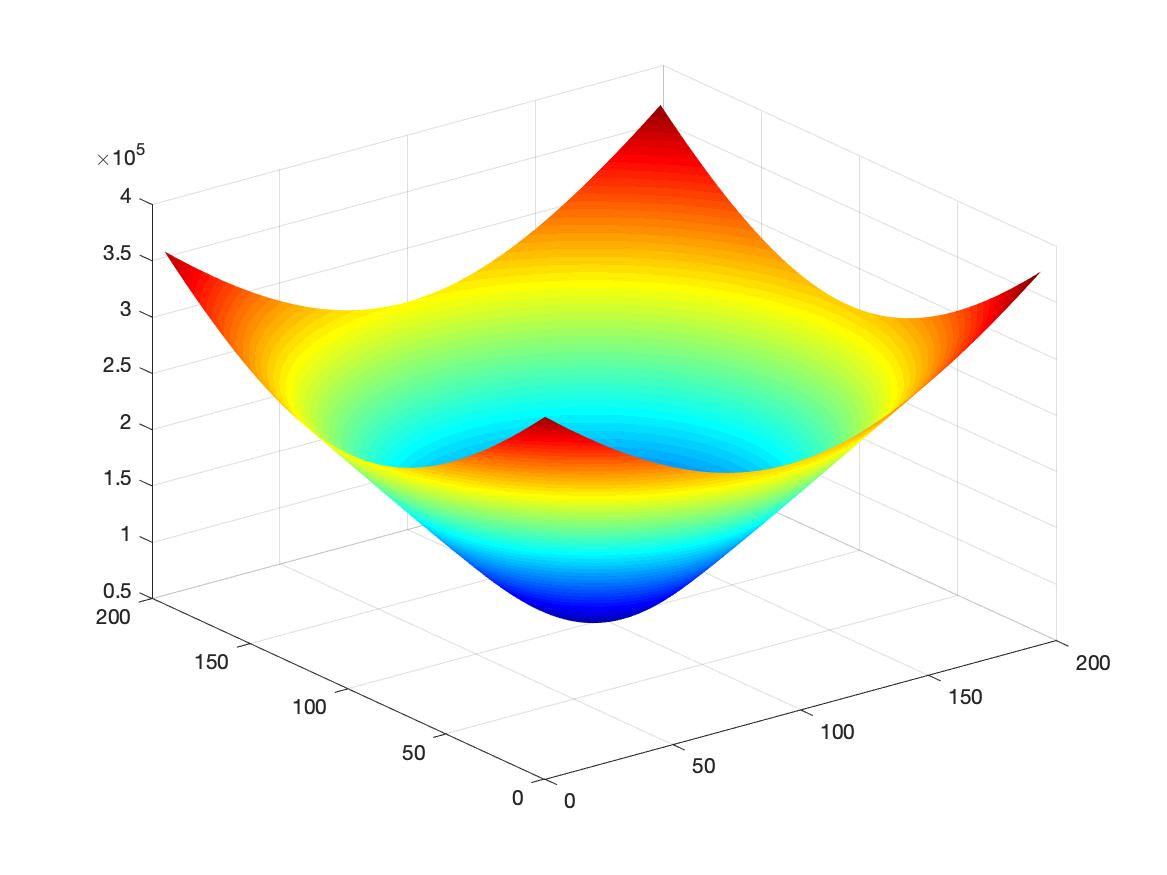}
			\caption{}
		\end{subfigure} \qquad
		\begin{subfigure}[t]{0.25\textwidth}
			\centering
			\includegraphics[width=\textwidth]{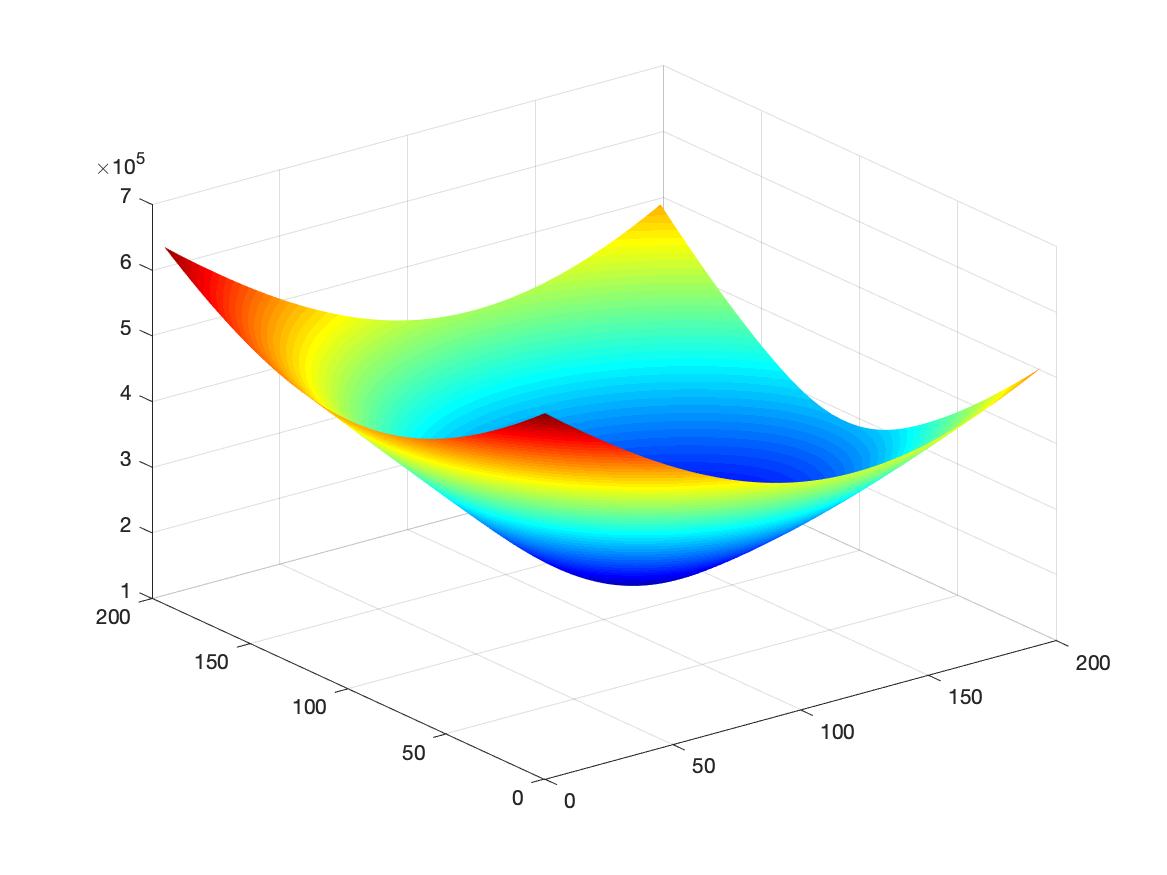}
			\caption{}
		\end{subfigure}
		\caption{Illustration of the misidentification of the GM. (a) Image of a hedgehog silhouette and a partial hedgehog silhouette. This image includes the GM of the hedgehog (blue $+$) and the GM of the total image (red x). (b) Landscape of the GM of the single hedgehog. (c) Landscape of the GM of the image in panel (a). The landscape of the GM is defined as $\sum_{p_i \in \pset} I \left( p_i \right) d(p_i,x)$.}
		\label{fig:geomean}
	\end{center}
\end{figure}

For images containing high levels of noise, as is the case with cryo-EM experimental images, the Fr\'echet mean is mostly determined by the noise. The GM, on the other hand, is more resilient to zero-mean noise. For other noise models, the GM is also highly affected by the added \rev{non-zero mean} noise. An illustration of this effect is presented in Figure~\ref{fig:noisy_sil}, where we observe gradual contamination by two types of noise. In our centering method, which we present next, we suggest a surrogate function to the CM. As seen in Figure~\ref{fig:noisy_sil}, our approach is robust to the aforementioned limitations of the CM and GM. 

\begin{figure}
	\begin{center}
		\begin{subfigure}[t]{\textwidth}
			\centering
			\includegraphics[width=.95\textwidth]{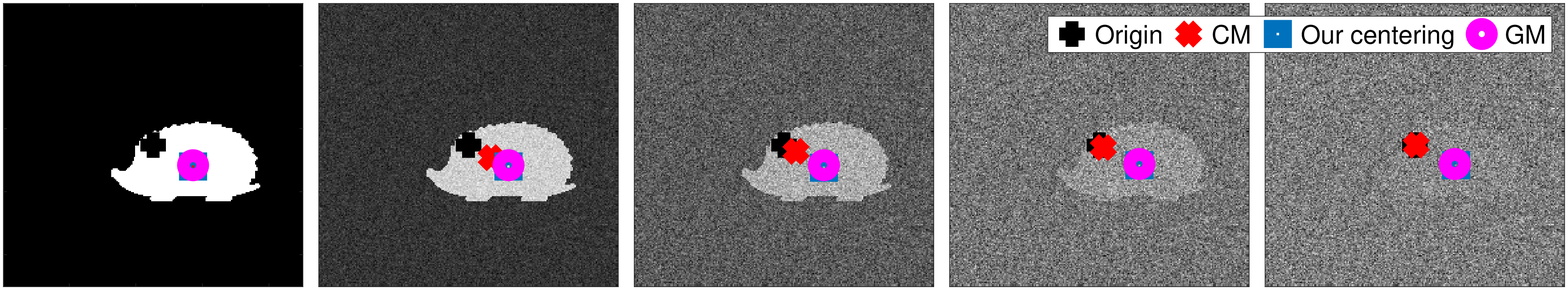}
			\caption{}
			\label{fig:noisy_gaussian_row}
		\end{subfigure} \\
		\begin{subfigure}[t]{\textwidth}
			\centering
			\includegraphics[width=.95\textwidth]{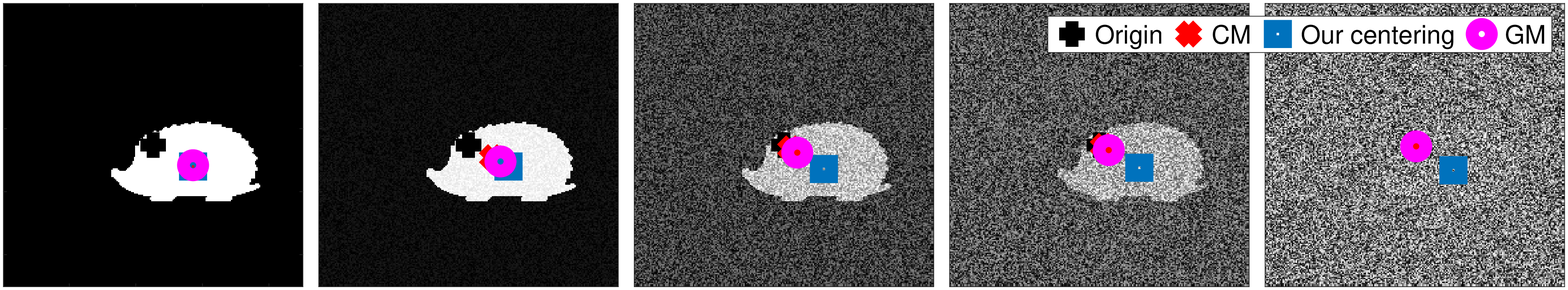}
			\caption{}
			\label{fig:noisy_uniform_row}
		\end{subfigure}
		\caption{The effect of noise over the CM, GM, and our centering. From left to right, we present a  clean shifted image followed by a series of noisy shifted silhouettes. In~(\subref{fig:noisy_gaussian_row}), we use  growing levels of Gaussian noise \rev{(which will cause pixel values to become negative)} while  in~(\subref{fig:noisy_uniform_row}) the added noise is positive uniform. As the level of noise increases, the CM gravitates towards the origin. In~(\subref{fig:noisy_gaussian_row}), the GM and our centering remain in approximately the same location along the entire series whereas in~(\subref{fig:noisy_uniform_row}) the GM joins the CM while our centering stays in the same location.}
		\label{fig:noisy_sil}
	\end{center}
\end{figure}

\subsection{Angular averaging}
\label{app:informal_description}

We once again consider the simple example of the silhouette of a hedgehog. The pixel values within the \rev{(noiseless)} silhouette are strictly positive and equal one, while the background pixel values equal zero. We will compare the case where the CM of the hedgehog coincides with the origin of image and the case where it does not, and demonstrate that the rotational average (that is the average of all possible rotations)  in each case  is distinct.

In the case where the silhouette image has its CM at the origin, any application of a rotation operator spins the silhouette around its CM while keeping the CM in place. Summing all possible rotations results in a rotationally symmetric ``spread silhouette." We illustrate this process in Figures~\ref{fig:informal1}--\ref{fig:informal4}. Specifically, in Figure~\ref{fig:informal1} we present a centered silhouette of a hedgehog. Then, we carefully rotate the hedgehog to obtain several  rotated copies of this image, as seen in Figure~\ref{fig:informal2}. Adding these rotated copies together results in the image of Figure~\ref{fig:informal3}. Finally, by averaging over all possible rotations, we  obtain Figure~\ref{fig:informal4}. 

The principal observation is that any translation of the original hedgehog image, moving its CM away from the origin, will cause the ``spread silhouette" to cover a larger area of the rotationally averaged image. To visualize this, we shift the hedgehog by $10 \%$ of the maximum possible shift, that is $20$ pixels along each axis in our $400 \times 400$ image, as seen in Figure~\ref{fig:informal5}. As before, we introduce in Figure~\ref{fig:informal6} the same finite set of rotations as in Figure~\ref{fig:informal2}, now applied over the shifted image of Figure~\ref{fig:informal5}. The sum of these rotations causes the average image to be smeared across a larger area, as seen both in the partial sum of Figure~\ref{fig:informal7} as well as the average of all possible rotations in Figure~\ref{fig:informal8}. The difference between Figure~\ref{fig:informal4} (the centered case) and Figure~\ref{fig:informal8} (the shifted case) is clearly seen.

\begin{figure}[!htbp]
	\captionsetup[subfigure]{width=.8\textwidth}
	\centering
	\begin{subfigure}[t]{0.22\textwidth}
		\centering
		\includegraphics[width=.8\textwidth]{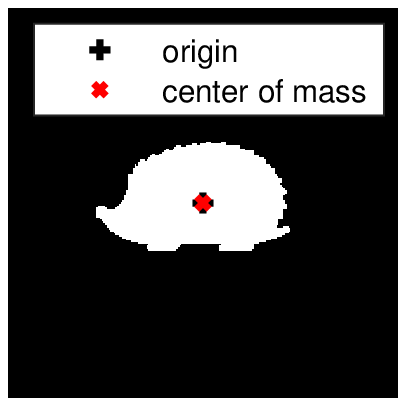}
		\caption{Centered}
		\label{fig:informal1}
	\end{subfigure} 
	\begin{subfigure}[t]{0.22\textwidth}
		\centering
		\includegraphics[width=.8\textwidth]{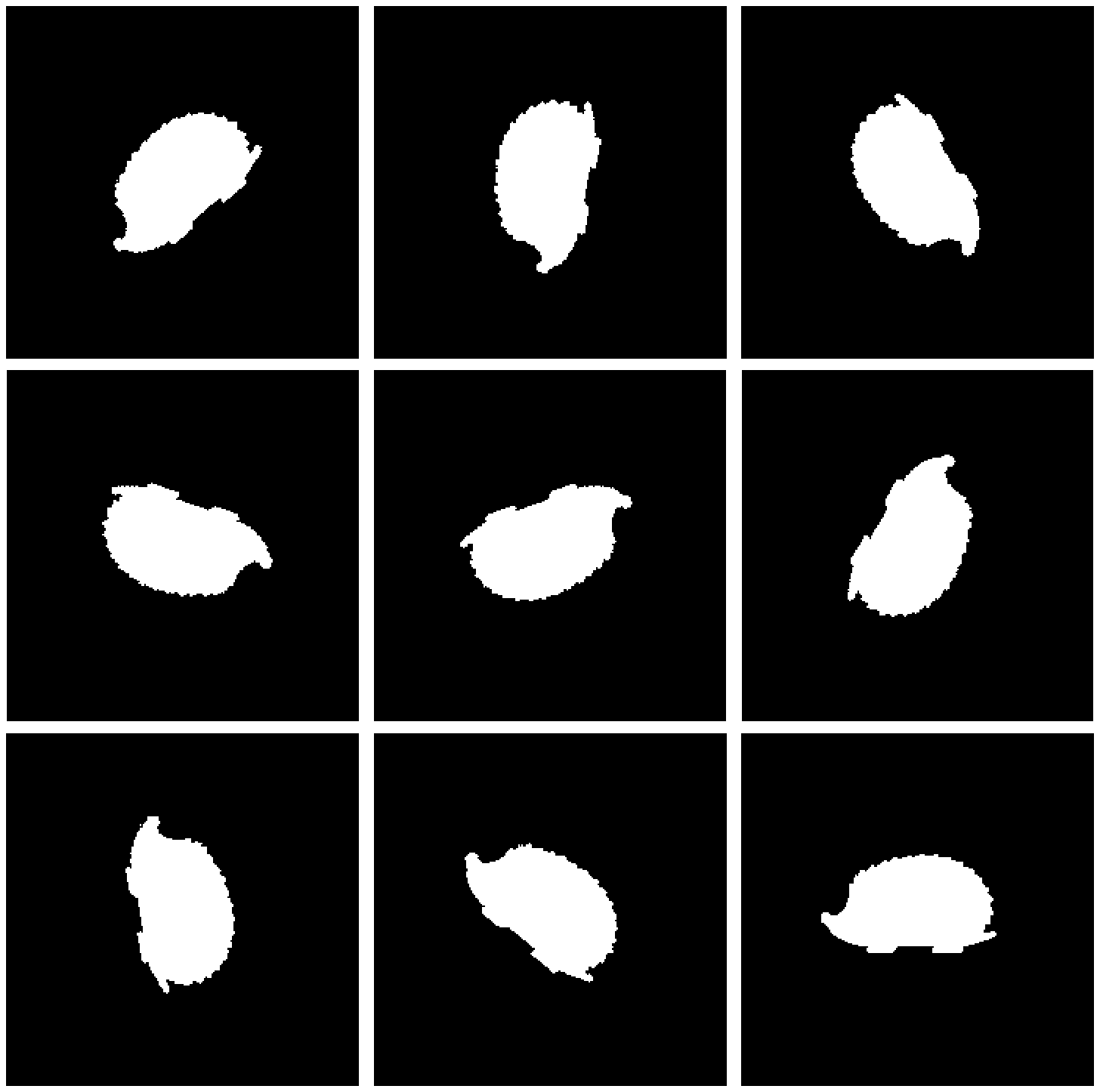}
		\caption{rotated copies}
		\label{fig:informal2}
	\end{subfigure} 
	\begin{subfigure}[t]{0.22\textwidth}
		\centering
		\includegraphics[width=.8\textwidth]{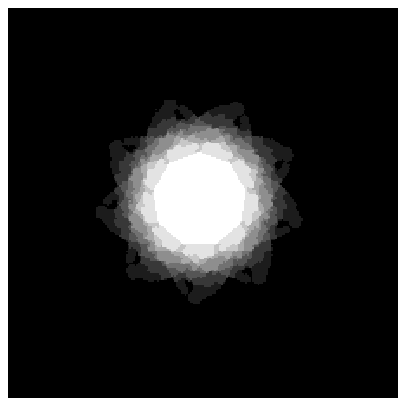}
		\caption{Sum of rotations in (\subref{fig:informal2})}
		\label{fig:informal3}
	\end{subfigure}%
	\begin{subfigure}[t]{0.22\textwidth}
		\centering
		\includegraphics[width=.8\textwidth]{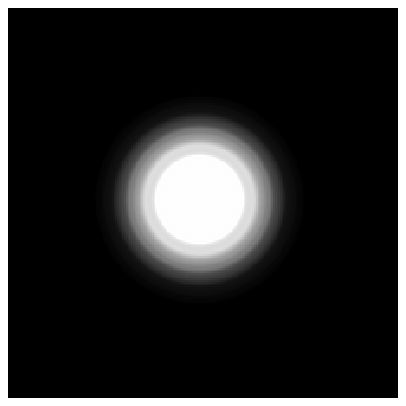}
		\caption{Average of all possible rotations}
		\label{fig:informal4}
	\end{subfigure} \\
	\begin{subfigure}[t]{0.22\textwidth}
		\centering
		\includegraphics[width=.8\textwidth]{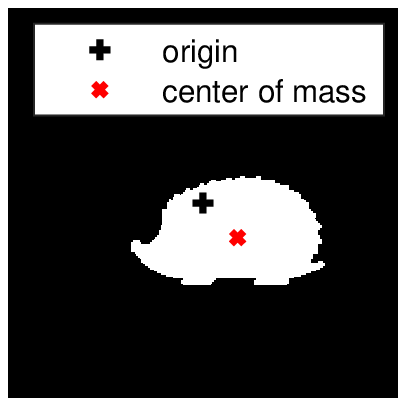}
		\caption{Shifted}
		\label{fig:informal5}
	\end{subfigure} 
	\begin{subfigure}[t]{0.22\textwidth}
		\centering
		\includegraphics[width=.8\textwidth]{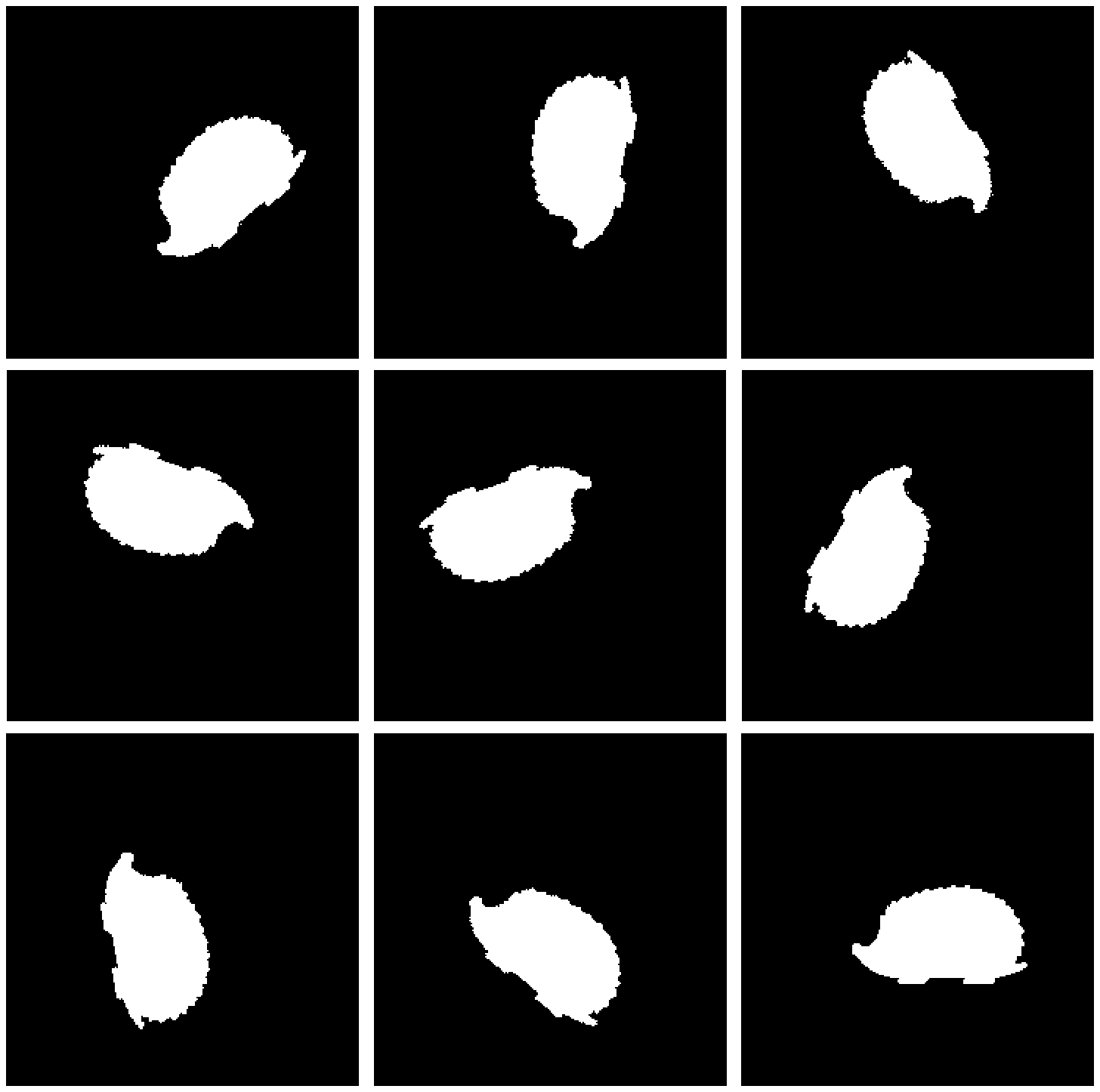}
		\caption{rotated copies}
		\label{fig:informal6}
	\end{subfigure} 
	\begin{subfigure}[t]{0.22\textwidth}
		\centering
		\includegraphics[width=.8\textwidth]{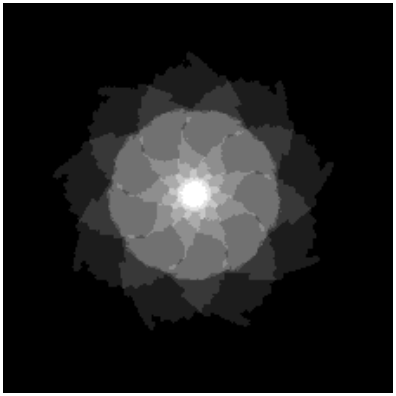}
		\caption{Sum of rotations in (\subref{fig:informal6})}
		\label{fig:informal7}
	\end{subfigure}%
	\begin{subfigure}[t]{0.22\textwidth}
		\centering
		\includegraphics[width=.8\textwidth]{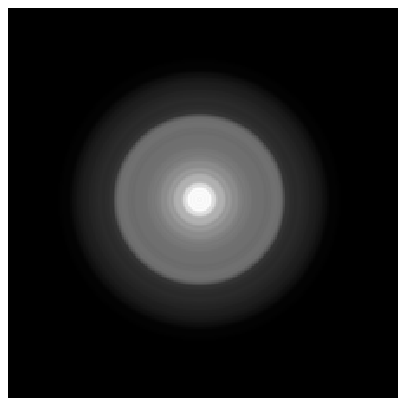}
		\caption{Average of all possible rotations}
		\label{fig:informal8}
	\end{subfigure}
	\caption{The difference in energy spread of rotational averaging of centered (top row) and uncentered (bottom row) hedgehogs. In each row, from left to right, we present the original silhouette, a few possible rotations, the sum of these rotated copies, and the average of all possible rotations of the leftmost image.}
	\label{fig:informal}
\end{figure}

One significant advantage of rotational averaging is its denoising effect. Namely, the pixels along each circle of a fixed radius are averaged, so the impact of noise decreases. Our angular denoising is therefore geared towards problems where the rotational average of an image is sufficient knowledge of its content. This allows us to contend with levels of noise that would render techniques applied directly on the noisy images useless. A demonstration of this phenomenon is given in Figure~\ref{fig:denoise}.

\begin{figure}
	\captionsetup[subfigure]{width=0.85\textwidth}
	\centering
	\begin{subfigure}[t]{0.22\textwidth}
		\centering
		\includegraphics[width=.8\textwidth]{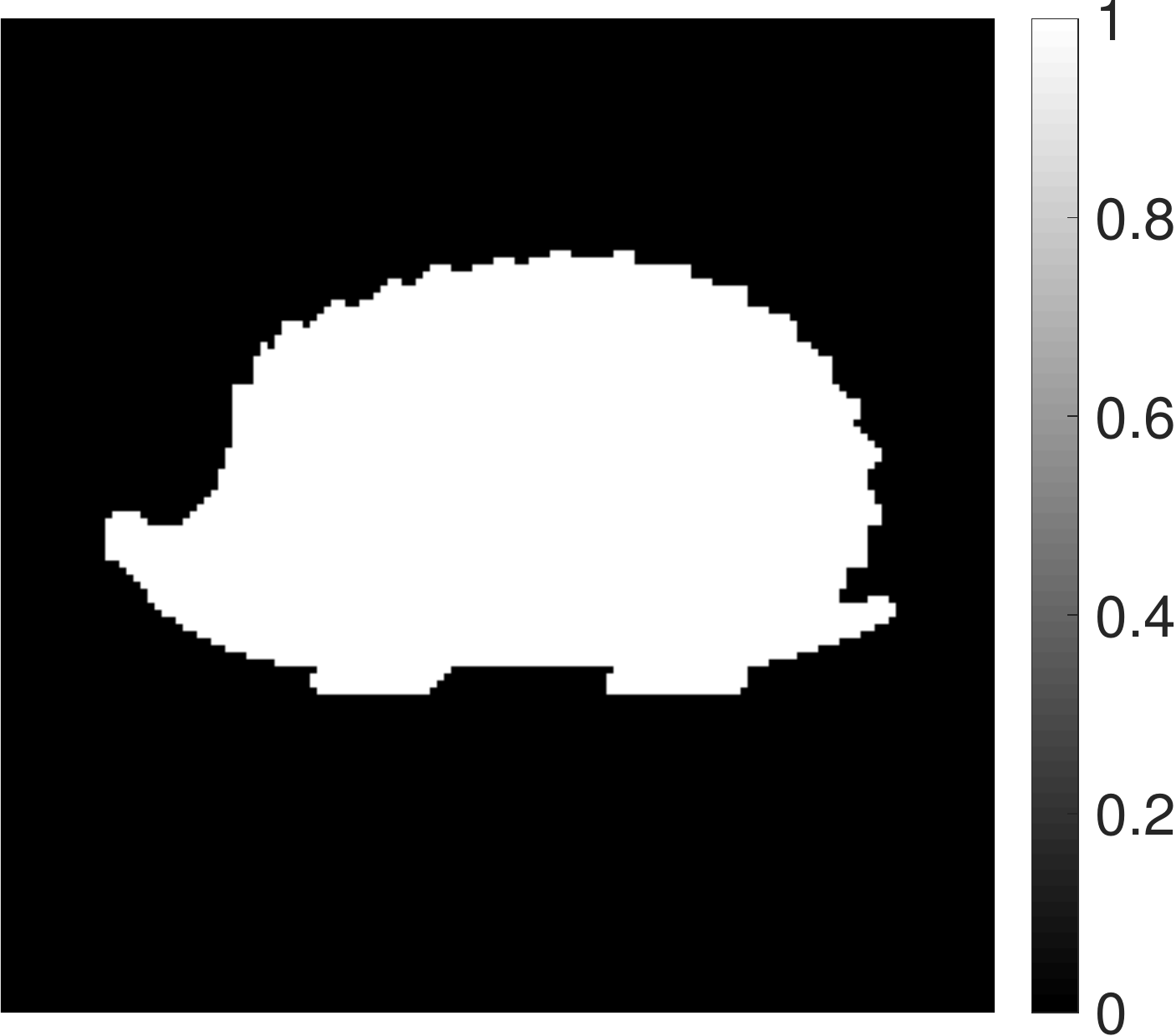}
		\caption{Clean silhouette}
		\label{fig:denoise1}
	\end{subfigure}
	\begin{subfigure}[t]{0.22\textwidth}
		\centering
		\includegraphics[width=.8\textwidth]{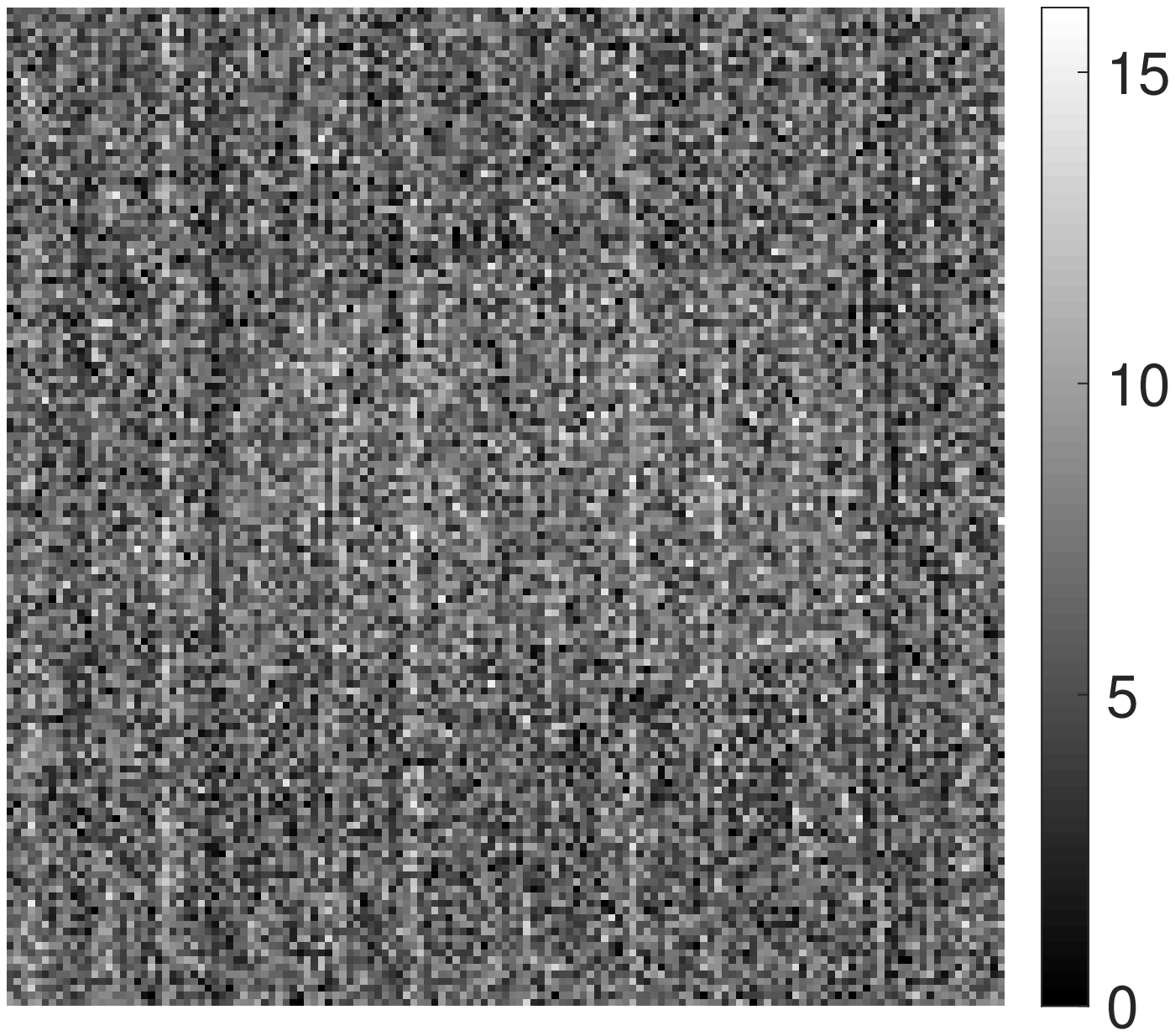}
		\captionsetup{width=.75\linewidth}
		\caption{Noisy silhouette with SNR $= 1 / 45$}
		\label{fig:denoise2}
	\end{subfigure}
	\begin{subfigure}[t]{0.22\textwidth}
		\centering
		\includegraphics[width=.7\textwidth]{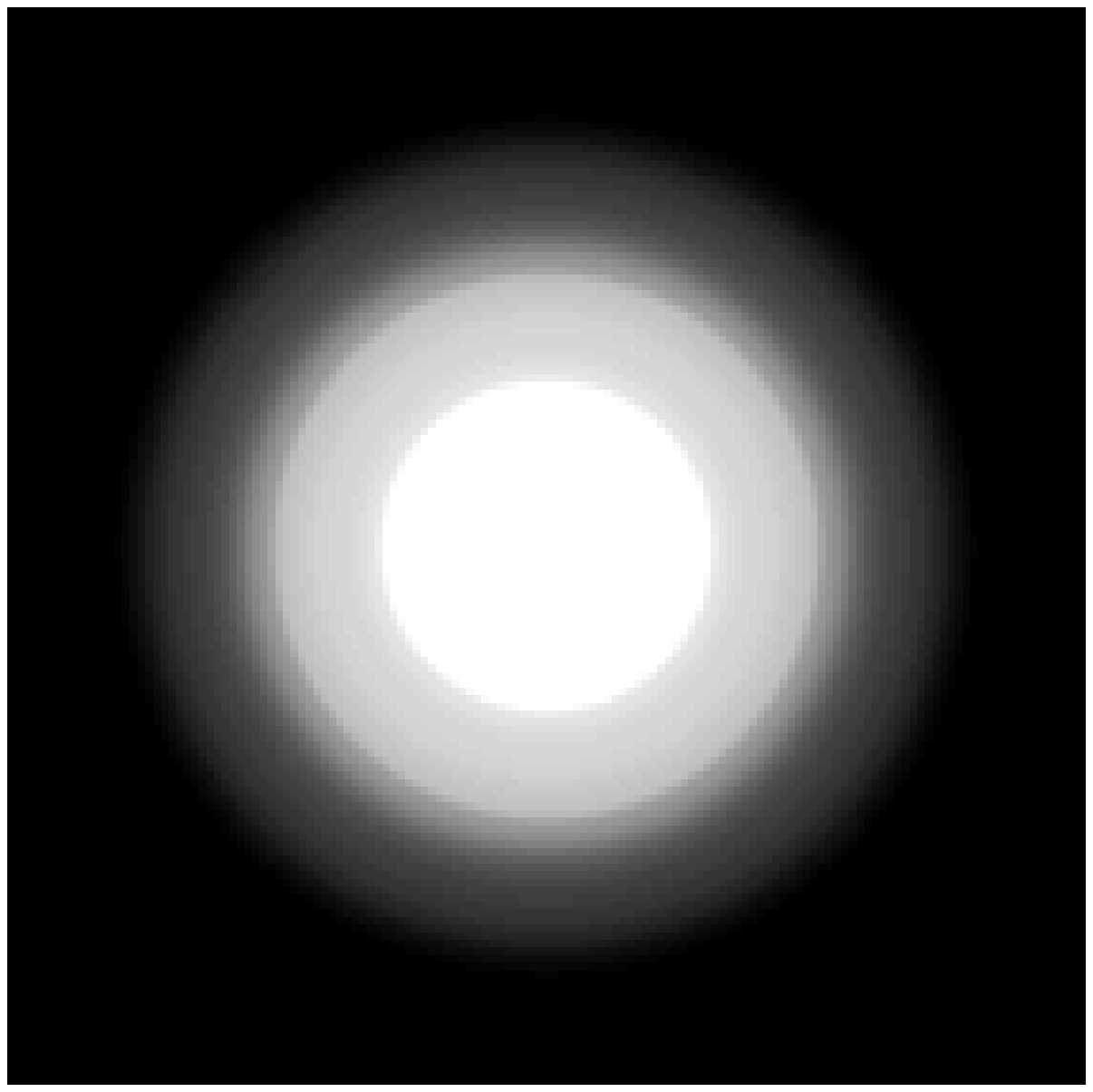}
		\caption{Rotational averaging of~(\subref{fig:denoise1}) }
		\label{fig:denoise3}
	\end{subfigure}
	\begin{subfigure}[t]{0.22\textwidth}
		\centering
		\includegraphics[width=.7\textwidth]{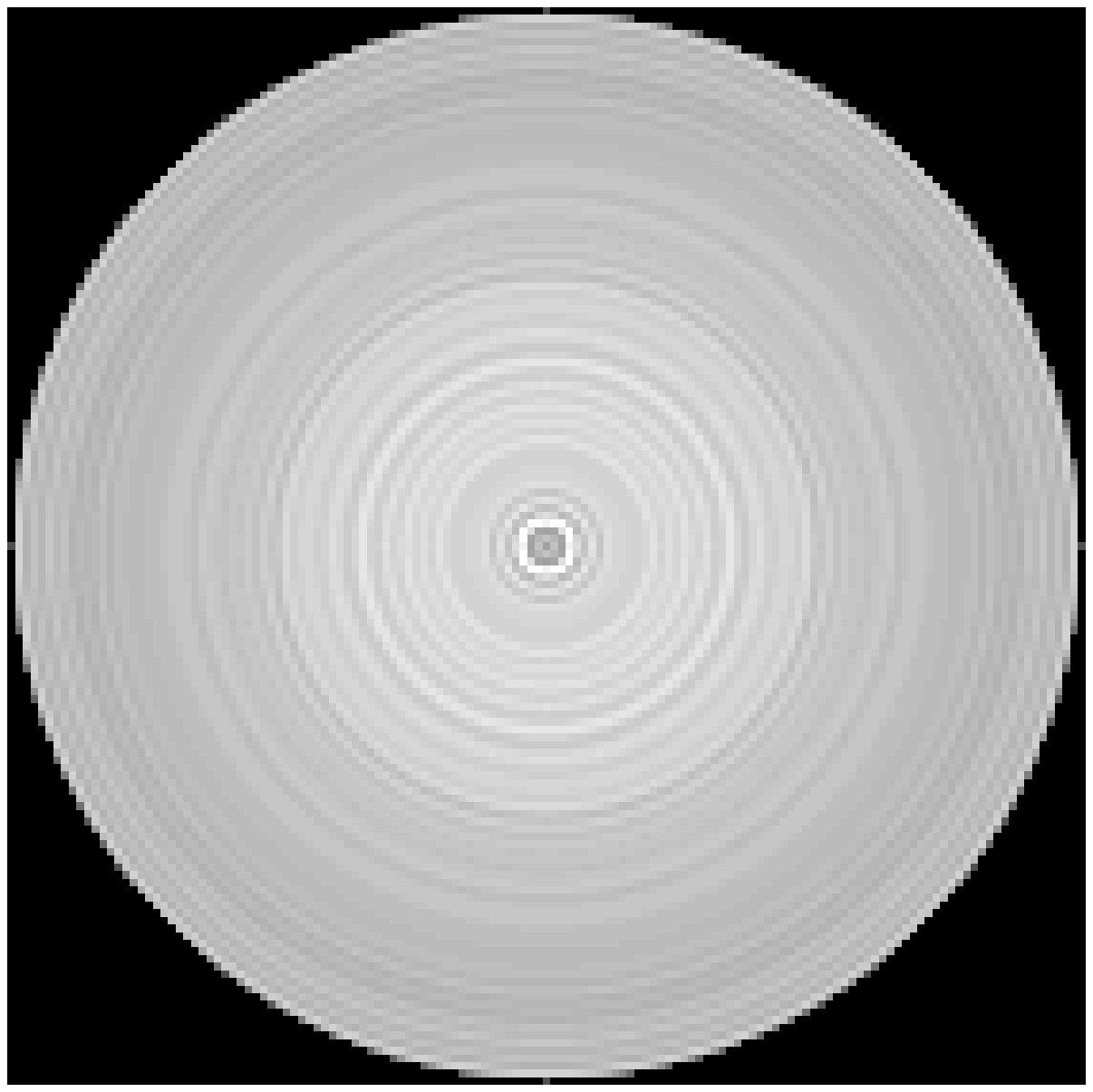}
		\captionsetup{width=.9\linewidth}
		\caption{Rotational averaging of~(\subref{fig:denoise2}). In view of~(\subref{fig:denoise3}), the SNR $= 1$}
		\label{fig:denoise4}
	\end{subfigure} \\
	\begin{subfigure}[t]{0.22\textwidth}
		\centering
		\includegraphics[width=.8\textwidth]{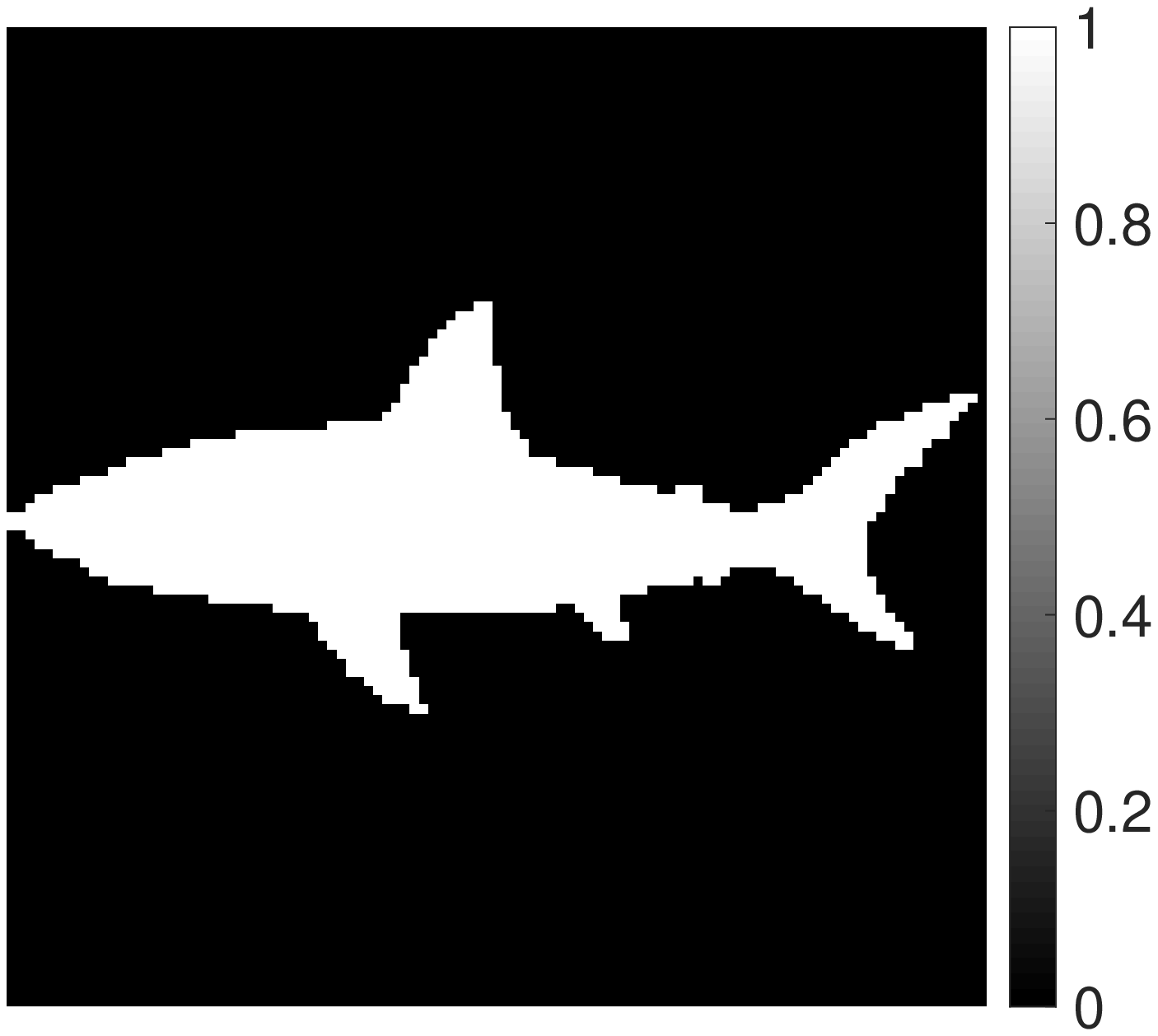}
		\caption{Clean silhouette}
		\label{fig:denoise1_shark}
	\end{subfigure}
	\begin{subfigure}[t]{0.22\textwidth}
		\centering
		\includegraphics[width=.8\textwidth]{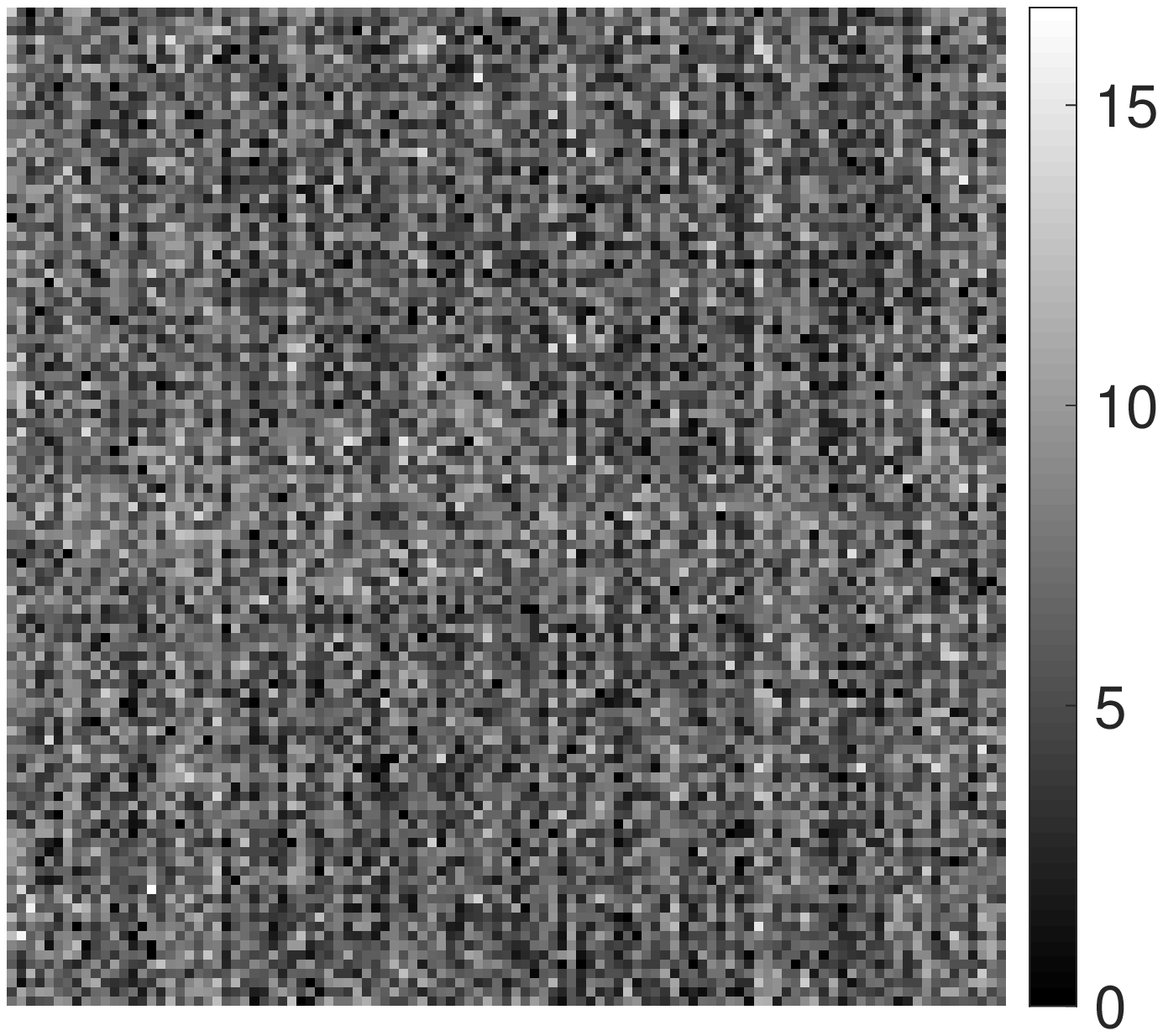}
		\captionsetup{width=.75\linewidth}
		\caption{Noisy silhouette with SNR $= 1 / 45$}
		\label{fig:denoise2_shark}
	\end{subfigure}
	\begin{subfigure}[t]{0.22\textwidth}
		\centering
		\includegraphics[width=.7\textwidth]{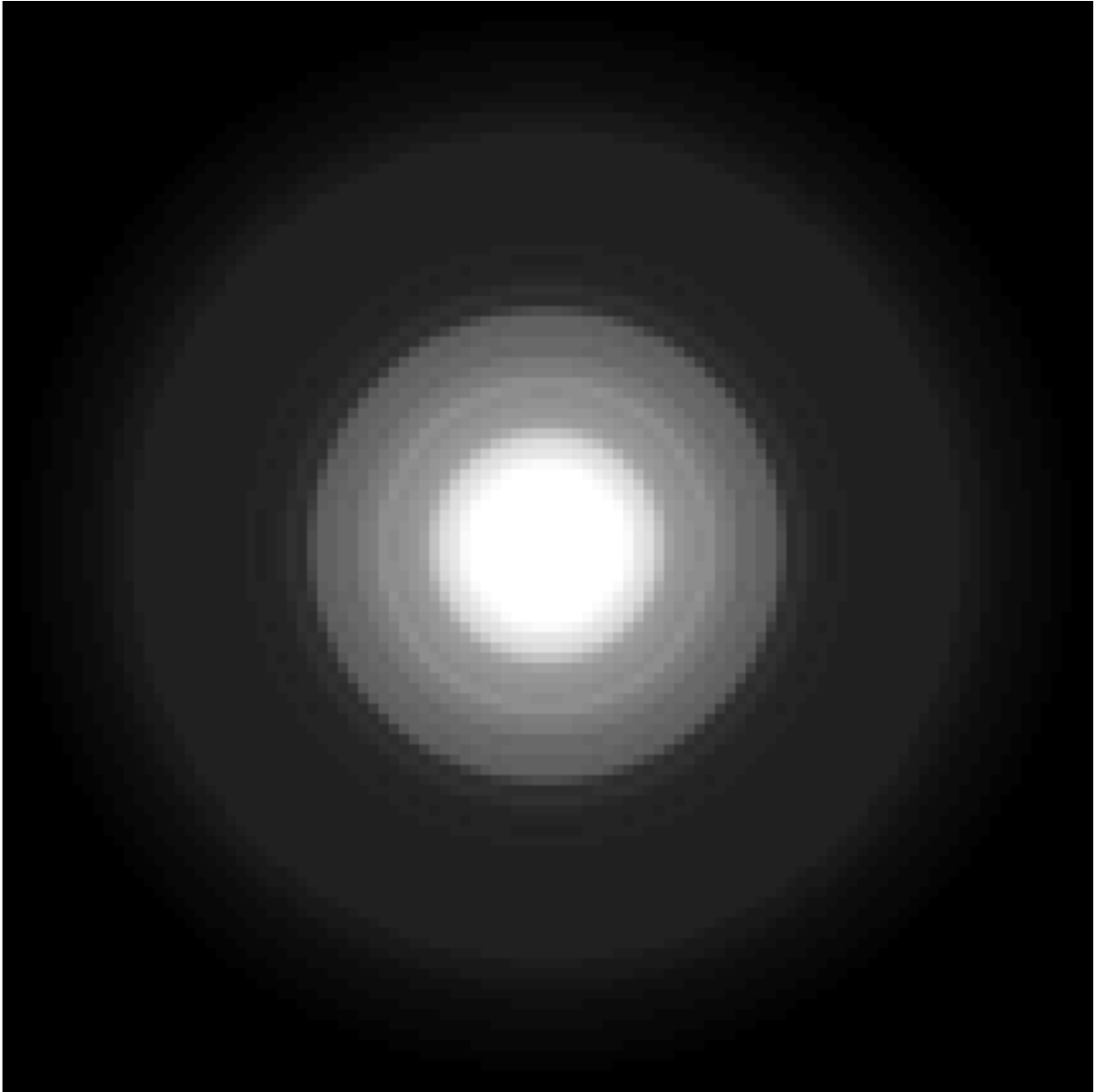}
		\caption{Rotational averaging of~(\subref{fig:denoise1_shark}) }
		\label{fig:denoise3_shark}
	\end{subfigure}
	\begin{subfigure}[t]{0.22\textwidth}
		\centering
		\includegraphics[width=.7\textwidth]{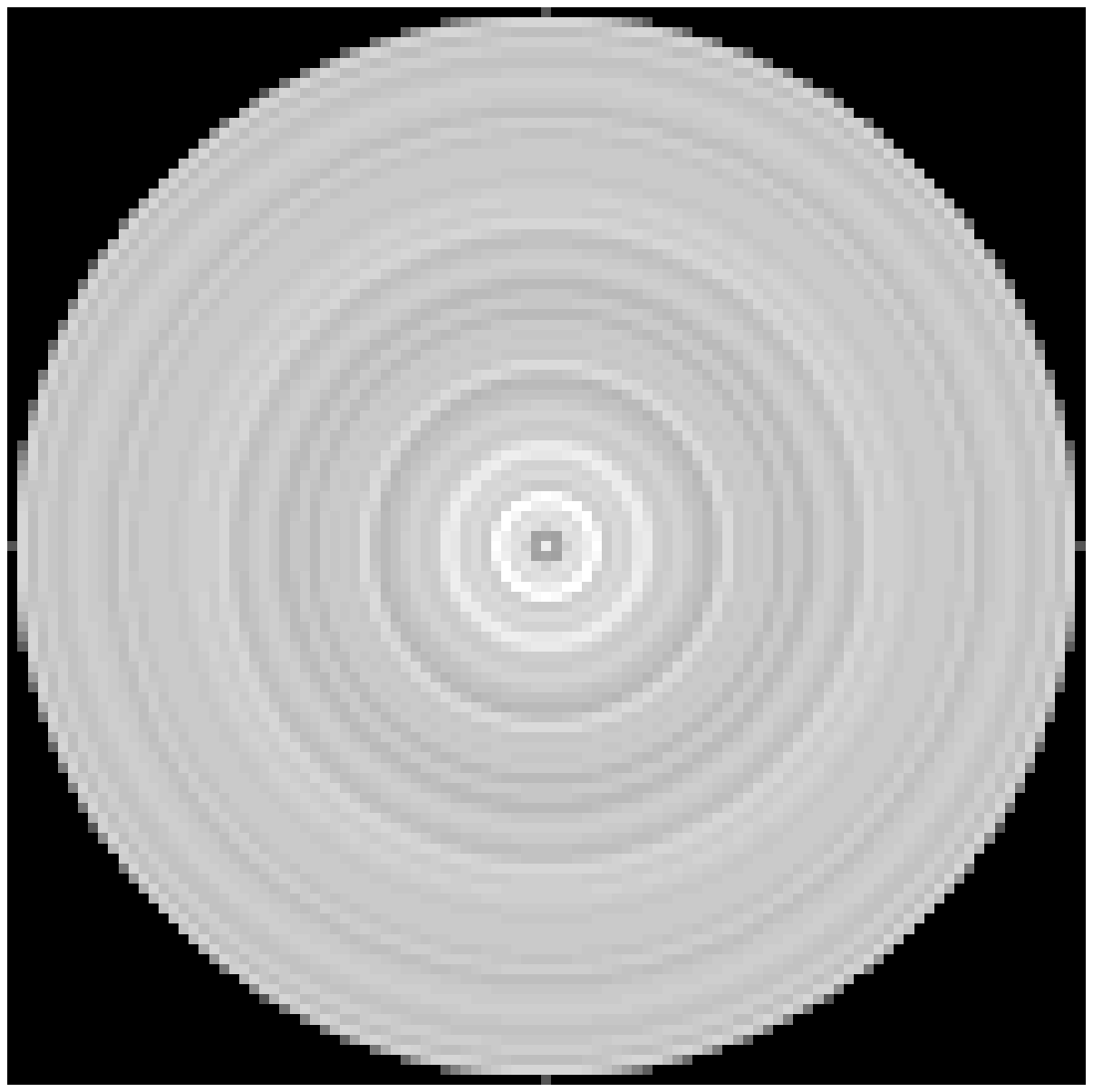}
		\captionsetup{width=.9\linewidth}
		\caption{Rotational averaging of~(\subref{fig:denoise2_shark}). In view of (\subref{fig:denoise3_shark}), the SNR $= 1 / 5$}
		\label{fig:denoise4_shark}
	\end{subfigure}
	\caption{Illustration of the SNR increase attained via rotational averaging. We compare the SNR  between images (\subref{fig:denoise1}), (\subref{fig:denoise2}), (\subref{fig:denoise1_shark}) and (\subref{fig:denoise2_shark}), to their rotational averages, (\subref{fig:denoise3}), (\subref{fig:denoise4}), (\subref{fig:denoise3_shark}) and (\subref{fig:denoise4_shark}), respectively. The noise levels have dropped from SNR of $1/45$ in(\subref{fig:denoise2}) and (\subref{fig:denoise2_shark}) to SNR of $1$ and $1/5$ in the rotational averages of the hedgehog and the shark, respectively. The difference between these two cases lies in the angular spread of the two objects. For a formal definition of SNR see~\eqref{eqn:snr}.}
	\label{fig:denoise}
\end{figure}  

The above demonstration shows how to reduce the problem of centering silhouettes to that of measuring the energy spread of rotationally averaged images. We obtain this measurement by posing a metric in which we can quantify the distance between a given rotationally averaged projection and the ideal centered image, which is a delta-image where the total pixel energy is concentrated in the origin. We expect that, as the CM of an image nears its origin, its rotational average will become closer to the delta-image.

We formally define our surrogate function next, in Section~\ref{subsubsec:surrogate}.

\subsection{Surrogate function for center of mass} \label{subsubsec:surrogate}

In Section~\ref{sec:cm_gm} we discussed the limitations of direct estimation of the CM of projection images \rev{as well as the limitations of the GM as a robust estimator of the CM}. In this section, we present our suggested surrogate function, which supplies a robust estimate of the \rev{CM}.  \rev{We note that, in the following, 
	we assume that in a \rev{noiseless} image the pixel values of the object are strictly larger than those of the background which consists of zero-valued pixels, corresponding to zero mass. This setting is typical in many applications~\cite{frank2006three, thomas2006comparison}.}  

We begin the definition of our surrogate function with some notations. First, we denote the parametric description of the image (the image grid) as the set of pixels $\pset$. We further denote the $m$th ring of pixels around $p \in \pset$ as $A_m(p)$. Formally, we define $A_m(p)$ as
\[ A_m(p) = \left\lbrace  s \in \pset \mid m-1 < d\left(p,s\right) \le m  \right\rbrace  , \quad m=1,\dots, R.  \]
For $m=0$ we set $A_0(p) = \{p\}$. We further define the disk
\[ B_R(x) = \bigcup_{m=0}^{R} A_m(x). \]
Unless stated otherwise, the metric $d$ is the standard Euclidean metric.
\begin{definition}[surrogate function for center of mass] \label{def:surrogate}
	We define the surrogate function for the CM \rev{(sCM)} of an image $I$ as the minimizer of the following sum of absolute deviations:
	\begin{equation} \label{eqn:surrogate}
	\sfun =  \arg \min_{x \in \pset} \sum_{l=0}^R  \left( { \emax - \sum_{m=0}^{l} \sum_{s \in A_m (x)} I(s) }  \right).
	\end{equation}
	Here, \rev{the image $I$ is normalized such that all pixel values are non-negative,} $R$ is an upper bound on, and strictly larger than, the radius of the object we aim to center and $\emax =  \max_{x \in \pset} \sum_{\ell=0}^R \sum_{p \in A_{\ell}(x)} I(p)$ \rev{is, in the case of a noiseless image, the total mass of the object}. 
\end{definition}

\rev{As our goal is to approximate the center of mass of an object, it follows that our estimate $\mu_s$ must rely on the entire mass of the object, and only the mass of this object. The parameter $R$ is determined such that the entire object will fit in a disk of diameter $2R+1$ around the true GM of the object, thereby guaranteeing that $\mu_s$ relies on the entire mass of the object. In Figure~\ref{fig:R} we illustrate the effect of the parameter $R$ on $\mu_s$. To this end, we use an image of a hedgehog, where the diameter of the hedgehog is $112$. When $R<57$ our estimate of $\mu_s$ has a constant change. However, when $R>57$, $\mu_s$ intersects with the GM and remains unchanged as $R$ grows. We discuss formal bounds on $R$ such that the optimum $\mu_s$ is unaffected by nearby objects in Theorem~\ref{thm:relation3}.} 

\begin{figure}
	\captionsetup[subfigure]{width=0.8\textwidth}
	\centering
	\begin{subfigure}[t]{0.4\textwidth}
		\centering
		\includegraphics[width=.65\textwidth]{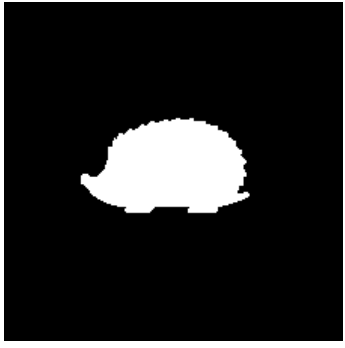}
		\label{fig:hedge_R}
	\end{subfigure}
	\begin{subfigure}[t]{0.4\textwidth}
		\centering
		\includegraphics[width=.8\textwidth]{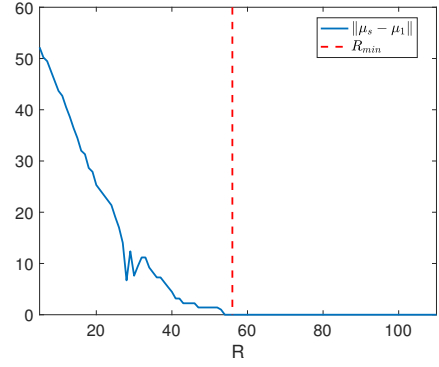}
		\label{fig:R_plot}
	\end{subfigure}
	\caption{\rev{Illustration of  the effect of $R$ on the sCM, $\mu_s$. As $R$ grows beyond $R_{min}$ (the radius of the hedgehog), $\mu_s$ remains stable and intersects with the GM $\mu_1$.}}
	\label{fig:R}
\end{figure}

We now explore the connection between the GM and \rev{the sCM} in the noiseless case.
\begin{theorem} \label{thm:relation}
	Consider a noise-free setting, where image $I$ contains a single object \rev{of non-zero mass}. Then, the GM of~\eqref{eqn:moving_MOM}  \rev{with respect to the composed metric $d(x,y) =  \left\lceil   \norm{x-y} \right\rceil$ (where $\left\lceil  \cdot  \right\rceil$ is the operation of rounding up) 
		coincides with a local minimum of the 
		landscape of \rev{the sCM},
		\begin{equation} \label{eqn:proof}
		L_I(x) = \sum_{l=0}^R \left( { \emax - \sum_{m=0}^{l} \sum_{s \in A_m (x)} I(s) } \right).
		\end{equation}}
\end{theorem}
\begin{proof}
	Recall that all pixel values of the image $I$ are non-negative. Then, for any pixel $x$ where $B_R(x)$ contains the entire object, \eqref{eqn:proof} becomes
	\begin{equation} \label{eqn:surrogate_landscape}
	L_I(x) = \sum_{l=0}^R \left( \emax  - \sum_{m=0}^{l} \sum_{s \in A_m (x)} I(s) \right) = (R+1) \ \emax - \sum_{l=0}^R \sum_{m=0}^{l} \sum_{s \in A_m (x)} I(s) ,
	\end{equation}
	We further note that
	\begin{eqnarray}\label{eqn:L_simplified2}
	\sum_{l=0}^R \sum_{m=0}^{l} \sum_{s \in A_m (x)} I(s) & = & (R+1)I(x) +  R \sum_{s \in A_1 (x)} I(s) + (R-1) \sum_{s \in A_2 (x)} I(s) + \cdots + \sum_{s \in A_R (x)} I(s) \\ \nonumber
	& = & \sum_{s \in B_R(x)} (R+1 - \left\lceil   \norm{x-s} \right\rceil ) I(s), 
	\end{eqnarray}
	We denote the GM as $\mu_1$.  By definition, for any point $x$ and, specifically, for any point $x \neq \mu_1$ where $B_R(x)$ contains the entire object
	\begin{equation}
	\sum_{s \in \pset}   \left\lceil   \norm{\mu_1-s} \right\rceil I(s) < \sum_{s \in \pset}   \left\lceil   \norm{x-s} \right\rceil I(s).
	\end{equation}
	Equivalently, this can be expressed as
	\begin{equation}
	\sum_{s \in B_R(\mu_1)}   \left\lceil   \norm{\mu_1-s} \right\rceil I(s) + \sum_{s \notin B_R(\mu_1)}   \left\lceil   \norm{\mu_1-s} \right\rceil I(s) < \sum_{s \in B_R(x)}   \left\lceil   \norm{x-s} \right\rceil I(s) + \sum_{s \notin B_R(x)}   \left\lceil   \norm{x-s} \right\rceil I(s).
	\end{equation}
	We note that, by the construction of the disk $B_R(x)$, every pixel $s \notin B_R(x)$ must equal zero. The same is true for every pixel  $s \notin B_R(\mu_1)$. Therefore,
	\begin{equation}
	\sum_{s \in B_R(\mu_1)}   \left\lceil   \norm{\mu_1-s} \right\rceil I(s)  < \sum_{s \in B_R(x)}   \left\lceil   \norm{x-s} \right\rceil I(s).
	\end{equation}
	
	Furthermore, as the object is fully contained in $B_R(x)$ and in $B_R(\mu_1)$,
	\begin{equation}
	(R+1)  \sum_{s \in B_R(\mu_1)}  I(s) = (R+1) \sum_{s \in B_R(x)}  I(s).
	\end{equation}
	
	In conclusion,
	\begin{equation} \label{eqn:L_simplified3}
	\sum_{s \in B_R(\mu_1)}   \left\lceil   \norm{\mu_1-s} \right\rceil I(s) - (R+1)  \sum_{s \in B_R(\mu_1)}  I(s)  < \sum_{s \in B_R(x)}   \left\lceil   \norm{x-s} \right\rceil I(s) - (R+1) \sum_{s \in B_R(x)}  I(s).
	\end{equation}
	Thus, \eqref{eqn:L_simplified3} together with~\eqref{eqn:surrogate_landscape} and~\eqref{eqn:L_simplified2} imply the claim.
\end{proof}

Theorem~\ref{thm:relation} states that in the regime where the GM provides a valid estimate of the CM, that is, a clean image containing a single object, \rev{the sCM} is guaranteed to reach a local minimum at the GM. \rev{This will become a global minimum} under the condition that for any disk $B_m(x)$, where $0 \le m \le R$, the intensity in the outer $R-m$ rings is limited relative to the internal $m$ rings \rev{(or, i}n other words, the intensity values of the object do not increase rapidly as we move away from its GM\rev{)}. We formally pose this condition in the following Theorem.

\begin{theorem} \label{thm:relation2}
	Consider a noise-free setting, where image $I$ consists of a single object. Assume that for any $x \in \pset$ the object satisfies
	\begin{equation} \label{eqn:thm_cond}
	\sum_{s \in \pset} \left\lceil   \norm{\mu_1-s} \right\rceil   I(s)  < \sum_{s \in B_R(x)} \left\lceil   \norm{x-s} \right\rceil   I(s) + (R+1) \sum_{s \in \pset \setminus B_R(x)} I(s), 
	\end{equation}
	where $R$ is an upper bound on, and strictly larger than, the radius of the object we aim to center. 
	Then, \rev{the GM of~\eqref{eqn:moving_MOM}  with respect to the composed metric $d(x,y) =  \left\lceil   \norm{x-y} \right\rceil$		coincides with the 
		global minimum of {the landscape of the sCM}~\eqref{eqn:surrogate}. }
\end{theorem}
\begin{proof}
	By~\eqref{eqn:thm_cond}, we have
	\begin{equation}
	\sum_{s \in \pset} \left\lceil   \norm{\mu_1-s} \right\rceil   I(s)  < \sum_{s \in B_R(x)} \left\lceil   \norm{x-s} \right\rceil   I(s) + (R+1) \sum_{s \in \pset} I(s) - (R+1) \sum_{s \in B_R(x)} I(s),
	\end{equation}
	that is
	\begin{equation}
	\sum_{s \in \pset} \left\lceil   \norm{\mu_1-s} \right\rceil   I(s) - (R+1) \sum_{s \in \pset} I(s)  < \sum_{s \in B_R(x)} \left\lceil   \norm{x-s} \right\rceil   I(s)  - (R+1) \sum_{s \in B_R(x)} I(s).
	\end{equation}
	Additionally, any pixel $x \in \pset$ which has value other than zero must be in the disk $B_R(\mu_1)$. Therefore,
	\begin{equation}
	\sum_{s \in B_R(\mu_1)} \left\lceil   \norm{\mu_1-s} \right\rceil   I(s) - (R+1) \sum_{s \in B_R(\mu_1)} I(s)  < \sum_{s \in B_R(x)} \left\lceil   \norm{x-s} \right\rceil   I(s)  - (R+1) \sum_{s \in B_R(x)} I(s).
	\end{equation}
	Recall~\eqref{eqn:surrogate_landscape} and~\eqref{eqn:L_simplified2}, to say, the GM is the global minimum of the landscape of \rev{the sCM}.
\end{proof}

So far we have seen that in the regime where the GM provides a reliable estimate of the CM, \rev{the sCM} will identify the GM as either a local or a global minimum. \rev{We will now consider two cases where the GM fails to identify the CM, whereas the sCM retains its reliability.}

\rev{In the first case, w}e consider image $I$ as a sum of two images, $I_1$ and $I_2$. Image $I_1$ fully contains a single object supported on the set of pixels $\pset_1$. \rev{Image $I_2$ includes a constant background, that is $I_2(p) = \alpha>0$ at any pixel $p$.}
\rev{
	\begin{lemma} \label{lemma:fixedBackgroound}
		Assuming the number of pixels at each disk $B_r(p)$, $r\le R$ is fixed for any $p\in \pset$. Then, the sCM of $I_1$, which agrees with the assumptions of Theorem~\ref{thm:relation}, is preserved for $I=I_1 + I_2$, with $I_2(p) = \alpha>0$ at any pixel $p$.
	\end{lemma}
	\begin{proof}
		Since the value $L_{I_2}(p)$ of the cost function~\eqref{eqn:proof} is constant for all $p \in \pset$ and the landscape of~\eqref{eqn:proof} is linear, that is $L_I(p) = L_{I_1}(p) + L_{I_2}(p)$, it follows that $\arg \underset{p\in \pset}{\min} L_I(p) = \arg  \underset{p\in \pset}{\min} L_{I_1}(p)$.
	\end{proof}
}
\rev{A consequence of Lemma~\ref{lemma:fixedBackgroound} is that the the sCM will be uneffected by an addition of some constant to the image $I$. It follows that iid noise with zero mean and iid noise with positive mean will have the same effect on the sCM. This is not the case for the GM, as was illustrated in Figure~\ref{fig:noisy_uniform_row}. To see this, we turn to}
the landscape of the GM 
\begin{equation}\label{eqn:gm_landscape}
\rev{L_I^G}(x) = \sum_{p \in \pset} I(p) d(p, x) .
\end{equation}
\rev{This landscape is also linear, that is, ${L_I^G}(x) = L^G_{I_1}(x) + L^G_{I_2}(x)$.}
As the two \rev{images $I_1$ and $I_2$} have distinct GMs, it follows that the overall GM of $I$ \rev{may} deviate from the GM of the object in $I_1$. \rev{In particular, since we assume $I_2$ is constant, its GM appears at (or very close to) the central pixel, that is, the origin. This is usually not the case for image $I_1$ (otherwise, centering would be unnecessary). Therefore, there is no result analogous to Lemma~\ref{lemma:fixedBackgroound} for the GM}.

\rev{In the second case where the GM fails to provide a reliable estimate of the CM of an object, images $I_1$ contains a single object supported on the set of pixels $\pset_1$ and image $I_2$ contains a single partial object supported on the set of pixels $\pset_2$. As the objects are distinct, it follows that $\pset_1 \cap \pset_2 = \emptyset$. Recall that our aim is to identify the center of the object supported on $\pset_1$. As before, since the objects have distinct GMs, the overall GM of $I$ is not guaranteed to coincide with GM of the object in $I_1$. Indeed, such an example is presented in Figure~\ref{fig:geomean}, where the minimum of $L_I^G$ migrates towards the area between both objects, since the distance to all non-zero pixels is reduced there. This is not the case for the sCM. We } analyze the performance of \rev{the sCM in the following theorem.}
\begin{theorem} \label{thm:relation3}
	Consider a noise-free setting, where image $I = I_1+I_2$ is a sum of two images. Image $I_1$ fully contains a single object supported on the set of pixels $\pset_1$. Image $I_2$ contains a single partial object supported on the set of pixels $\pset_2$.  
	Then, the minimizer of \rev{the sCM} will remain unchanged so long the following conditions are satisfied:
	\begin{enumerate}
		\item \textit{distance assumption:} The GM of $I_1$ has a distance of at least $R$  from any non-zero pixels in $I_2$.
		\item \textit{intensity assumption:}  $$\max_{x \in \pset} \sum_{s \in B_R(x)} I(s) = \sum_{s \in \pset_1} I_1(s).$$
	\end{enumerate}
\end{theorem}
\begin{proof}
	We divide the landscape~\eqref{eqn:surrogate_landscape} into two parts, 
	\begin{equation} \label{eqn:decompose_L}
	L_I(x) = \underbrace{(R+1) \ \emax - \sum_{l=0}^{R} \sum_{m=0}^{l} \sum_{s \in A_m(x)} I_1(s)}_{L_{I_1}(x)}
	- \underbrace{ \sum_{l=0}^{R} \sum_{m=0}^{l} \sum_{s \in A_m(x)} I_2(s)}_{L_{-}(x)} .
	\end{equation} 
	By Theorem~\ref{thm:relation}, the landscape $L_{I_1}(x)$ reaches a minimum at the GM of the object, $x_1^\ast$. Furthermore, by the distance assumption, $L_{-}(x_1^\ast) = 0$. 
	That is, for any pixel $x$ such that the object depicted in $I_2$ is outside the disk $B_R(x)$,  
	\begin{equation} \label{eqn:same_min}
	L_{I}(x) = L_{I_1}(x).
	\end{equation}
	Since $x_1^\ast$ is a (local) minimum of $L_{I_1}(x)$ it must also be  a (local) minimum of   $L_I(x)$. 
	
	Additionally, when $x_1^\ast$ is the global minimum of $L_{I_1}(x)$, it is also the global minimum of $L_I(x)$.  To show this, we assume in contradiction that there exists a pixel $p$ such that $L_{I}(p)<L_{I}(x_1^\ast)$. Then, by~\eqref{eqn:same_min}, we have
	\begin{multline*}
	(R+1) \ \emax - \sum_{l=0}^{R} \sum_{m=0}^{l} \sum_{s \in A_m(p)} I_1(s)
	- \sum_{l=0}^{R} \sum_{m=0}^{l} \sum_{s \in A_m(p)} I_2(s) <  \\
	(R+1) \ \emax - \sum_{l=0}^{R} \sum_{m=0}^{l} \sum_{s \in A_m(x_1^\ast)} I_1(s)
	- \sum_{l=0}^{R} \sum_{m=0}^{l} \sum_{s \in A_m(x_1^\ast)} I_2(s) .
	\end{multline*}
	However, $\sum_{l=0}^{R} \sum_{m=0}^{l} \sum_{s \in A_m(x_1^\ast)} I_2(s) = 0$ and so,
	\begin{equation} \label{eqn:contradict}
	\sum_{l=0}^{R}  \sum_{m=0}^{l}  \sum_{s \in A_m(p)} \left( I_1(s) + I_2(s) \right) > \sum_{l=0}^{R}  \sum_{m=0}^{l}  \sum_{s \in A_m(x_1^\ast)} I_1(s) .
	\end{equation}
	In conclusion, \eqref{eqn:contradict} contradicts the intensity assumption, which proves that no such pixel $p$ exists. 
\end{proof}

Contrary to the GM, \rev{the sCM}~\eqref{eqn:surrogate} uses an upper bound on the radius of the object as additional information. It is natural to examine a similar restriction on the cost function~\eqref{eqn:gm_landscape} of the GM. We therefore define the landscape of the \textit{local geometric median} as
\begin{equation} 	\label{eqn:gm_landscape_mod}
L^R_{I}(x) = \sum_{s \in B_R(x)} I(s) d(s, x) .
\end{equation}
Both the \rev{sCM}~\eqref{eqn:surrogate} and the local GM~\eqref{eqn:gm_landscape_mod} provide an indication on the object's GM. Additionally, the compact support of both landscapes ensures that distant partial objects will not affect the center estimation at the vicinity of an object. Nevertheless, these methods differ in other aspects. For example, the landscape of $L_{I_1}^R(x)$ reaches a local minimum at the GM and grows with distance from the GM. It then falls to zero, and remain zero at pixels $x$ where $B_R(x)$ does not contain any object. The low values outside an object can create spurious minima at pixels residing between multiple objects. In comparison, at these locations \rev{the} landscape \rev{of the sCM} will remain high due to the comparison with $\emax$ and our regularization, which assigns higher weights to the inner rings. A visual comparison between the GM, local GM, and \rev{the sCM} in the case of single and multiple objects is given in Figure~\ref{fig:cases}. This comparison summarizes the above analysis.

\begin{figure}[t]
	\begin{center}
		\begin{subfigure}[t]{0.185\textwidth}
			\centering
			\includegraphics[width=\textwidth]{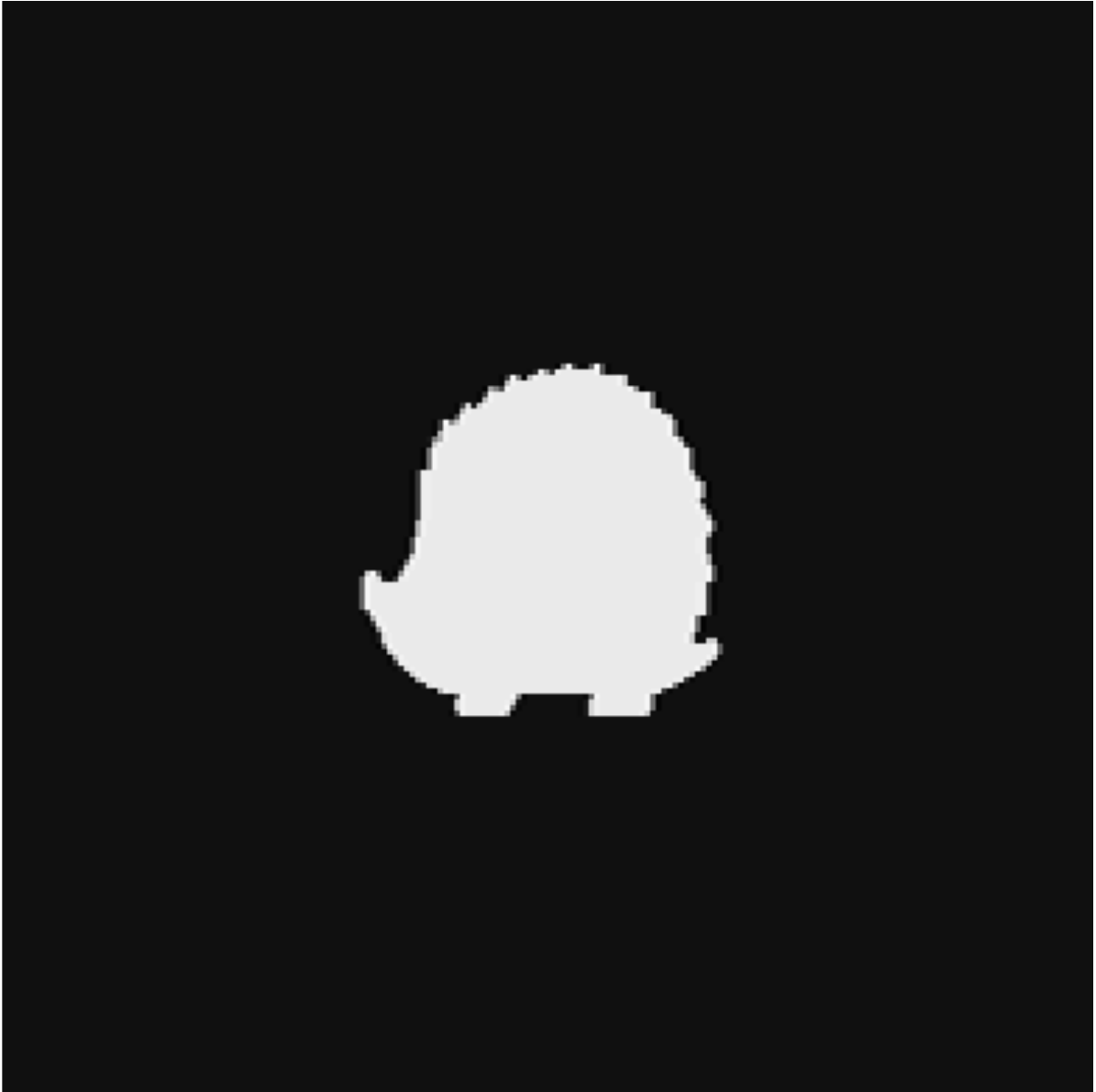}
		\end{subfigure}  \quad
		\begin{subfigure}[t]{0.22\textwidth}
			\centering
			\includegraphics[width=\textwidth]{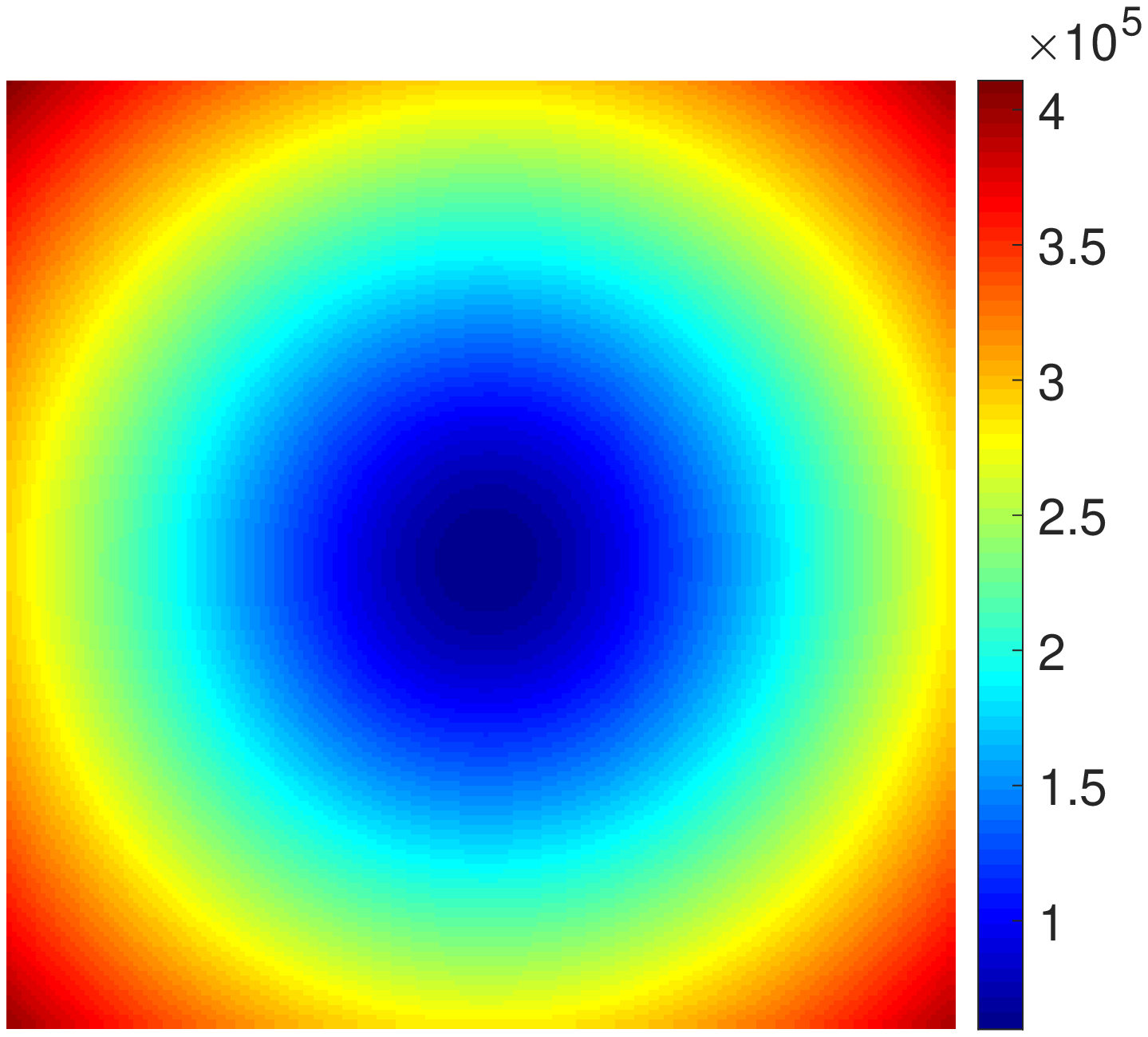}
		\end{subfigure}  \quad
		\begin{subfigure}[t]{0.22\textwidth}
			\centering
			\includegraphics[width=\textwidth]{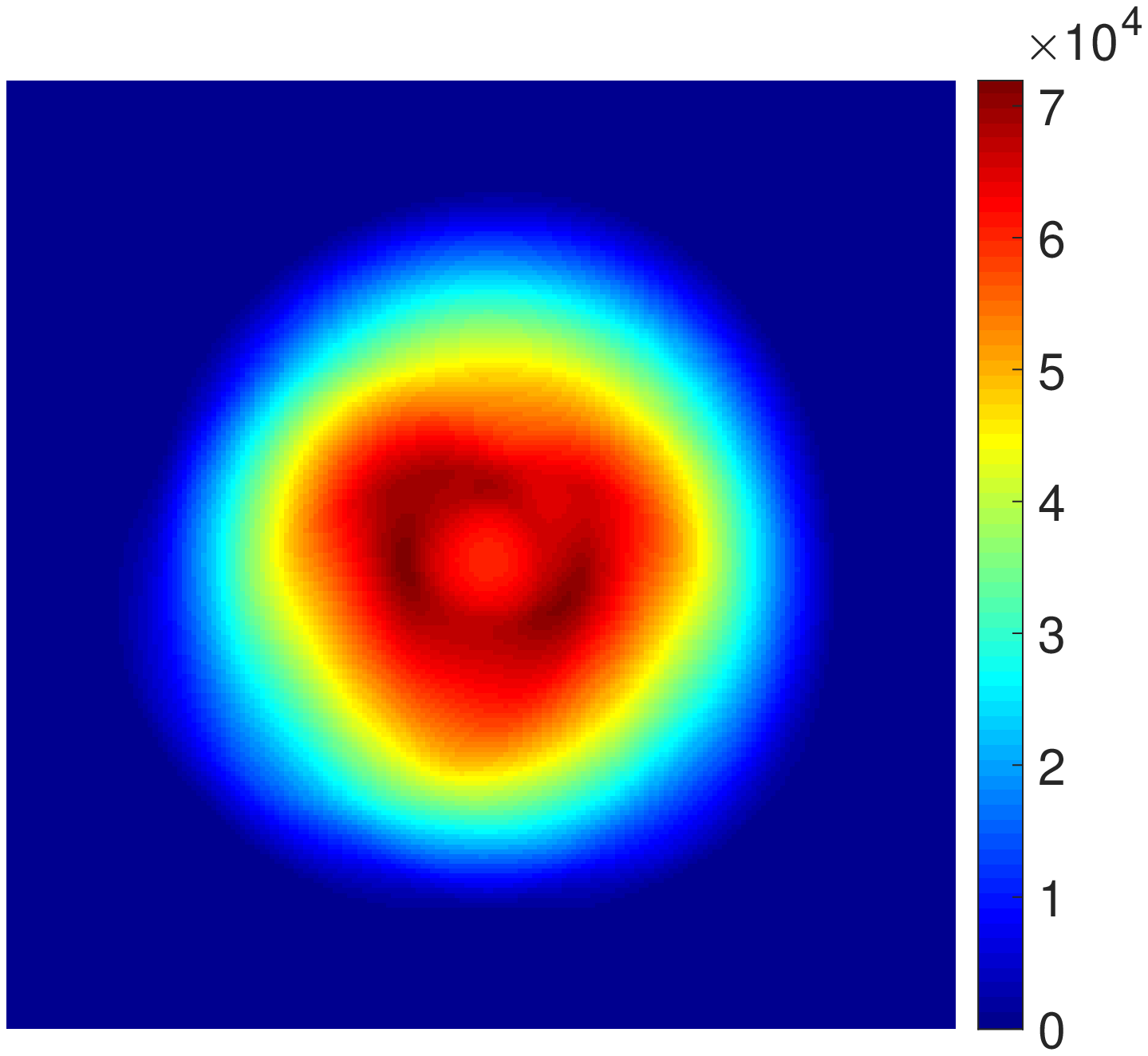}
		\end{subfigure}  \quad
		\begin{subfigure}[t]{0.23\textwidth}
			\centering
			\includegraphics[width=\textwidth]{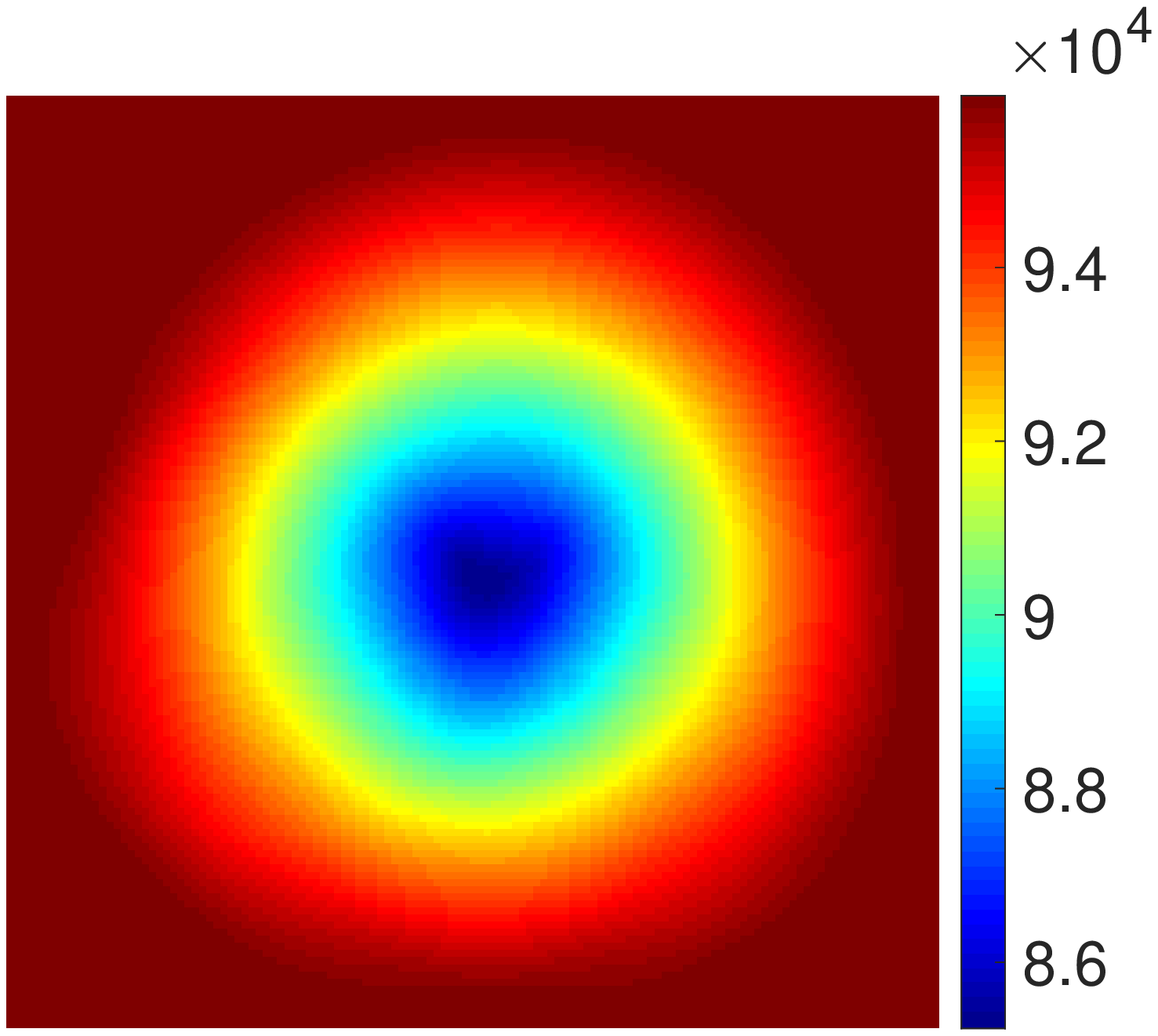}
		\end{subfigure}\\
		\begin{subfigure}[t]{0.185\textwidth}
			\centering
			\includegraphics[width=\textwidth]{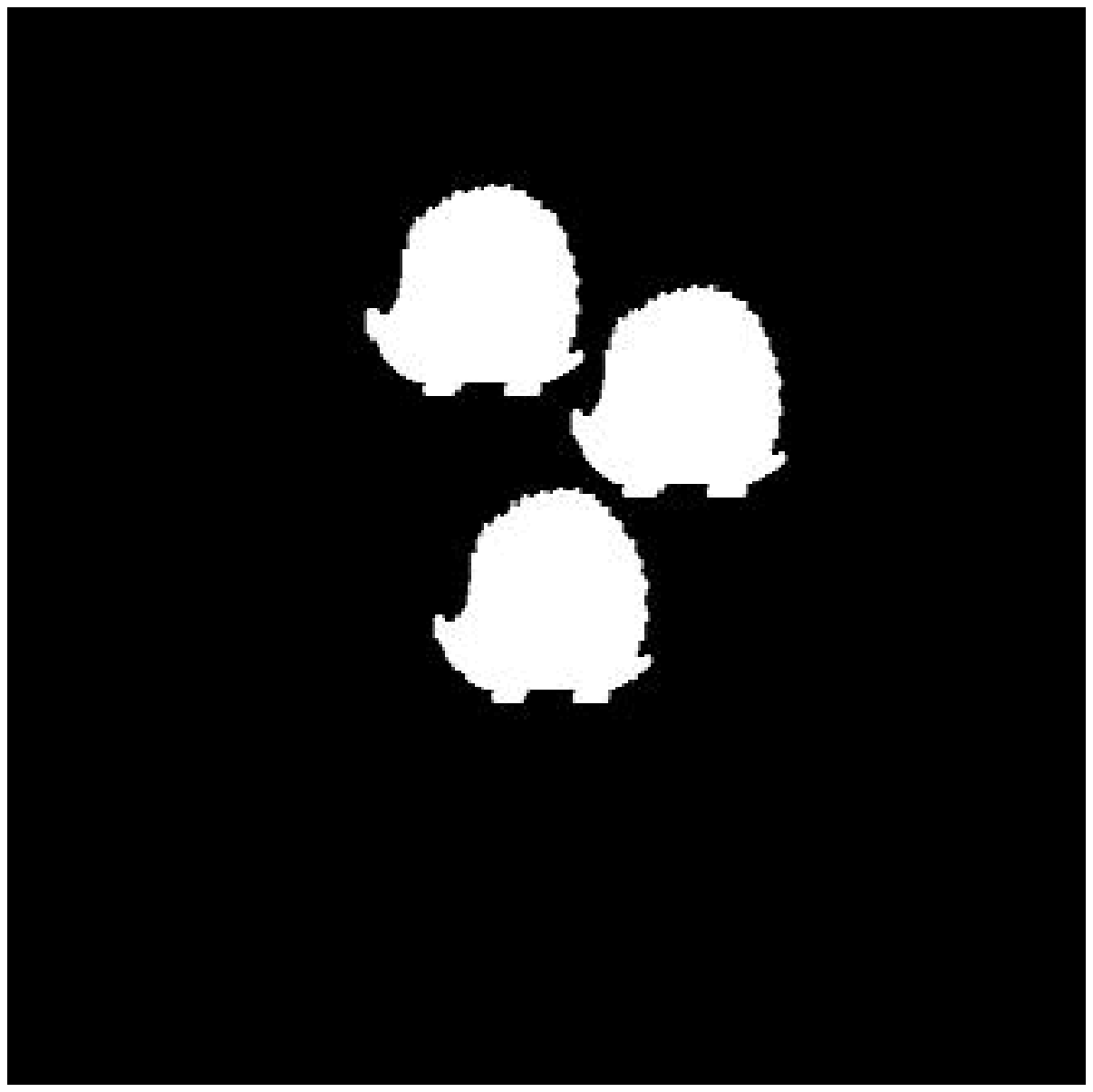}
		\end{subfigure}   \quad
		\begin{subfigure}[t]{0.22\textwidth}
			\centering
			\includegraphics[width=\textwidth]{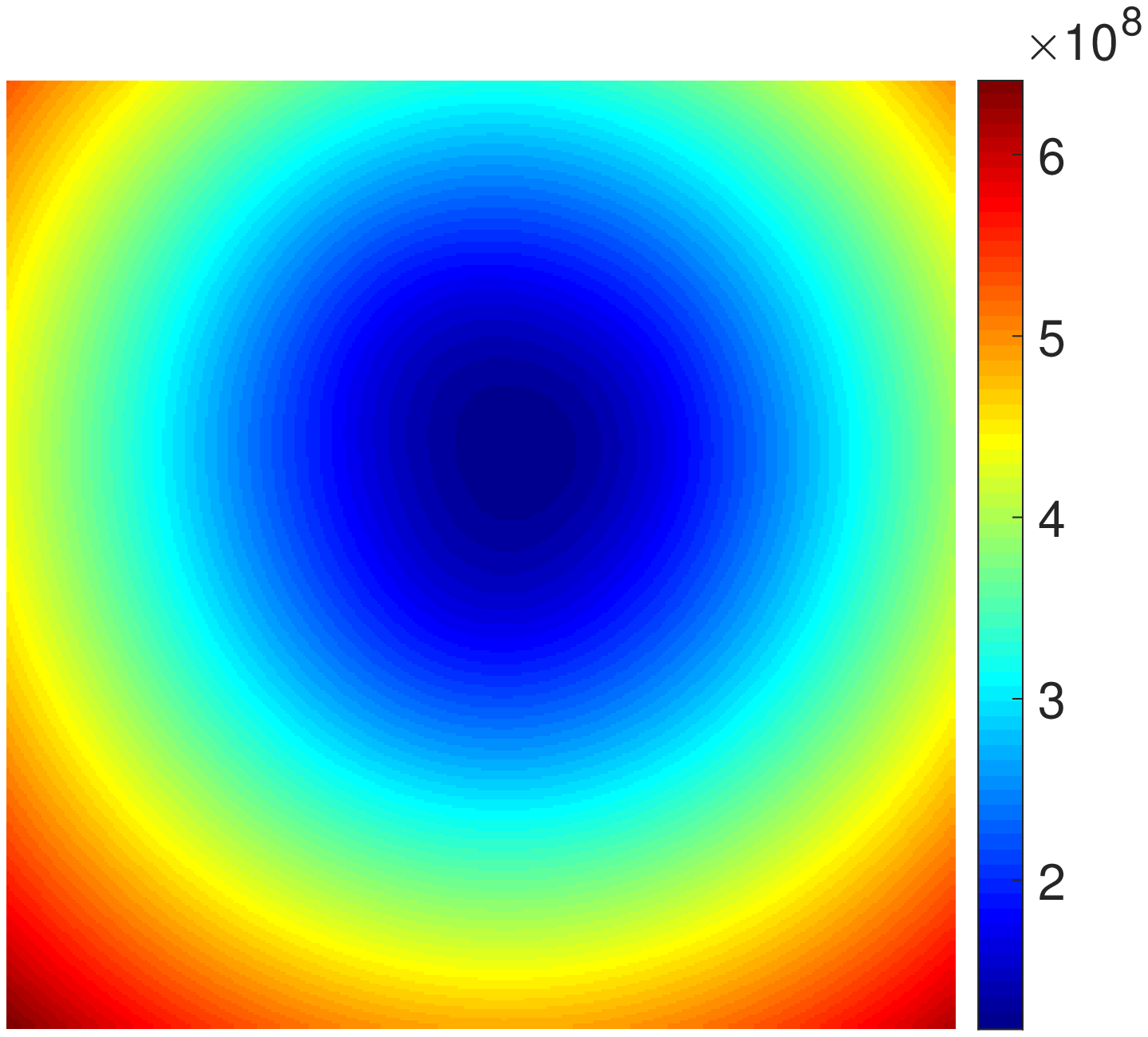}
		\end{subfigure}  \quad
		\begin{subfigure}[t]{0.22\textwidth}
			\centering
			\includegraphics[width=\textwidth]{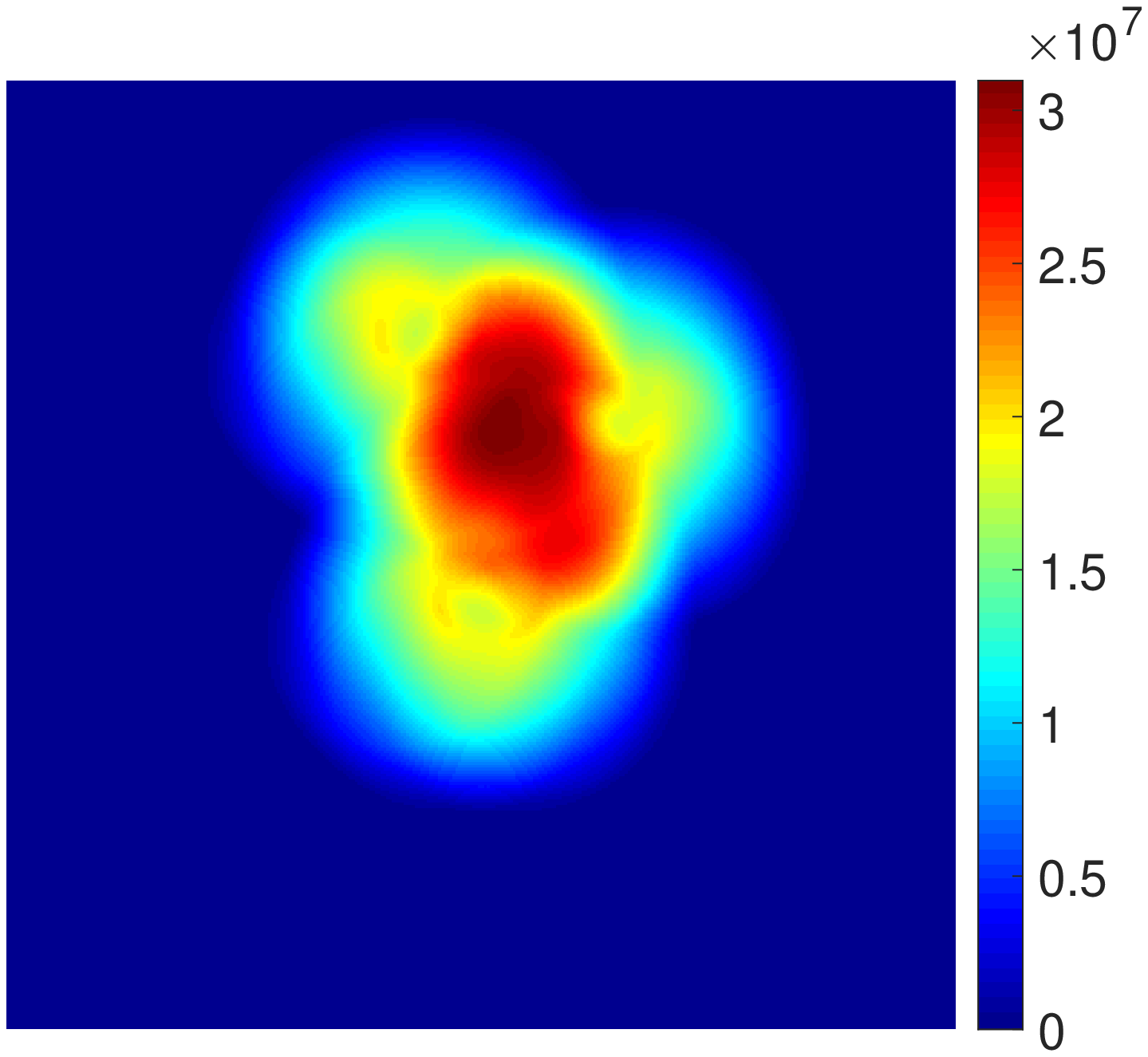}
		\end{subfigure}   \quad
		\begin{subfigure}[t]{0.23\textwidth}
			\centering
			\includegraphics[width=\textwidth]{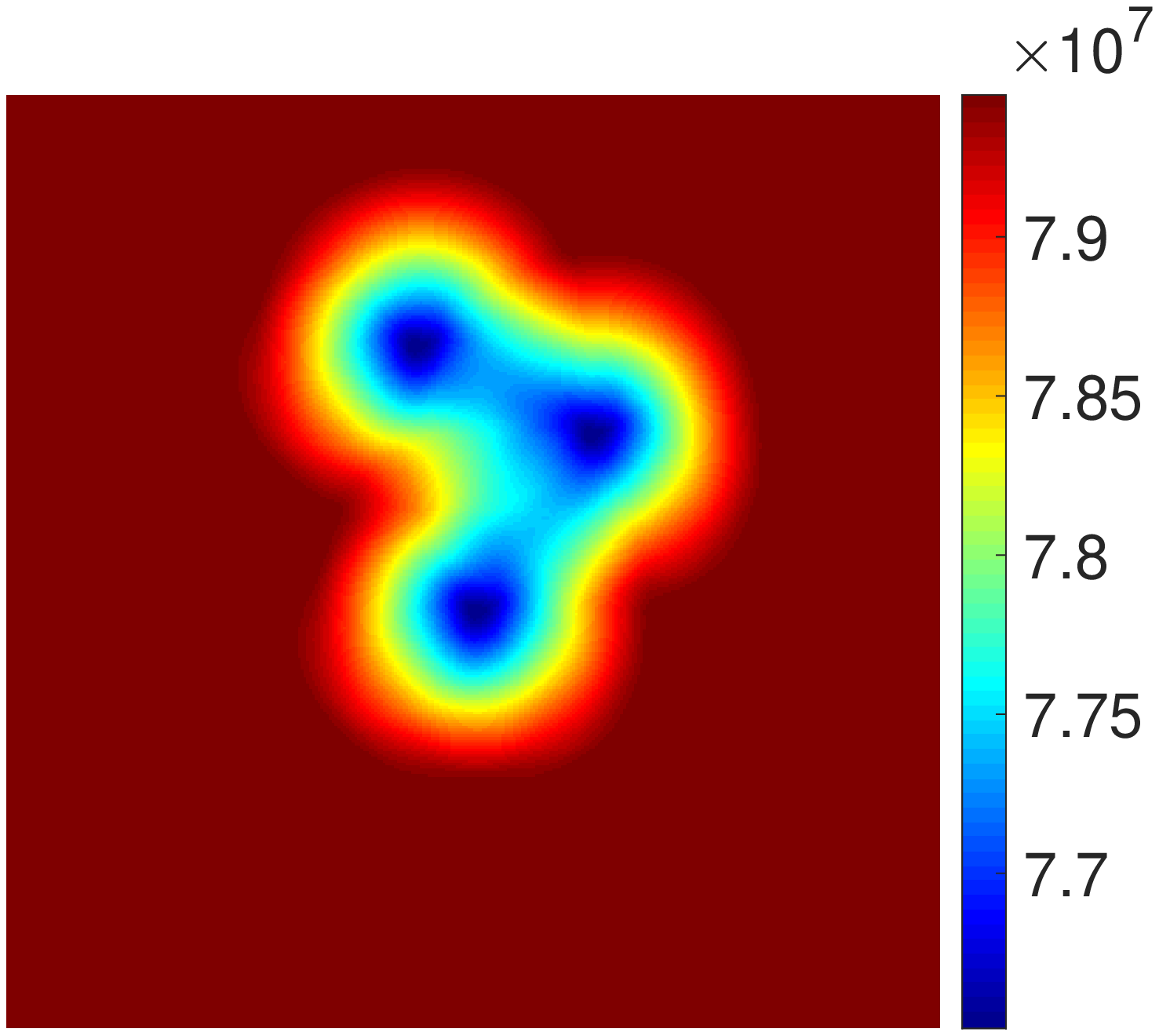}
		\end{subfigure}
		\caption{Comparison between the landscape functions of the GM, local GM, and \rev{the sCM} for images of a single and multiple hedgehogs. The leftmost column contains the hedgehog images. The \rev{center-left} column contains at each pixel $x$ the sum~\eqref{eqn:gm_landscape}. The \rev{center-right} column contains at each pixel $x$ the sum~\eqref{eqn:gm_landscape_mod}. The rightmost column contains the landscape of \rev{the sCM}, that is, for each pixel $x$ the sum~\eqref{eqn:proof}}
		\label{fig:cases}
	\end{center}
\end{figure}

\subsection{The connection to  Earth Mover's Distance} \label{subsec:EMD}

To conclude this section, we present the connection between \rev{the sCM} and the Earth Mover's Distance (EMD). We have previously alluded to this connection in Section~\ref{app:informal_description} when discussing the motivation for our method. Here we present this connection in more detail. We start with \rev{ the definition} of EMD\rev{, see e.g.,~\cite{rubner2000earth}.}\rev{
	\begin{definition}[Earth Mover's Distance]
		The EMD between two images $I_{1}$ and $I_{2}$ given a joint parameterization is defined as
		\begin{align} \label{eqn:EMD_def}
		d_{\text{EMD}}(I_{1},I_{2}) =  \min_{ \{f_{i,j}\}} \frac{\sum_{i,j \in \mathcal{P}} f_{i,j} d_{i,j}}{\sum_{i,j \in \mathcal{P}} f_{i,j}}, \quad \text{s.t.} \quad & \sum_{j} f_{i,j} \le I_{1}(p_i),\\ \nonumber
		& \sum_{i} f_{i,j} \le I_{2}(p_j),\\
		& \sum_{i,j} f_{i,j} = \min \left( \sum_{p \in \pset} I_{1}(p), \sum_{p \in \pset} I_{2}(p) \right), \quad  f_{i,j}\ge 0. \nonumber
		\end{align}
		Here $d_{i,j}=d(p_i,p_j)$ is the ground distance and $\{f_{i,j}\}$ denotes the flows from $I_{1}$ to $I_{2}$.
\end{definition} }

A key ingredient is what we define as a \rev{delta-image.
	\begin{definition}
		Given an image $I$, the associated delta-image $\delta_p(I)$ at the pixel $p$, is the image with the same size as $I$, defined as
		\[ \delta_p\left( I \right)(s) =
		\left\{
		\begin{array}{ll}
		0  & \mbox{if } s \neq p \\
		\sum_j \abs{I(j)} & \mbox{if } s = p .
		\end{array}
		\right. \]
	\end{definition}
	Next, we obtain the EMD of an image with its delta-image.}
\begin{lemma} \label{lemma:EMD}
	Denote by $E = \sum_{s \in \pset} \rev{\abs{I(s)}}$ the total sum of pixel intensities of an image $I$. Then, 
	\begin{equation}
	d_{\text{EMD}}(I, \rev{\delta_p\left(I\right) \left(c\right)}) =  \sum_{s \in \pset}   \rev{I(s)}  d(s,p)/E, 												
	\end{equation}
	\rev{where $c$ is the central pixel (origin) of the image $I$.}
\end{lemma}
\begin{proof}
	By the definition of \rev{$\delta_p\left(I\right)\left(c\right)$}, we get that, of the flow $f_{i,j}$ between \rev{$I$ and $\delta_p\left(I\right)\left(c\right)$}, only $f_{i,j_p}$ is nonzero, where $j_p$ is the index associated with the pixel $p$. Consider $s_i$ to be associated with $i$, the constraints on the flow imply that \rev{$f_{i,j_p} = I(s_i)$}. Then, the lemma follows by the EMD definition~\eqref{eqn:EMD_def}.
\end{proof}

Focusing on the case of a single object and the composed metric $d(x,y)= \left\lceil   \norm{x-s} \right\rceil$ \rev{as the ground distance}, we combine Theorem~\ref{thm:relation2} and Lemma~\ref{lemma:EMD} to deduce:
\begin{corollary}
	Let \rev{$I$ be an image that fully contains a single, noise-free object. Assume that $p_1$ is the center pixel as determined by \rev{the sCM}~\eqref{eqn:surrogate}. Then, under the condition~\eqref{eqn:thm_cond} of Theorem~\ref{thm:relation2}, and for any pixel $p_2 \neq p_1$,}
	\[  d_{\text{EMD}}\rev{(I,\delta_{p_1}\left(I\right)\left(c\right)}) \le d_{\text{EMD}}\rev{(I,\delta_{p_2}\left(I\right)\left(c\right)})  . \]
	In other words, with the above settings, the \rev{sCM} minimizes the EMD  to the delta-image.
\end{corollary}

\section{The algorithm} \label{sec:algorithm}

In this section, we describe the implementation of our suggested approach. For convenience, we present a sequential algorithm, focusing on a single input image. In practice, one can easily apply a concurrent version of the algorithm, and process a stack of multiple input images simultaneously.

\subsection{Surrogate function reformulated}

Let $z_p$ be the rotational summation vector around the pixel $p$. The $m$th entry of this vector is defined as
\[ z_p[m] = \sum_{s \in B_m(p) } I(s)  , \quad m=0,\ldots , R .  \]
Then, the induced vector of cumulative sum is
\[ u_p[m] = \sum_{\ell=0}^{m}  z_p[\ell] , \quad m=0,\ldots , R . \]
Recall that we normalize the image $I$ to ensure that all pixel values are non-negative. Therefore, the total accumulated energy of an image $I$ over the disk  $B_R(p)$ is
\[  u_p[R]  = \sum_{\ell=0}^R \sum_{s \in B_\ell(p)} I(s)   . \]
Denote by $P \subset \pset $ the set of pixels $p$ such that $B_R(p)$ is fully contained in the image $I$. Now, our minimization criterion can be expressed as
\begin{equation} \label{eqn:main_criterion}
\sfun = \arg \min_{p \in P} \rev{ \hat{L}_I (p) , \quad \text{ where } \quad \hat{L}_I (p) = \norm{u_p/\emax - \mathbf{1}}_1 .}
\end{equation}
\rev{Here,} $\mathbf{1} = \left(1,\ldots,1\right)$ is the all ones vector, $\norm{\cdot}_1$  is the standard $\ell_1$ norm and $\emax=\underset{p \in P}{\max}\, u_p[R]$. 

\subsection{Algorithm description}

We note that any object can be contained in an image of size $(2R+1) \times (2R+1)$, where $R$ is an upper bound on the radius of the object and the origin of the image coincides with the center of the object. In our case, the location of the center is initially unknown. As the general location of the object is known, we can take as the initial center $p_0$ a point within the object. This guarantees that our initial center is at some distance smaller than $R$ from the true center. It follows that the object is fully contained within an image $I$ of size $ (4R+1)\times (4R+1)$ around our $p_0$\rev{, and we therefore apply the sCM to this image}.

Given the image $I$, our goal  is to find a pixel $p$ that is best centered with respect to the \rev{object's} CM. We do this by computing the scaled version of \rev{the sCM} for each pixel $p \in P$ and then identifying the minimizer\rev{, as seen in~\eqref{eqn:main_criterion}.}
This process includes cropping images around possible center pixels in $P$ and applying rotational averaging. The averaging is done through expansion over a basis of orthogonal, rotationally symmetric functions, and is presented in detail in Section~\ref{subsec:PSWF}. Once the rotational averaging has produced the vector $z_p$ for each $p \in P$, we compute the cumulative sums $u_p$, $\emax$ and, finally, the landscape of\rev{~\eqref{eqn:main_criterion}.} 

We summarize the above process with pseudo-code given in Algorithm~\ref{alg:anchor}. In Line~\ref{alg2:line1} we set the initial grid of potential centers $P$. The search over these centers starts with the main loop in Line~\ref{alg2:line2}.  Rotational averaging is done on Line~\ref{alg2:line4} and $\hat{L}_I (\cdot )$ is calculated on Line~\ref{alg2:line5}. After a search over all potential centers, the point that minimizes \rev{the sCM} is chosen, and the algorithm terminates.

\begin{algorithm}
	\caption{Robust Translational Centering} \label{alg:anchor}
	\begin{algorithmic}[1] \label{alg:nu}
		\REQUIRE  A (non-negative) image $I$, and \sout{an estimated} the particle radius $R$  
		\ENSURE  The chosen center $\mu_s$
		\STATE Define the set $P$ of possible centers. \label{alg2:line1}
		\FORALL{ $p \in P$} \label{alg2:line2}
		\STATE $I_p \gets$ Crop $I$  around the pixel $p$ according to $R$  \label{alg2:line3}
		\STATE $u_p \gets$ Apply rotational averaging over $I_p$   \label{alg2:line4}
		\STATE $\hat{L}_I (p) \gets$ Calculate and store, see~\eqref{eqn:main_criterion} \label{alg2:line5}
		\ENDFOR  \label{alg2:endfor}
		\STATE $\mu_s \gets \arg \min_{p \in P} \hat{L}_I (p) $ 
		\STATE  Return $\mu_s$
	\end{algorithmic}
\end{algorithm}

\subsection{Rotational averaging via steerable basis functions} \label{subsec:PSWF}

When computing a rotational average of an image, it is natural to convert the image into polar coordinates. Averaging directly in Cartesian coordinates will produce an inexact result due to the discretization of the image. Therefore, in our implementation, the image is decomposed into polar-separable building blocks (that is, basis functions of the form $f ( \rho, \theta) = g(\rho) \Phi(\theta)$,  where $\rho$ and $\theta$ are the radial and angular coordinates, respectively). The rotational averaging can now be attained directly from the expansion coefficients that are associated with the zero-frequency of  the angular part, $\Phi$. Once an average is computed, the sum $z_p$ is easily recovered.

The above radial coordinate system is supported by many possible bases. One popular example is  the 2-D prolates spheroidal wave functions (PSWFs). PSWFs have been known for more than a century, but it was not until the 1960s that they received a comprehensive analysis in several seminal papers~\cite{landau1961prolate, landau1962prolate, slepian1964prolate, slepian1961prolate}. The PSWFs functions rise as a solution to the optimal concentration problem, which strives to identify the most concentrated function inside a given disk out of all band-limited functions  defined on the plane. This optimality property allows a very efficient representation of band-limited functions mainly supported inside a disk or functions on a disk that are ``almost" band-limited. Since we examine objects which are concentrated on a compact disk and as images are naturally band-limited, we fit precisely the central assumption in the backbone of the PSWF construction. Therefore, following~\cite{landa2017approximation}, we use orthogonal 2-D PSWFs for calculating the rotational average of a projection image $I_p$ by expanding over basis functions of the form 
\begin{equation} \label{eqn:prolate_basis_function}
\Psi_{n,m}^R(\rho,\theta) =  
\begin{cases}
\frac{1}{\sqrt{2\pi} } g^R_{m,n}(\rho)e^{im\phi}  & \rho<R, \\
0 & \text{otherwise}. 
\end{cases}
\end{equation}
The rotational average is then
\begin{equation} \label{eqn:PSWF_expansion}
\mathcal{A}_I (\rho) = \sum_{n=0}^{n_R} \alpha_n  \Psi_{n,0}^R(\rho)  ,
\end{equation}
where\rev{, from orthogonality, the expansion coefficients take the form} $ \alpha_n = \int_{\rho \le R } I(\rho)  \Psi_{n,0}(\rho) d\rho $ \rev{, and are calculated via numerical integration using quadratures~\cite{shkolnisky2007prolate}. T}he radial truncation parameter $n_R$ is determined by the support size $R$ \rev{and its size is of $\mathcal{O}(R)$, see details in~}\cite{landa2017approximation}. Note that by~\eqref{eqn:prolate_basis_function}, the expansion~\eqref{eqn:PSWF_expansion} does not depend on $\theta$ and is therefore rotationally symmetric, as required.

\subsection{Computational complexity}

We begin our complexity analysis with some notations. \rev{We denote the size of the input image $I$ by $N \times N$, the search grid by $G$, and the number of pixels on this grid (potential centers) by $\abs{G}$}. Each potential center is tested using a subimage of size $(2R+1) \times (2R+1)$. The following analysis is  given as a function of the parameters $N$, $R$, and $\abs{G}$.

The main loop in the algorithm (Lines ~\ref{alg2:line2}--\ref{alg2:endfor}) runs over \rev{$G$}, which is made of $\abs{G}$ potential centers. In Line~\ref{alg2:line3}, we crop the image. \rev{This step} is quadratic in $R$ for each potential center\rev{, and} yields $\mathcal{O}(\abs{G} R^2)$ operations. \rev{The} rotational averaging of Line~\ref{alg2:line4} is performed as explained in Section~\ref{subsec:PSWF} which consists of the following three steps. First, the transform that obtains the coefficients of the zero angular frequencies is applied. Next, the resulting low-frequency approximation~\eqref{eqn:PSWF_expansion} is formed, and, lastly,  the  rotationally symmetric image is sampled to obtain a representative radial vector $z_p$. 

The transform includes \rev{reading the cropped} image which is quadratic in $R$ (and linear in the number of pixels of the cropped image), that is $\mathcal{O}(R^2)$. We have $\mathcal{O}(R)$ radial basis functions, all calculated in advance. The coefficients are \rev{computed via} inner product in a total complexity of $\mathcal{O}(\abs{G} R^3)$ operations. Note that in practice, most of the above-mentioned work can be carried out in parallel. \rev{The} evaluation of the low-frequency approximation~\eqref{eqn:PSWF_expansion} is done in an economical way; recall that any such image is rotationally symmetric, and thus we have to evaluate it over a single radial direction\rev{, that is, using} $\mathcal{O}(R)$ coordinates \rev{ per } basis function \rev{or, in total,} $\mathcal{O}(R^2)$ operations. For the full algorithm, this task costs $\mathcal{O}(\abs{G} R^2)$. Finally,  calculating the metric~\eqref{eqn:main_criterion}, we once again act on vectors of size $\mathcal{O}(R)$, which yields $\mathcal{O}(\abs{G} R)$ operations. 

In conclusion, the complexity of the algorithm is of leading order 
\[  \mathcal{O} \left( \abs{G}  R^3 \right) , \] 
where the next leading term is of $\mathcal{O}(\abs{G} R^2)$, and $\abs{G}$ is bounded by $(2R+1)\times (2R+1)$. 

\section{Numerical examples} \label{sec:numerics}

We present a few numerical examples to illustrate the performance of our centering method.  We begin, in Section~\ref{subsec:synthetic}, with an example of a synthetic data image. In Section ~\ref{subsec:real_data} we demonstrate the performance of our method on experimental data obtained from the Electron Microscopy Public Image Archive (EMPIAR) \cite{iudin2016empiar}. Specifically, we test our centering method on datasets of  80S ribosome, $\beta$-galactosidase and TRPV1 macromolecules.

In the synthetic experiments, the noise is additive. We therefore measure the level of noise in a given image $I$ using the following signal-to-noise
ratio (SNR),
\begin{equation}  \label{eqn:snr}
\text{SNR}\left( I \right) = \norm{I_c}^2/\norm{\varepsilon}^2 ,
\end{equation}
where $I_c$ is the clean image and $\varepsilon = I- I_c$ is the noise added to $I_c$.

\subsection{Synthetic demonstration}
\label{subsec:synthetic}

We start with a comparison of our algorithm and the standard method of cross-correlation. While  cross-correlation  is a traditional method, it is still the basis for many of today's algorithms. In this method, we compute the shift $s^\ast$ from the center of the  template as 
\[  s^\ast =   \arg \max_{s = (s_1,s_2)}  \sum_{i,j} I_{R}(i,j) I(i+s_1,j+s_2)  , \]
where $I$ is the test image and $I_R$ is the template. That is, for each test image we find a shift from the template's center by identifying the peak of the cross-correlation function. 

Our experimental setup is as follows; we use a centered image of a hedgehog, including its background, as the basic template $I_R$. The size of this image is $101 \times 101$. The image itself is restricted to a disk, as seen in  Figure~\ref{fig:synthcomparison1}. Next, we create our test image $I$. This image, shown in the leftmost column of Figure~\ref{fig:synthcomparison1}, is of size $211 \times 211$. The reference $I_R$ is embedded within $I$ in an arbitrary deviation of about $10\%$ along the horizontal axis and $15\%$ along the vertical axis from the origin of $I$. Next, we contaminate the test image (left picture in Figure~\ref{fig:synthcomparison1}) \rev{with colored added noise that has a spectrum which decays like $1/\sqrt{1+\rho^2}$ over the radial direction $\rho$. The noise levels vary between SNR of $1/2$ and $1/200$.}  The input to our algorithm includes the noisy test image and the radius of our object (the hedgehog), which is $50$ pixels. The output is the estimated CM for this noisy image. We compare our method to the cross-correlation algorithm. As cross-correlaiton requires a template, and to obtain a  realistic setting, we provide four different references as templates. The first two references, presented in the third and fourth columns of Figure~\ref{fig:synthcomparison1}, consist of noisy versions of the hedgehog image, where the corrupting noise is  Gaussian and iid. One of the noisy references has SNR of approximately $2.5$, and the other has SNR of approximately $1/2$.  The third reference image is a low-pass filtered version of the first reference. The low-pass filter disposes of much of the iid Gaussian noise, yielding a template with SNR of approximately $6$. Our last template is simply a Gaussian with width determined by the known width of the hedgehog. The low-pass template and the Gaussian reference are presented in the rightmost columns of Figure~\ref{fig:synthcomparison1}. Note that, in this example, the template image is rotationally aligned with the test image. Thus, we only need to seek the translation and not to consider rotating the template. In this perspective, the setting is optimal for the cross-correlation method. 

Since the true CM of $I$ is known \rev{and is identical in this test image to the GM}, we compute the sum of pixel deviation of each estimated center from this true center. For robustness, we repeat each test \rev{ten} times and average over the pixel deviations. The results are presented in Figure~\ref{fig:synthcomparison2}. Our centering method, which is represented by a solid blue line, outperforms cross-correlation \rev{methods}, and retains the center throughout the entire range of noise. While our centering method produces better results, we note that it is also slightly higher in runtime. Specifically, the runtime for centering the silhouette is $0.04$ seconds for cross-correlation and $1.65$ seconds for our method to run on a 2.6 GHz Intel Core i7 CPU with four cores and 16 GB of memory.
\begin{figure}[t]
	\centering
	\includegraphics[width=.95\textwidth]{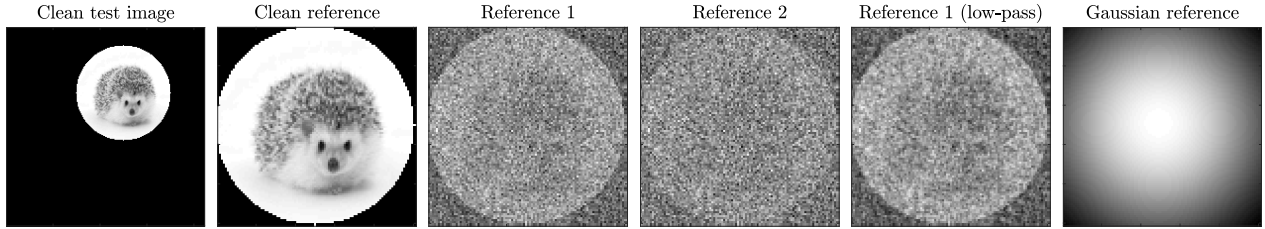}
	\caption{Data for the synthetic comparison. We present, from left to right, the clean non-centered test image, the clean reference image, and the four templates for the cross-correlation method. The references consist of three noisy copies of the centered hedgehog  image contaminated with varying noise levels and a Gaussian reference.}
	\label{fig:synthcomparison1}
\end{figure}

\begin{figure}[t]
	\centering
	\includegraphics[width=0.45\linewidth]{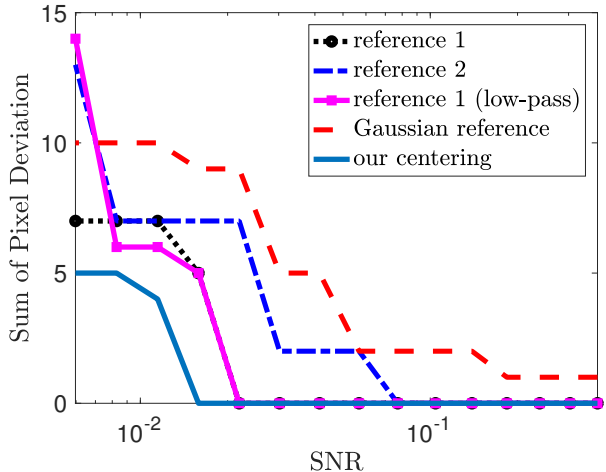}
	\caption{Comparison between performances of cross-correlation and our centering method. The comparison is performed over noisy versions of the hedgehog test image (see Figure~\ref{fig:synthcomparison1}). We measure the error as the sum of pixel deviation (in each axis) from the ground truth shift. As shown, the performance of the cross-correlation is profoundly affected by the quality of the template and is outperformed by our centering method.}
	\label{fig:synthcomparison2}
\end{figure}

Additionally, we use the hedgehog silhouette to demonstrate the connection between the GM (see Definition~\ref{def:GM}) and \rev{the sCM} (see Definition~\ref{def:surrogate}). In Section~\ref{sec:P3DA}, we focused on the noise-free setting. Now, in Figure~\ref{fig:diff_LS} we present  a numerical test that visually shows the effect of noise on the landscape of the local GM~\eqref{eqn:gm_landscape_mod} and \rev{the sCM}~\eqref{eqn:surrogate_landscape}. This comparison illustrates that in the presence of noise, the landscape of the GM dramatically varies (see the middle column). In particular, the region of minimum value stretches across many areas of the image. Thus, choosing the point that leads to a minimum may result in a substantial deviation from the real median of mass. On the other hand, the overall landscape of \rev{the sCM} remains similar to the noiseless case (see the rightmost column), and the point that leads to the minimum provides a reasonable estimation of the true CM. \rev{In Figure~\ref{fig:diff_LS2} we repeat this numerical test for the case of two silhouettes.}

\begin{figure}[!htbp]
	\captionsetup[subfigure]{width=0.85\textwidth}
	\centering
	\begin{subfigure}[t]{0.3\textwidth}
		\centering
		\includegraphics[width=\textwidth]{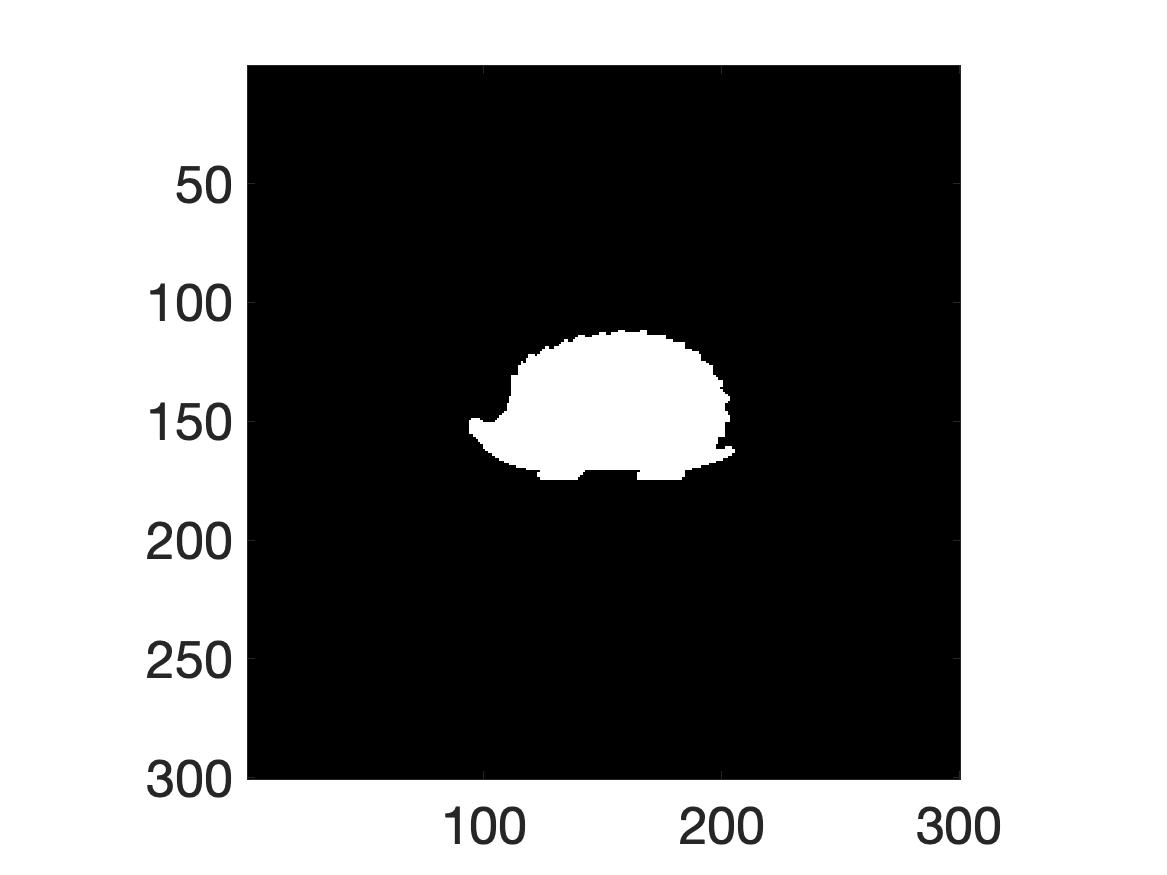}
	\end{subfigure}
	\begin{subfigure}[t]{0.3\textwidth}
		\centering
		\includegraphics[width=\textwidth]{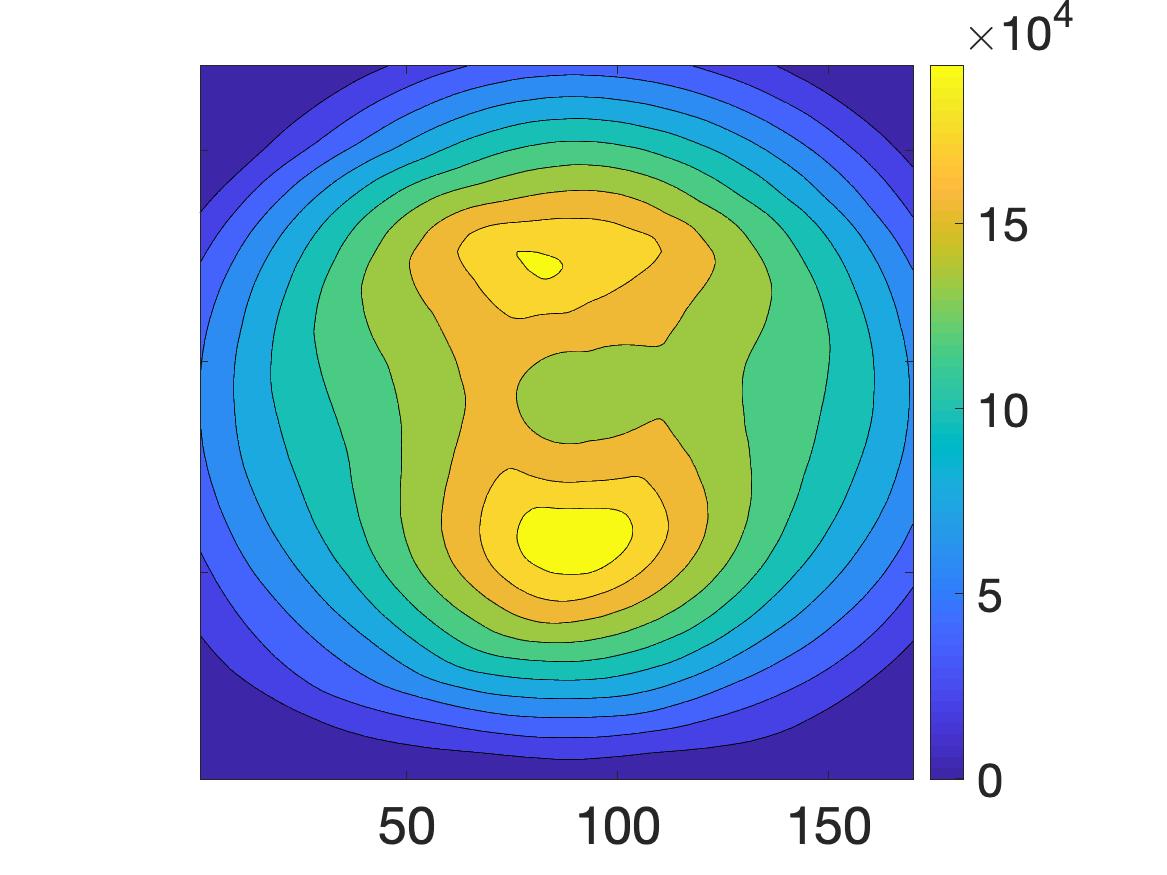}
	\end{subfigure}
	\begin{subfigure}[t]{0.3\textwidth}
		\centering
		\includegraphics[width=\textwidth]{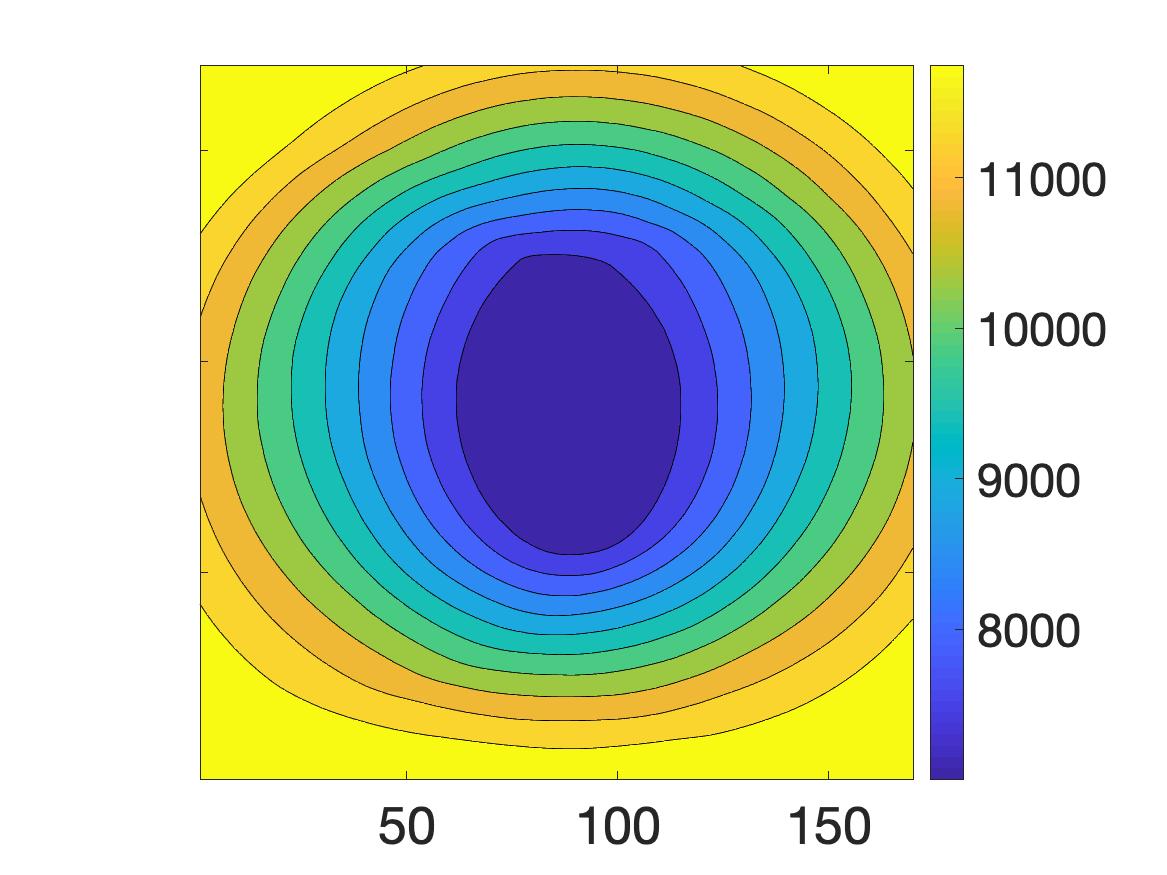}
	\end{subfigure} \\
	\begin{subfigure}[t]{0.3\textwidth}
		\centering
		\includegraphics[width=\textwidth]{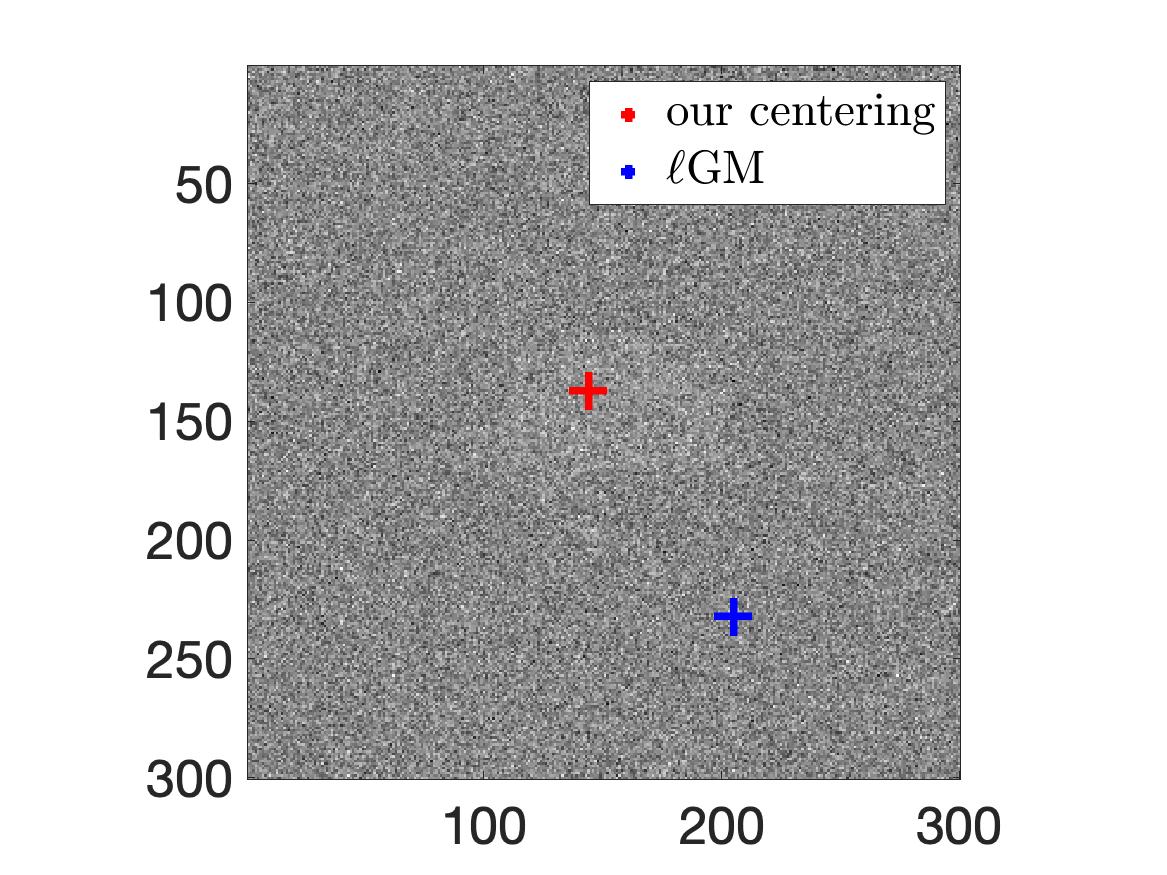}
	\end{subfigure}
	\begin{subfigure}[t]{0.3\textwidth}
		\centering
		\includegraphics[width=\textwidth]{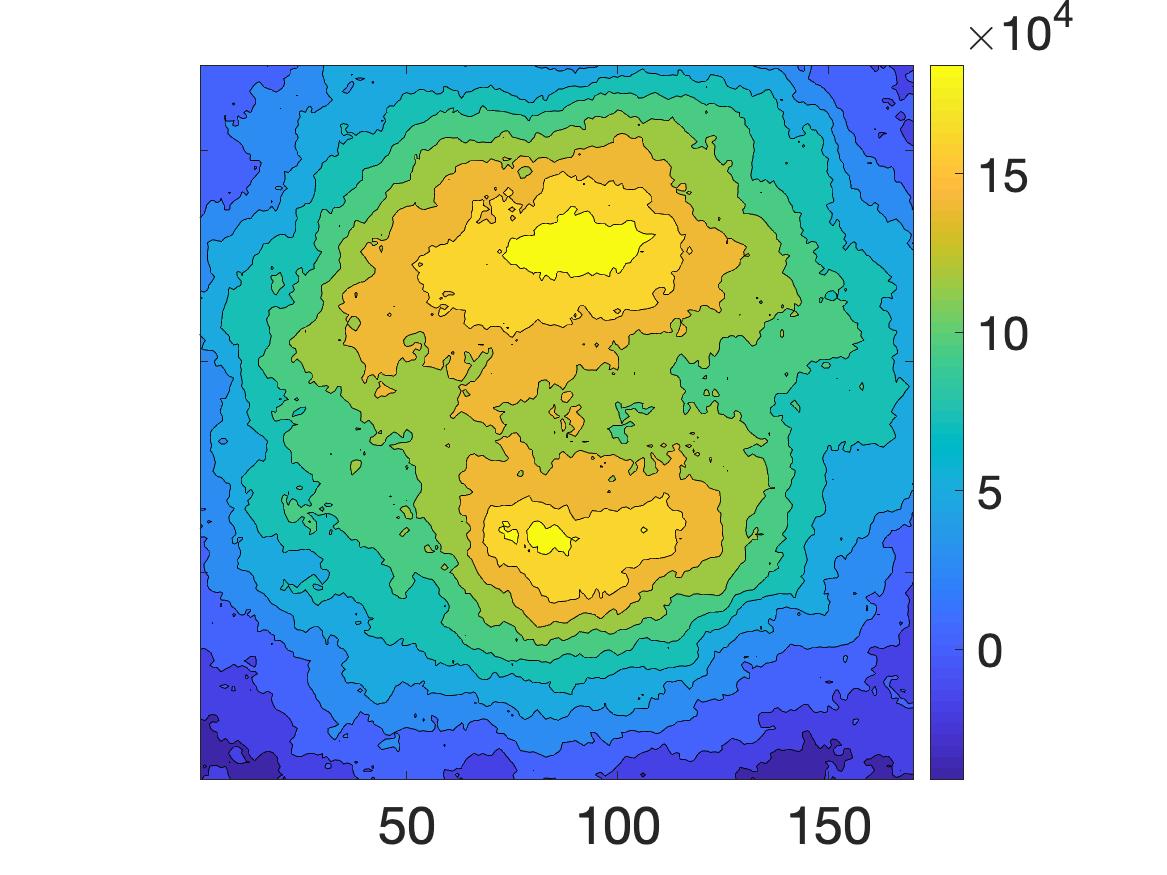}
	\end{subfigure}
	\begin{subfigure}[t]{0.3\textwidth}
		\centering
		\includegraphics[width=\textwidth]{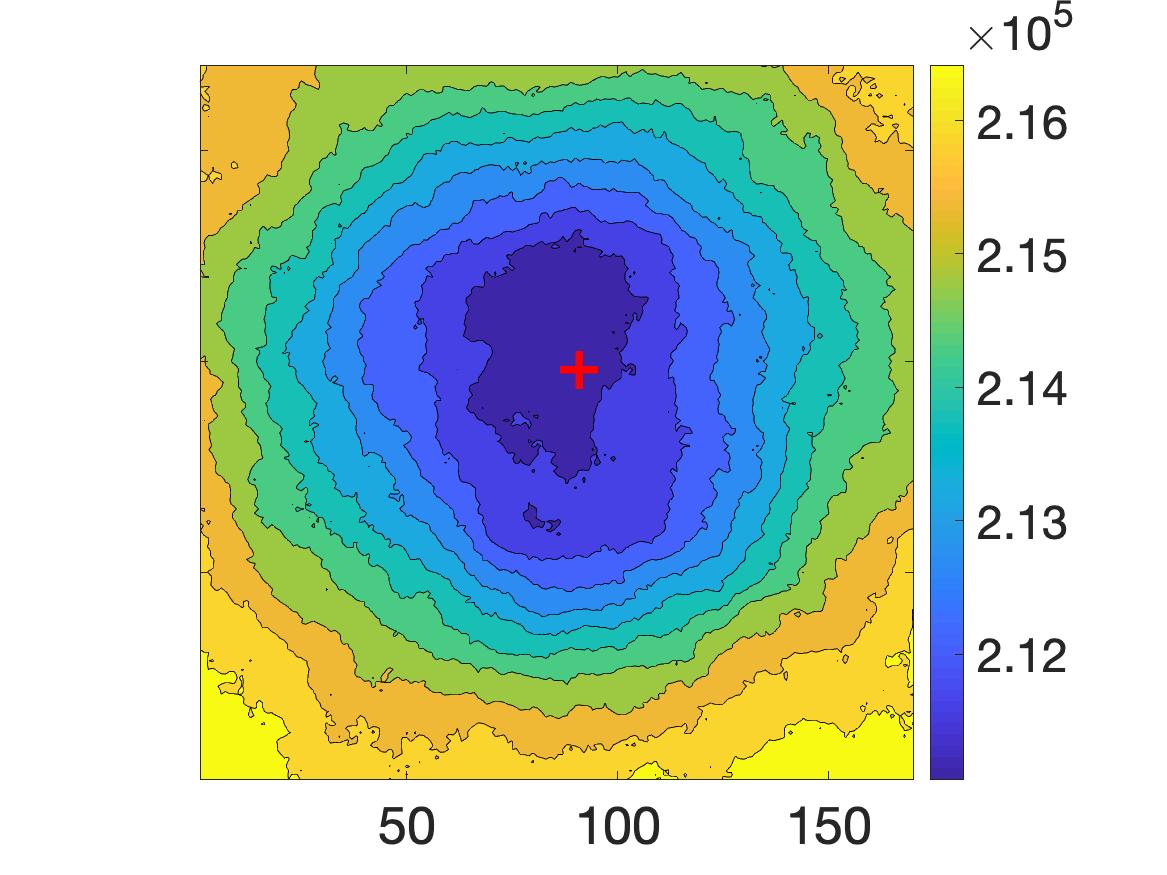}
	\end{subfigure}
	\caption{Two examples of comparison between the landscapes of the local GM, as appears in~\eqref{eqn:gm_landscape_mod} versus the suggested cost function of~\eqref{eqn:surrogate_landscape}. The left upper image includes the clean hedgehog silhouette, and beneath it is a noisy version with SNR $= 1/10$. The maps of the level set in the case of local GM ($\ell$GM) appear in the middle column. While a unique minimizer is available in the noise-free setting, with noise, many minima appear as the landscape varies between areas that are close to the correct GM and regions far away. On the other hand, our centering cost function has a unique minimizer, as seen in the images on the rightmost column in both cases. The minimum in the noisy case is marked by a plus sign (bottom right picture).}
	\label{fig:diff_LS}
\end{figure}

\begin{figure}[!htbp]
	\captionsetup[subfigure]{width=0.85\textwidth}
	\centering
	\begin{subfigure}[t]{0.3\textwidth}
		\centering
		\includegraphics[width=\textwidth]{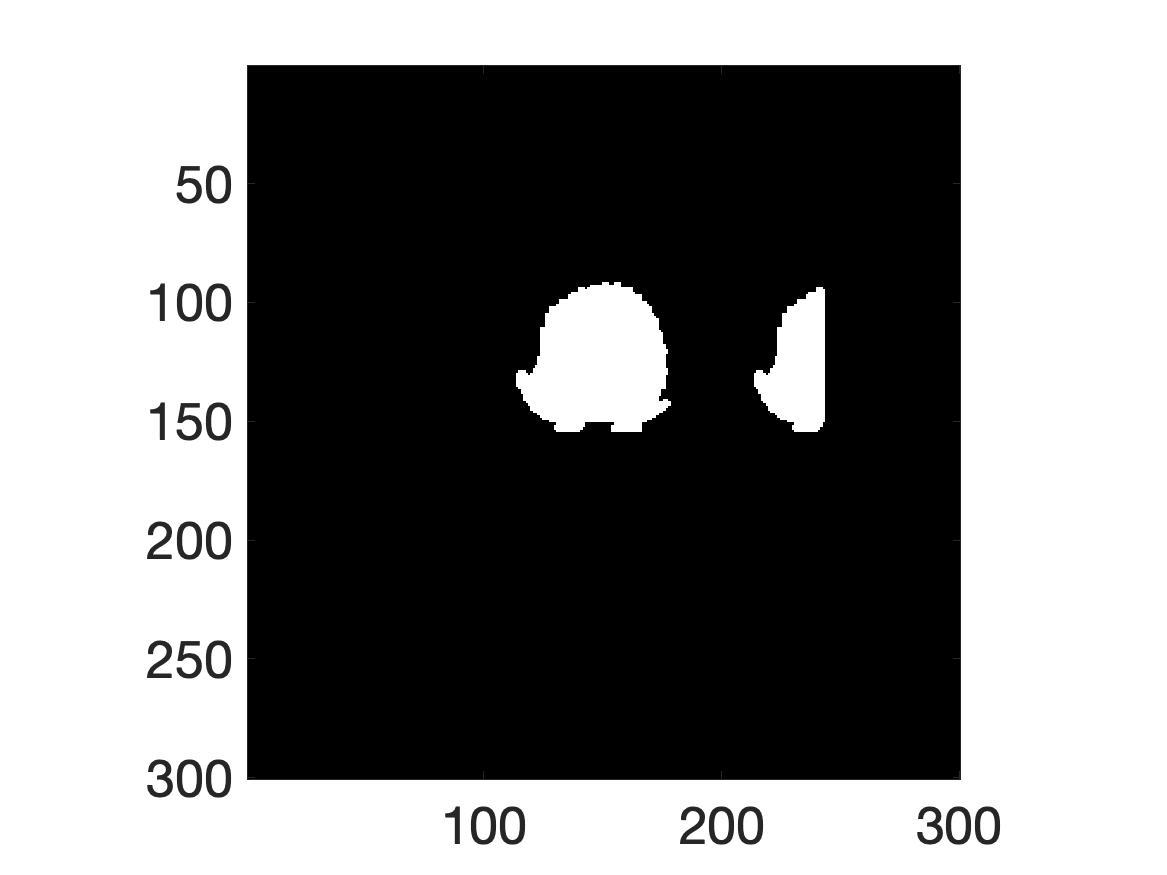}
	\end{subfigure}
	\begin{subfigure}[t]{0.3\textwidth}
		\centering
		\includegraphics[width=\textwidth]{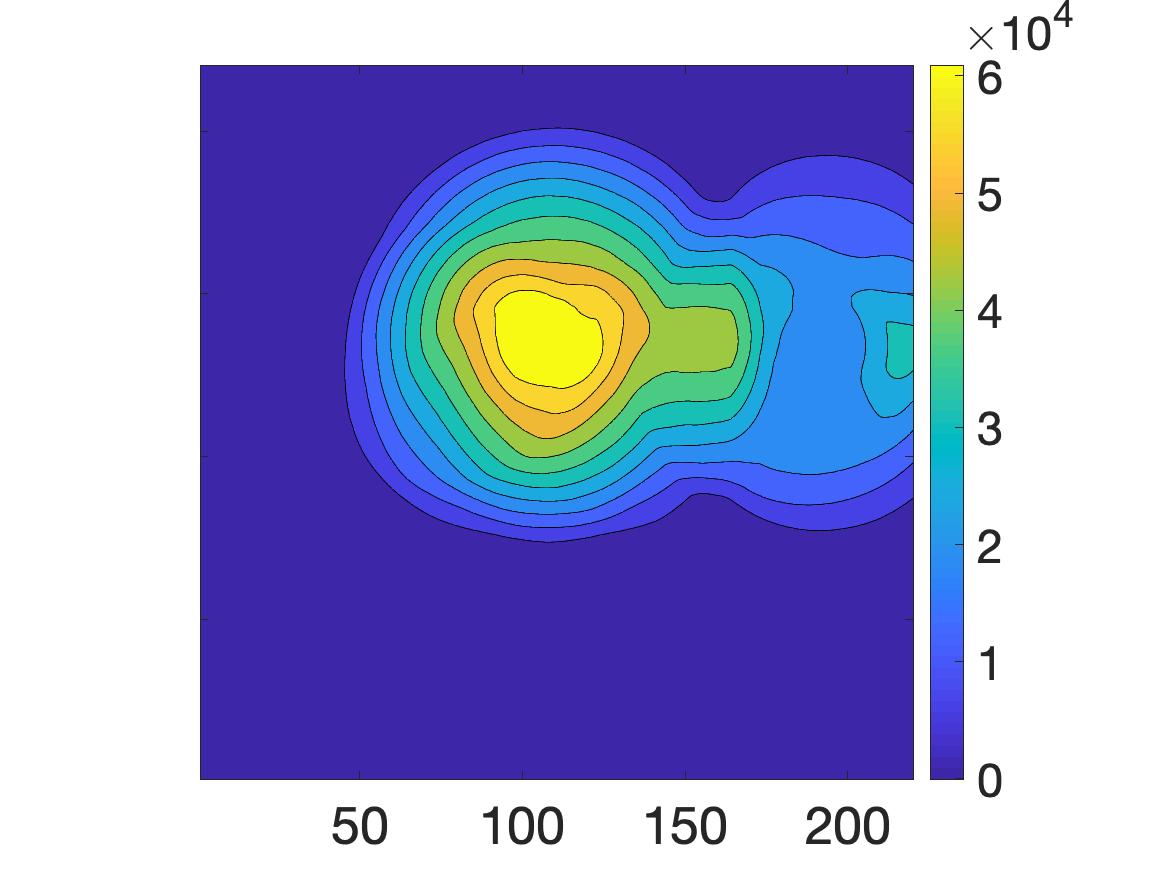}
	\end{subfigure}
	\begin{subfigure}[t]{0.3\textwidth}
		\centering
		\includegraphics[width=\textwidth]{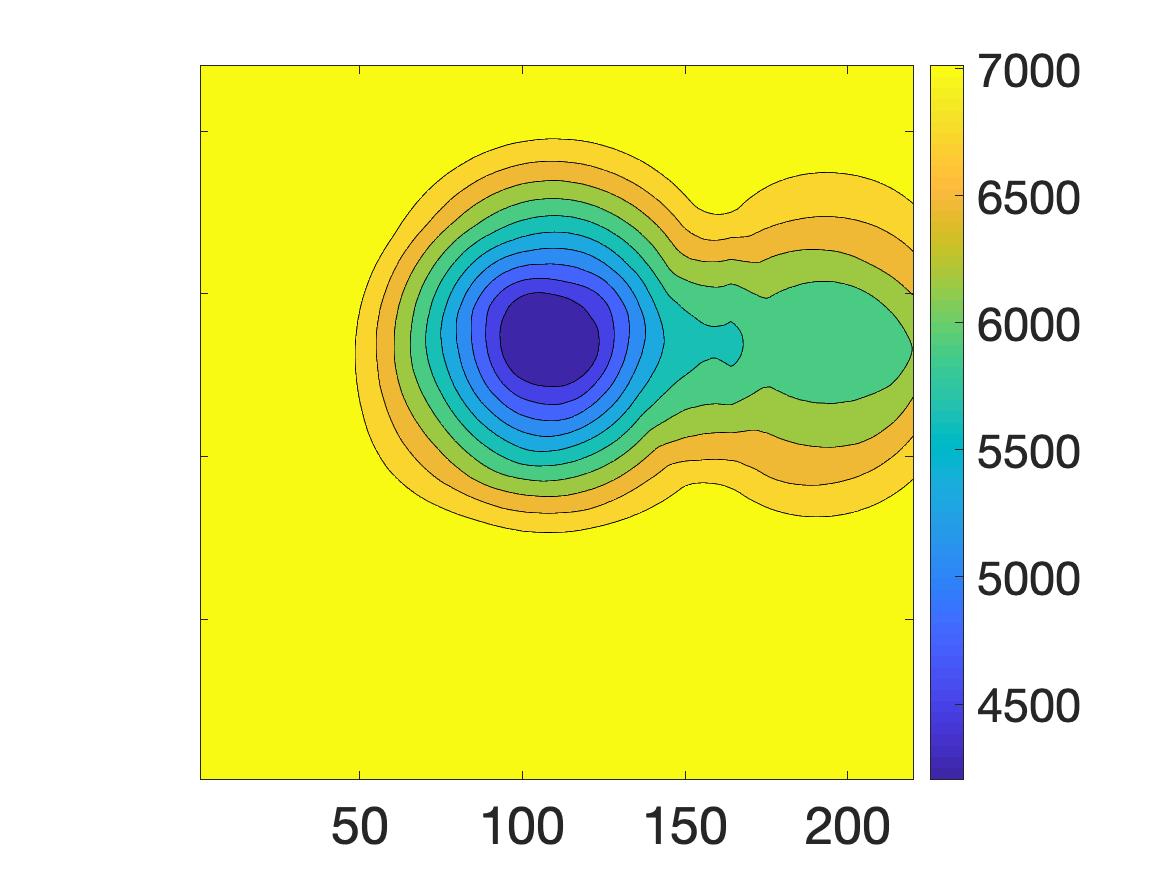}
	\end{subfigure} \\
	\begin{subfigure}[t]{0.3\textwidth}
		\centering
		\includegraphics[width=\textwidth]{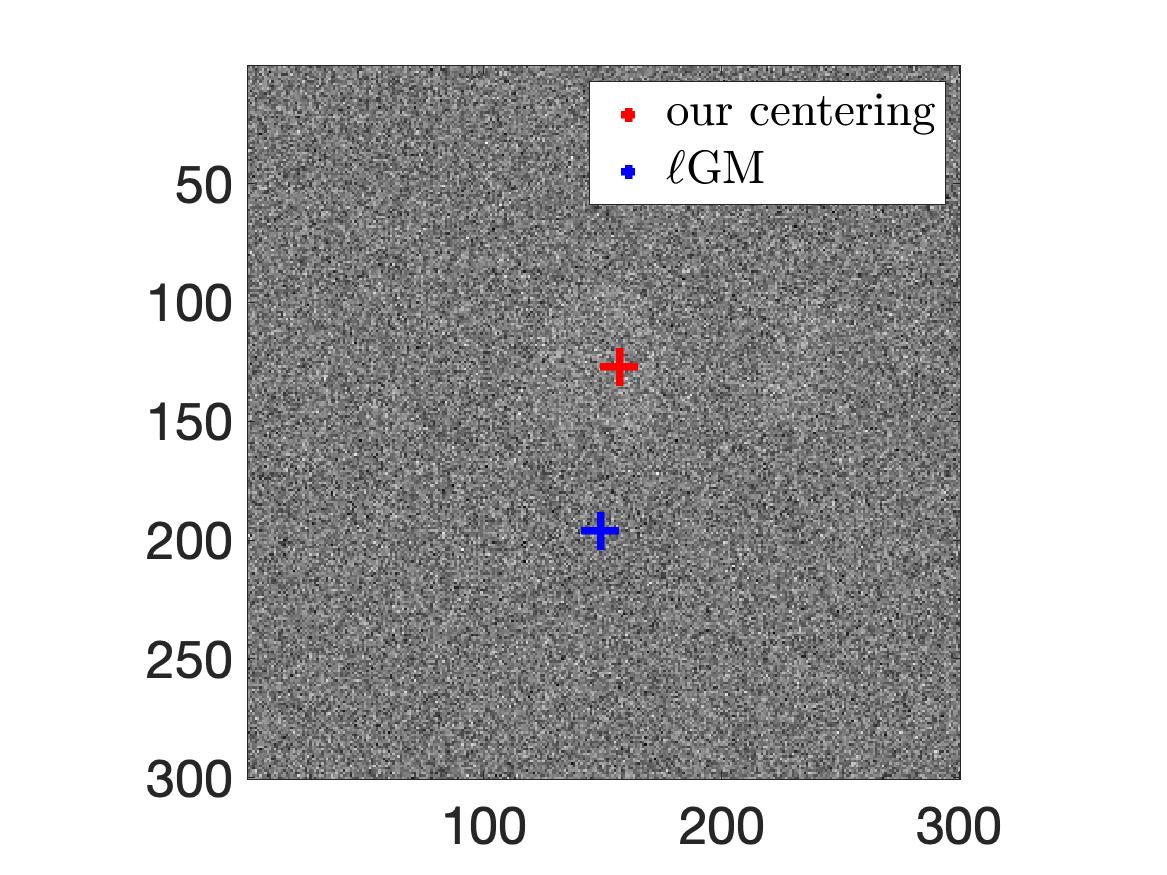}
	\end{subfigure}
	\begin{subfigure}[t]{0.3\textwidth}
		\centering
		\includegraphics[width=\textwidth]{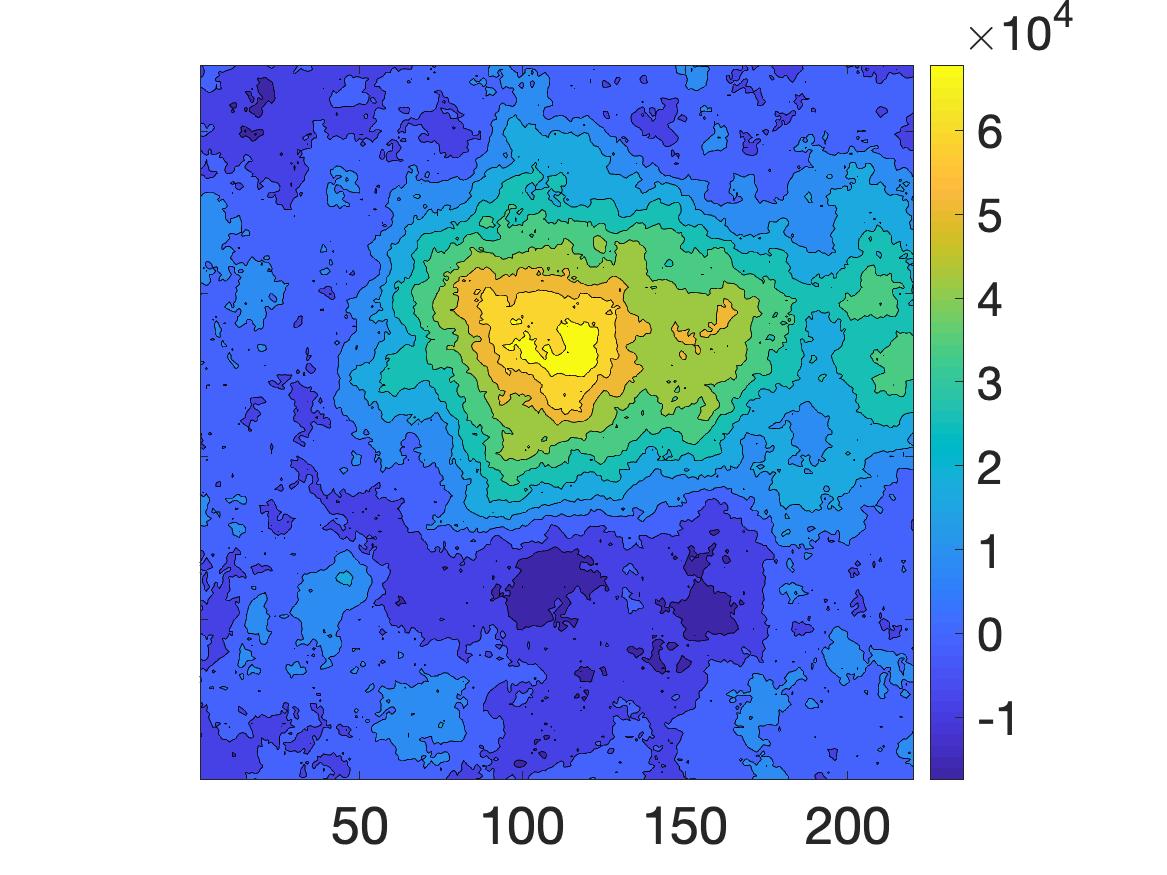}
	\end{subfigure}
	\begin{subfigure}[t]{0.3\textwidth}
		\centering
		\includegraphics[width=\textwidth]{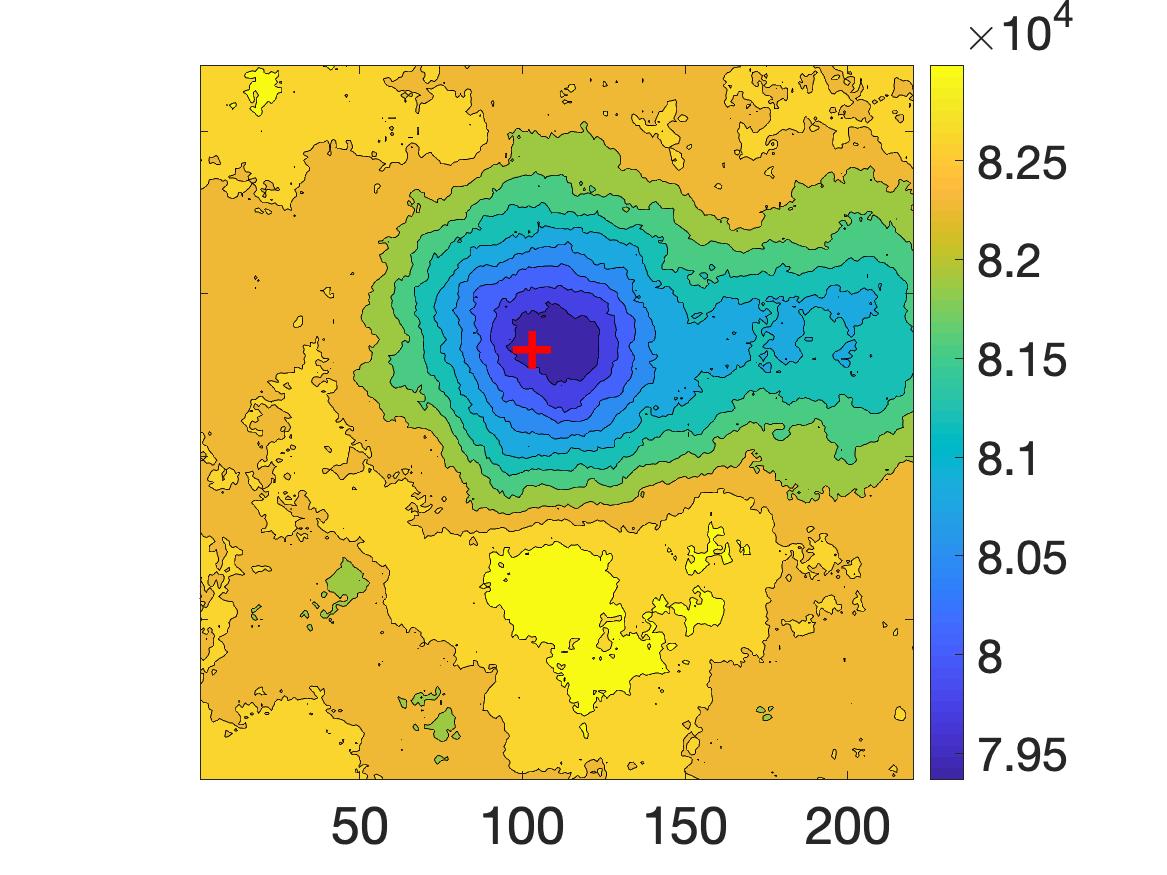}
	\end{subfigure}
	\caption{\rev{Two examples of comparison between the landscapes of the local GM, as appears in~\eqref{eqn:gm_landscape_mod} versus the suggested cost function of~\eqref{eqn:surrogate_landscape}. The left upper image includes a clean hedgehog silhouette and a partial hedgehog silhouette, and beneath it is a noisy version with SNR $= 1/10$. The maps of the level set in the case of local GM appear in the middle column. In the noise-free setting, the global minimum is far from the silhouettes. Additionally,, the area between both silhouettes contains the local minimum. When noise is added, once again many spurious minima appear as the landscape varies between areas that are close to the correct GM and regions far away. On the other hand, our centering cost function has a unique minimizer, as seen in the images on the rightmost column in both cases. The minimum in the noisy case is marked by a plus sign (bottom right picture).}}
	\label{fig:diff_LS2}
\end{figure}

\subsection{Cryo-electron microscopy} \label{subsec:real_data}

Before presenting our numerical results, we provide a short background of the centering problem in cryo-EM, \rev{which we use as a case study in the following}. 

In single-particle cryo-EM, data is acquired by embedding many instances of a particle in vitreous ice and projecting along the axis perpendicular to the ice. Let $\phi \colon \mathbb{R}^3 \to \mathbb{R}$ be the Coulomb potential of the 3-D volume (particle). We define the projection operator $\mathcal{P} \colon \mathbb{R}^3 \rightarrow \mathbb{R}^2$ as
\begin{equation*}
\mathcal{P}\phi(x_1,x_2) := \int_{-\infty}^{\infty} \phi(x_1,x_2,x_3) \, dx_3.
\end{equation*}

Each instance of the particle is captured in some unknown rotation. We denote the rotation applied to the $j$th particle instance as $R_j$. We further denote the rotated particle as $R_j\phi$. The $j$th projection image $I_j$ is formed by tomographically projecting $R_j\phi$. This projection is filtered by the point spread function~\cite{Mindell2003ctffind3, Rohou2015ctffind4, turovnova2017efficient}, which is denoted by $h_j$. Furthermore, each projection image is contaminated by high levels of noise mainly due to an inherent restriction on the number of imaging electrons.

Mathematically, the image formation model of $I_j$ is \cite{bhamre2016formation, frank2006threeB}
\begin{equation} \label{eqn:model_formation}
I_j = h_j * \mathcal{P}\left(R_j\phi \right) + \varepsilon_j , \quad R_j \in \so{3}, \quad  j=1,\ldots,n .  
\end{equation}
Here, $\varepsilon_j$ is a random field modeling the additive noise term and $ \so{3}$ is the 3-D special orthogonal group consisting of orientation-preserving rotations. While the shot noise  follows a Poisson distribution, for discretized projection images of size $N\times N$, it is often modeled as a random field $\varepsilon_j \sim \mathcal{N}(0,\sigma^2 I_{N^2})$, $ j=1,\ldots,n$.  Our algorithm can operate under this and other noise models.

We assume the image has been corrected for $h$ (this is known as CTF correction). Thus, the image formation model is
\begin{equation} \label{eqn:model_images}
I_j =  \mathcal{P}\left(R_j\phi \right) + e_j . 
\end{equation}
One method of such correction is called phase-flipping. This method 
preserves the noise statistics. Therefore, when $\varepsilon_j$ has a Gaussian distribution, so will $e_j$.

So far, we have assumed that all projection images $I_1,\dots, I_n$ are perfectly aligned. In practice, this is not the case.  Rather, many models use an additional translation operator to describe the projections. The challenge of reverting to the shift-free model of~\eqref{eqn:model_images} is what we refer to as ``\textit{the centering problem in cryo-EM}." 

This alignment problem indicates that centering should be applied individually on each image in such a way that it will be possible to register it, up to a global rotation, with the 3-D volume $\phi$. A different perspective is given in the reciprocal space; by the Fourier slice theorem, the Fourier transform of each image equals to a slice in the 3-D Fourier transform of the volume. A translation that acts on the projection image modulates its Fourier transform. Therefore, the Fourier transform of unaligned projection images equals to slices in different modulations of the 3-D reciprocal space of $\phi$. 
Centering the images means modulating their Fourier transform so the resulting projections will, in reciprocal space, form slices that originate from the same 3-D Fourier volume. 

To exemplify the difficulty of the centering problem, we present in Figure~\ref{fig:centering_problem} an  experimental image, called a micrograph, which is a tomographic projection of a section of the ice and the many particles embedded within (each in unknown location and with an unknown rotation). Additionally, we present individual projection images selected out of the micrograph using the APPLE-picker~\cite{heimowitz2018apple}. Visually identifying the center of each projection is a daunting task. Particle pickers, such as the APPLE-picker, will provide many roughly centered projection images as well as other projection images which are poorly centered.

\begin{figure}
	\begin{center}
		{\includegraphics[height=0.25\linewidth]{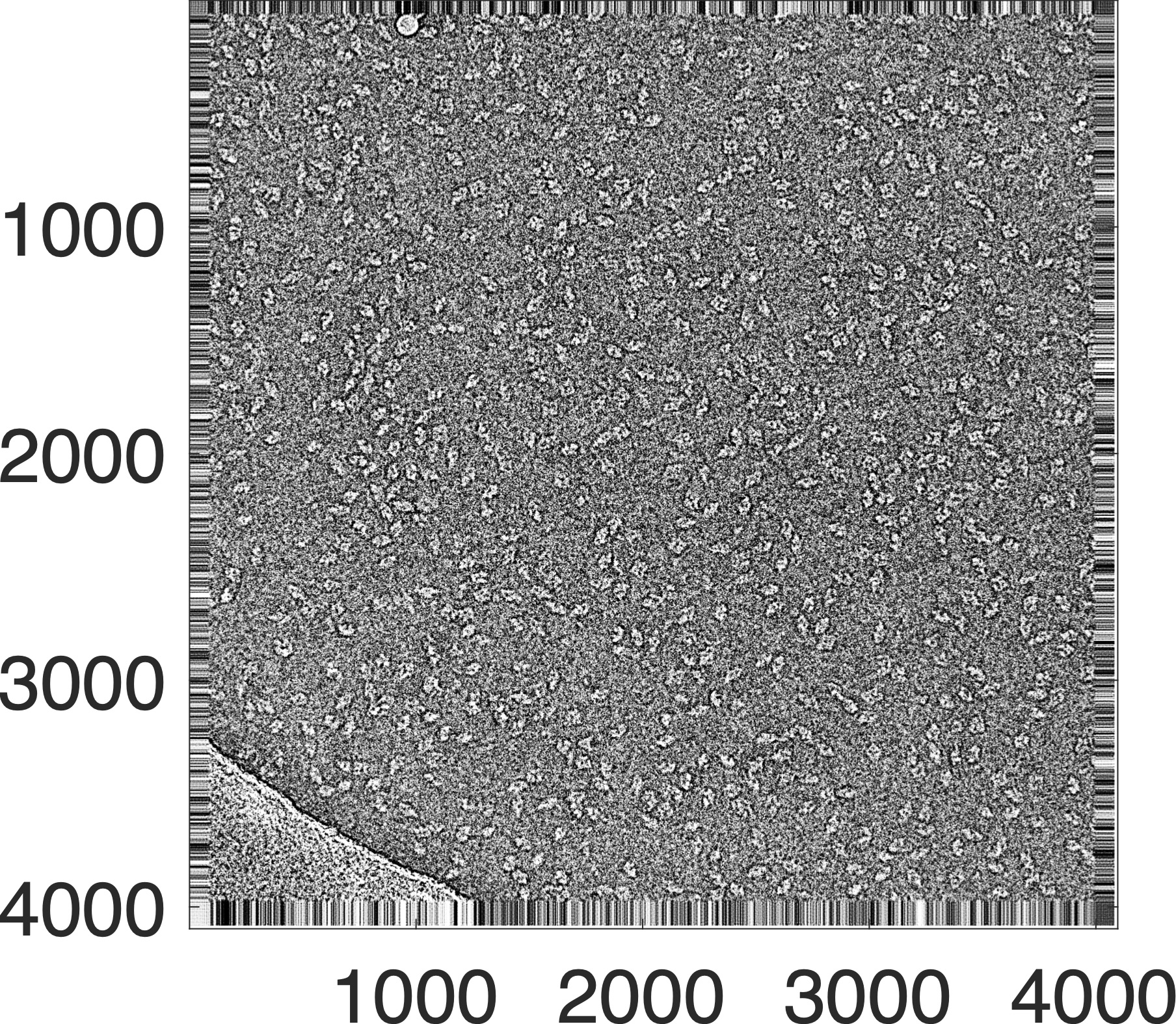}}
		\hspace{1cm}
		{\includegraphics[height=0.25\linewidth]{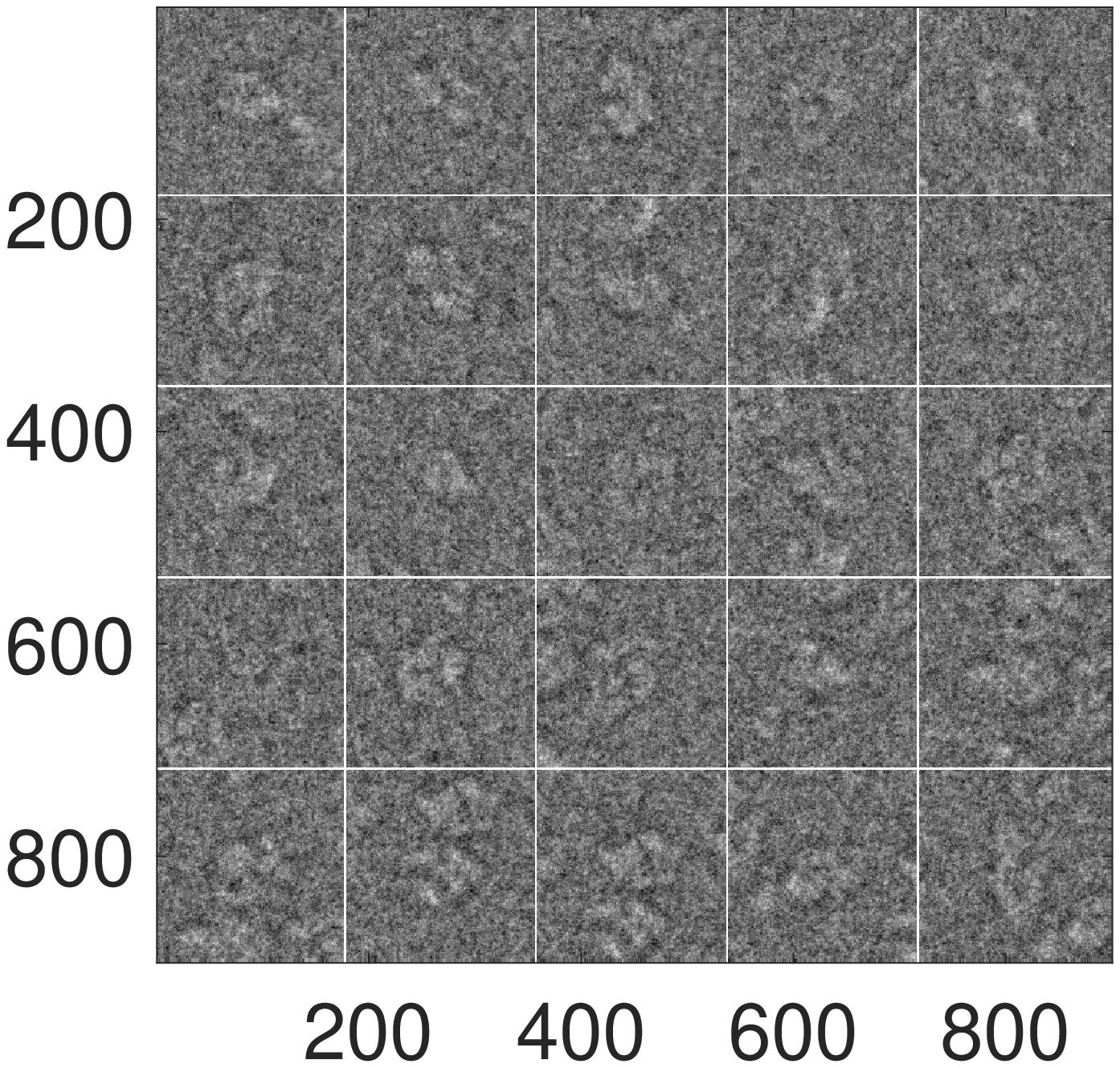}}
	\end{center}
	\caption{Example of experimental images from the EMPIAR-10017~\cite{scheres2015relion} publicly available dataset. On the left we present one micrograph, which is the raw data containing  over $600$ individual projection images. On the right we present individual projection images extracted from the micrograph.}
	\label{fig:centering_problem}
\end{figure}

Currently, centering of projection images is typically done during 2-D classification and 3-D refinement. This is problematic as many reconstruction methods assume, implicitly or explicitly, that the projection images outputted by the particle picker are (all) roughly centered. This assumption is explicitly stated in works on 2-D classification~\cite{sigworth1998maximum} and implicitly manifested in the assumption that the projection exists only within a circle of some radius within the projection images. This assumption is used, e.g., for normalizing noise~\cite{scheres2012relion} or when expanding projection images over a steerable basis with a finite number of radial frequencies~\cite{landa2017approximation}. Another manifestations of the aforementioned assumption appears in the initial modeling stage, as the initial model is created prior to translational alignment. The last manifestation we mention appears during 3-D refinement, as only a limited amount of translations are considered~\cite{scheres2012relion, sigworth2016principles}.  Older modeling methods contain this assumption as well. For example, in 3-D ab-initio modeling method of~\cite{goncharov1988determination}, the statistical moments of a 3-D object are estimated from the moments of its 2-D projection images. Once again, for the moments to be meaningful, the projection images must be centered. 

An example of the output of the 2-D classification stage and the 3-D refinement stage is presented in Figure~\ref{fig:centering_problem2}. On the left are class averages. A class average is an average of many projection images, each with similar viewing direction. While in some of these class averages we do see images of high quality and improved resolution, others are very blurred. The blurred class averages include many outliers and must be discarded. Centering the projection images prior to 2-D classification will reduce the number of blurred class averages.

\begin{figure}
	\begin{center}
		\hspace*{-10pt}
		{\includegraphics[width=0.3\linewidth]{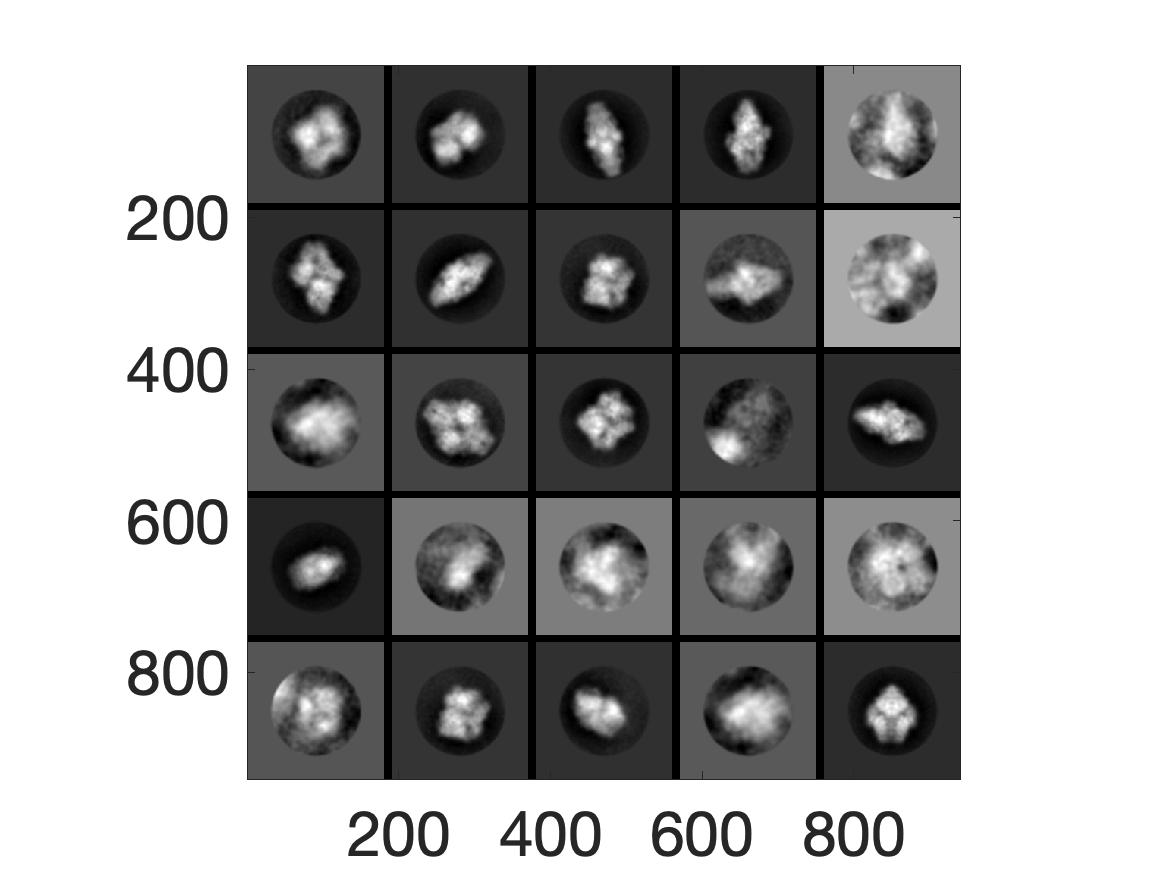}}  
		\hspace*{-20pt}
		{\includegraphics[width=0.4\linewidth]{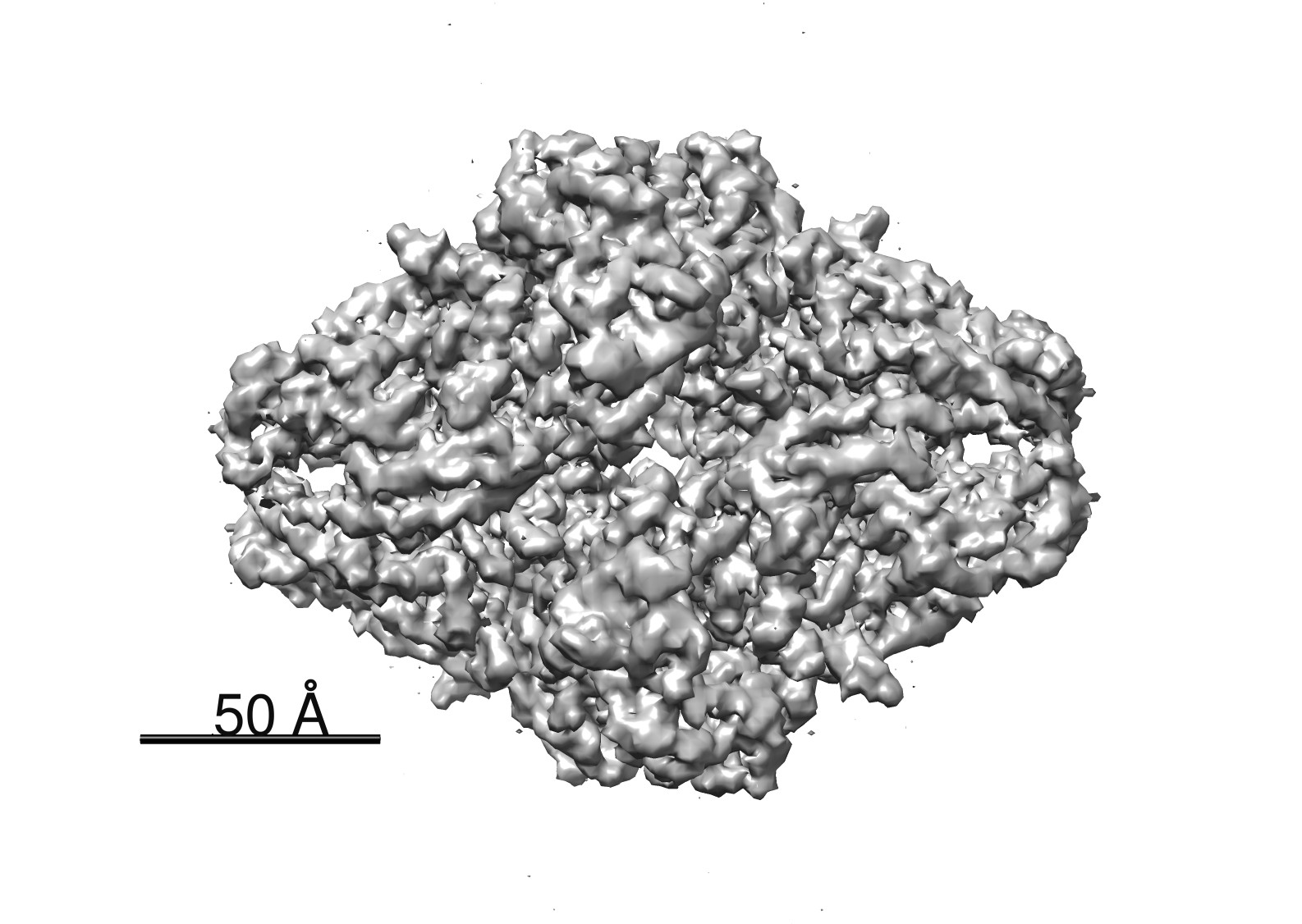}} 
		\hspace*{-35pt}
		{\includegraphics[width=0.4\linewidth]{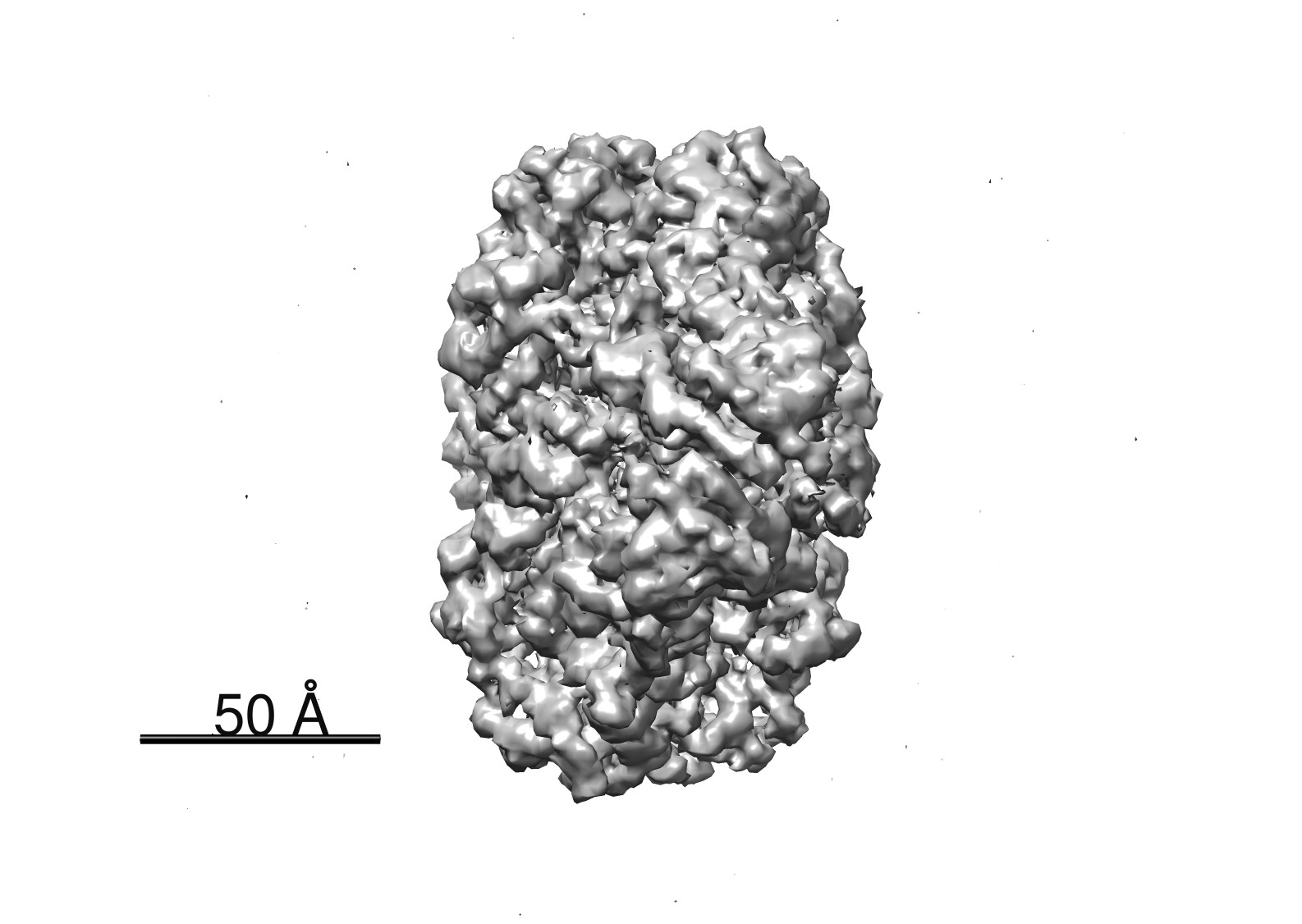}}
		\caption{Results of 2-D classification and 3-D refinement. On the left we present the results of 2-D classification, that is, averages of projections obtained from similar viewing directions. These were computed using RELION~\cite{zivanov2018new}. On the right we present two different views of the 3-D macromolecule resolved during 3-D refinement. The 3-D model used here is publicly available on The Electron Microscopy Data Bank (EMDB)  at the Protein Data Bank in Europe (PDBe) as EMD-2824. The experimental data related to this model is available on EMPIAR as EMPIAR-10017.}
		\label{fig:centering_problem2}
	\end{center}
\end{figure}

One method of accomplishing this was suggested in~\cite{penczek1992three} as the first part of an algorithm known as reference-free alignment (RFA). In this method, when centering the $i$th projection image (out of $n$ images), first the average image is computed by averaging over the remaining $n-1$ projection images,
$$I_{avg} = \sum_{\underset{k \neq i}{k=1,\dots,n}} I_k.$$  
Next, the cross-correlation between $I_{avg}$ and the $i$-th image is calculated. The peak of this function will determine the shift leading to the best match between $I_i$ and $I_{avg}$. This operation is done for each image $i=1,\dots,n$, and repeated until convergence.  
Since projections typically have random in-plane rotations, the average image in the final iterations is approximately circularly symmetric. As a result, the performance of this method is similar to cross-correlation with a Gaussian template. In other words, individual projections are not centered independent of their rotations. Instead, this method aligns each set of images where a similar viewing direction is applied. In contrast, our method is geared towards aligning all projection images, regardless of their rotation. We note that this is an advantage when reconstructing a 3-D volume, as the Fourier transform of aligned projection images are slices in the same 3-D volume.

In the following, we use three publicly available datasets to provide numerical results for a comparison between non-centered projection images, RFA and our centering method.  We show that our method will reduce the necessary translations for 3-D reconstruction of macromolecules, may reduce runtimes and can increase the number of  projections used in the reconstruction and, therefore, may improve resolution.

\subsubsection{Plasmodium falciparum 80S ribosome bound to the anti-protozoan drug emetine}

We test our suggested framework on an 80S ribosome  dataset, which is publicly available as EMPIAR-10028~\cite{wong2014empiar}. This dataset contains $1081$ micrographs, each of size $4096 \times 4096$. Additionally, this dataset includes a CTF-estimation\footnote{That is, an estimation of $h$ in~\eqref{eqn:model_formation}} performed by CTFFIND3~\cite{Mindell2003ctffind3} and a particle picking of more than $100,000$ projection images (determined automatically by RELION~\cite{scheres2015relion}) . 

We correct for CTF according to the parameters computed by CTFFIND3, and run our centering framework on the picked particles supplied as part of the dataset. Next, we perform 3-D refinement in RELION~\cite{zivanov2018new} using 1) the original projection images. 2) centered projection images, where the centering was done via our centering method, and, 3) centered projection images, where the centering was done via RFA~\cite{penczek1992three}. We use as a reference to the 3-D refinement task the model supplied as EMD-2660~\cite{wong2014empiar} in EMDB\footnote{https://www.ebi.ac.uk/pdbe/emdb/}. As RELION outputs the translations performed on each projection image during the reconstruction process, we compare the average translation reported for each reconstruction.  We summarize the results in Table~\ref{tab:10028}. 

\begin{table}[t]
	\begin{center}
		\begin{tabular}{|c|c|c|c|}
			\hline 
			&Original & Our centering & RFA~\cite{penczek1992three}\\
			\hline
			Mean shift & 16 & 10.2 & 12.4\\
			\hline
			Median shift & 14.7 & 9.2 & 11.7\\
			\hline
		\end{tabular}
		\caption{Results for EMPIAR-10028. The center-left column presents mean and median shift for the original centers included in the dataset. The center-right column presents  mean and median shifts for our corrected centers. The rightmost column presents mean and median shift for the RFA method~\cite{penczek1992three}, which is currently the prevalent method of alignment in cryo-EM. The baseline centers are taken from the output of RELION 3-D iterative refinement reconstruction. Shifts are given in pixels, where the pixel spacing is $1.34$ \r{A}.}
		\label{tab:10028}
	\end{center}
\end{table}

In the above experiments, we ran only 3-D refinement, and all reconstructions were created using the same parameters. The resolution of the reconstruction ($3.8$ \r{A}) was unchanged by  centering.  Unlike resolution,  runtimes, reported  in Table~\ref{tab:10028_runtime}, were influenced. Our centering leads to a $50$ minute speedup in the 3-D refinement task which runs on $4$ GPUs (and $10$ CPU cores). While the centering itself ran for an hour and nine minutes, we note that, as it ran on a single GPU due to platform constraints (and $12$ CPU cores), these runtimes are  not comparable.

\begin{table}
	\begin{center}
		\begin{tabular}{ |l | c | c | c|}
			\hline
			RELION task & run-time &  run-time  & run-time\\ 
			& (original projections) & (centered projections) &  (RFA~\cite{penczek1992three})\\ 
			\hline
			3-D refinement & 6:51:12   & 6:00:20  &  06:40:51  \\ \hline
		\end{tabular}
	\end{center}
	\caption{Runtimes for centering and for RELION's 3-D refinement task on the EMPIAR-10028 dataset. We note that all 3-D refinement tasks were initialized with  the same parameters, and utilized four nVidia P100  GPUs, $5$ MPI processes and $10$ cores.}
	\label{tab:10028_runtime}
\end{table}

\subsubsection{ \texorpdfstring{$\beta$}{Beta}-galactosidase}

We test our suggested framework on a $\beta$-Galactosidase
dataset, which is publicly available as EMPIAR-10017~\cite{scheres2015relion}.  
This dataset contains $84$ micrographs, each of size $4096 \times 4096$. A typical micrograph of this dataset includes a few hundred projections.

Our computational setup was as follows; 
we ran the APPLE picker~\cite{heimowitz2018apple} on the raw micrographs, and  performed 
particle picking as detailed in Section 3.1 of \cite{heimowitz2018apple}. This led to $42664$ picked particles.  
After this initial picking, we ran our suggested framework on phase-flipped (\textit{i.e.} CTF-corrected) projection images selected by the APPLE picker. 
The parameters of the CTF were estimated using CTFFIND4 \cite{Rohou2015ctffind4}.  

We use RELION to compare three reconstructions. The first reconstruction is made from the original projection images supplied by the APPLE-picker. The second and third reconstructions are made from the centered projection images, where the centering was done via our centering method and RFA, respectively. We present a comparison of the runtime of each stage in the RELION pipeline in Table~\ref{tab:betagal_runtime}. We see that using our centering method will cause a speedup of $10\%$ in 2-D classification and allow the final reconstruction to be made from an additional $4000$ projection images. These additional images do slow down the reconstruction process. However, in general, we prefer to use as many projections as possible to resolve the 3-D structure of macromolecules. Furthermore, we see from Table~\ref{tab:betagal_runtime} that a reconstruction produced from projection images centered via RFA achieves \rev{comparable} resolution with fewer projection images as compared to the 
original (non-centered) projection images. This indicates that reconstruction from centered projection images improves the success rate of 3-D classification.

For a complete picture of runtimes, we note that the runtime of our suggested method is five minutes when running on a single nVidia P100  GPU and 12 CPU cores.

In Table \ref{tab:betagal_shift} we present the average shifts RELION applies to the set of projection images during reconstruction. We note that the average shift is significantly smaller after our centering is applied. 

\begin{table}
	\begin{center}
		\begin{tabular}{ |l | c | c | c | }
			\hline
			RELION task & run-time &  run-time  & run-time\\ 
			& (original projections) & (centered projections) &  (RFA~\cite{penczek1992three})\\ 
			\hline
			2-D classification   & 00:42:12 & 00:38:08 & 00:44:23  \\ \hline
			\# remaining projections & 32594 &  37693 & 37698  \\ \hline
			Initial modeling     & 0:40:05 & 00:40:37   & 00:41:04  \\ \hline
			3-D  classification  & 00:25:42 & 00:28:16 & 00:29:47  \\ \hline
			remaining projections & 24148 &  28861 & 22321 \\ \hline
			3-D refinement       & 00:17:34  & 00:18:17 & 00:16:58  \\ \hline
			Overall                         & 02:05:33  & 02:05:18 & 02:12:12  \\ \hline
			Resolution & 4.3 \r{A} & 4.2\r{A} & 4.2\r{A}\\  \hline
		\end{tabular}
	\end{center}
	\caption{Run times for RELION's 3-D reconstruction on EMPIAR-10017. We present runtimes for the full reconstruction pipeline when run on original projections selected by the APPLE picker~\cite{heimowitz2018apple} (center-left column), centered projections when using our suggested method (center-right column) and centered projections when using RFA~\cite{penczek1992three} (rightmost column). All runtimes refer to 4 nVidia P100  GPUs and multiple cores (between $10$ and $112$ depending on the task).}
	\label{tab:betagal_runtime}
\end{table}

\begin{table}
	\begin{center}
		\begin{tabular}{ |l | c | c | c |}
			\hline
			& APPLE picking &  Our centering & RFA~\cite{penczek1992three}  \\ \hline
			Mean shift & 12.3 & 8.1  & 10.7  \\ \hline
			Median shift & 10.6 & 7.9 & 8.9 \\ \hline
		\end{tabular}
	\end{center}
	\caption{Average shifts during 3-D refinement as reported by RELION on EMPIAR-10017. Shifts are given in pixels, where the pixel spacing is $1.77$ \r{A}.}
	\label{tab:betagal_shift}
\end{table}

\subsubsection{TRPV1}

Lastly, we verify our centering method on the dataset EMPIAR-10005~\cite{liao2013empiar}. This dataset contains $80443$ picked particles, as well as a subset of $32387$ particles that were used in the reconstruction uploaded to EMDB. We used RELION to reconstruct a 3-D volume using the entire subset of $32387$ picked particles. Additionally, we ran our alignment algorithm on all the picked particles, and used RELION to reconstruct a 3-D volume out of all particles that were not discarded during reconstruction. The same was done after centering with RFA. We summarize the results of this experiment in Table~\ref{tab:10005}. Our centering allowed us to use $43701$ particles with no adverse effect on resolution. {We do not compare the average shift as the reconstructions are all done using a vastly different number of projection images.}

As shown in Table~\ref{tab:10005}, RFA reduces the number of projection images that survive 2-D classification. We note that approximately $10000$ projections  are discarded because their class averages  contain a partial 2-D view of the macromolecule. We believe this is caused by nearby projections. That is, it is possible that projection images that contain a projection and a partial projection will produce unexpected results in the cross-correlation employed by RFA, causing large shifts and reducing the number of usable projection images. Indeed, the class averages that contained partial 2-D views of the macromolecule all included blurred areas, indicating that nearby projections were averaged in.

\begin{table}
	\begin{center}
		\begin{tabular}{|c|c|c|c|c|}
			\hline 
			RELION task & run-time &  run-time  & run-time\\ 
			& (original projections) & (centered projections) &  (centered projections~\cite{penczek1992three})\\ 
			\hline
			2-D Classification 
			runtime & 01:29:50 & 01:24:31 & 01:50:15 \\
			\hline
			Number of particles & 32,387 & 43,701 & 23,535\\
			\hline
			3-D Refinement  
			runtime & 00:51:40 & 01:23:13 & 00:40:40\\ \hline
			Resolution & 4.45 \r{A} & 4.32 \r{A} & 4.72 \r{A}\\
			\hline
		\end{tabular}
		\caption{Results for EMPIAR-10005. After centering, the reconstruction can be done from an extra 11000 particles. Without our centering method, these projections would be dismissed during 2-D classification resulting in slightly reduced resolution. The 2-D Classification task utilized $4$ nVidia P100  GPUs and $112$ CPU cores. The 3-D Refinement task utilized 4 nVidia P100  GPUs, $5$ MPI processes and $10$ CPU cores. Our centering method utilized a single  nVidia P100  GPU and $12$ CPU cores and RFA utilized a single  nVidia P100  GPU and single CPU core.}
		\label{tab:10005}
	\end{center}
\end{table}

\section{Conclusion} \label{sec:conclusion}

The main contribution in this paper is the introduction of a surrogate function to the center of mass, which leads to a robust estimator of the center of mass. We mathematically motivate our construction, along with examples for which the geometric median fails to provide an accurate estimate of the center of mass, yet our surrogate function prevails. We further analyze the mathematical base behind our surrogate function, including the connection with the geometric median and the earth mover's distance. In the computational part, we present a comparison with state-of-the-art methods, which highlights the advantage of our approach for extremely noisy images. Lastly, we show that, when applied to experimental cryo-EM data, our centering algorithm enables faster 3D recovery and better exploitation of valuable computational resources.

\section*{Acknowledgement}
The authors thank B. Landa for the PSWF code and discussions. The authors are also indebted to T. Bendory, J. Kileel, W. Leeb, and Y. Shkolnisky for helpful comments and discussions. 

A.S. and A.H. were partially supported by NSF BIGDATA award IIS-1837992, awards FA9550-20-1-0266 and FA9550-17-1-0291 from AFOSR, NIH/NIGMS award 1R01GM136780-01, the Simons Foundation Math+X Investigator Award, and the Moore Foundation Data-Driven Discovery Investigator Award. N.S. and A.S. were partially supported by BSF grant no. 2019752, and NSF grant DMS-2009753. N.S. was partially supported by BSF grant no. 2018230.

This research was conducted while A.H was a postdoctoral research associate at the Program for Applied and Computational Mathematics in Princeton University.

\bibliographystyle{plain}
\bibliography{RACER_bib.bib}


\end{document}